\documentclass{extarticle}
\usepackage{graphicx} 
\usepackage{xcolor}
\usepackage{longtable}
\usepackage{amsmath} 
\usepackage{amsthm} 
\usepackage{amssymb} 
\usepackage{float}
\usepackage{tikz-cd}
\usepackage[margin=1.2in]{geometry}
\usepackage{mathtools}
\usepackage[natbib=true, style=numeric,sorting=none]{biblatex}
\usepackage{hyperref}
\usepackage{authblk}
\usepackage{arydshln}
\usepackage{enumitem}
\usepackage{caption}

\newtheorem{prop}{Proposition}
\newtheorem{thm}{Theorem}[section]
\newtheorem{lem}[thm]{Lemma}

\theoremstyle{definition}
\newtheorem{dfn}{Definition}
\theoremstyle{definition}
\newtheorem{rmk}{Remark}
\newtheorem{ex}{Example}
\newtheorem{conj}{Conjecture}
\newcommand{\bbZ}{\mathbb{Z}}

\newcommand{\bbF}{\mathbb{F}}
\newcommand{\Mod}{\text{Mod}}
\DeclareMathOperator{\Ima}{im}
\title{Sequences of Bivariate Bicycle Codes from Covering Graphs}
\author[1,3]{Benjamin C. B. Symons\footnote{bsymons@ed.ac.uk}}
\author[2]{Abhishek Rajput\footnote{abhishek.rajput@maths.ox.ac.uk}}
\author[3]{Dan E. Browne}
\affil[1]{School of Informatics, University of Edinburgh, United Kingdom}
\affil[2]{Mathematical Institute, Oxford University, Oxford, OX2 6GG, United Kingdom}
\affil[3]{Department of Physics \& Astronomy, University College London, London, WC1E 6BT, United Kingdom}
\date{\today}

\begin{document}

\maketitle

    \begin{abstract}
	We show that given an instance of a bivariate bicycle (BB) code, it is possible to generate an infinite sequence of new BB codes using increasingly large covering graphs of the original code's Tanner graph. When a BB code has a Tanner graph that is a $h$-fold covering of the base BB code's Tanner graph, we refer to it as a $h$-\textit{cover code}. We show that for a BB code to be a $h$-cover code, its lattice parameters and defining polynomials must satisfy simple algebraic conditions relative to those of the base code. By extending the graph covering map to a chain map, we show there are induced projection and lifting maps on (co)homology that enable the projection and lifting of logical operators and, in certain cases, automorphisms between the base and the cover code. The search space of cover codes is considerably reduced compared to the full space of possible polynomials and we find that many interesting examples of BB codes, such as the $[[144,12,12]]$ gross code, can be viewed as cover codes. We also apply our method to search for BB codes with weight 8 checks and find many new codes, including a $[[64,14,8]]$ and $[[144,14,14]]$ code. For an $h$-cover code of an $[[n,k,d]]$ BB code with parameters $[[n_h = hn, k_h, d_h]]$, we prove that $k_h \geq k$ for any $h$, and $d_h \leq hd$ when $h$ is odd. Furthermore if $h$ is odd and $k_h = k$, we prove the distance lower bound $d \leq d_h$. We conjecture it is always true that an $h$-cover BB code of a base $[[n,k,d]]$ BB code obeys the distance bounds $d \leq d_h \leq hd$. While the focus of this work is on bivariate bicycle codes, we expect these methods to generalise readily to many group algebra codes and to certain code constructions involving hypergraph, lifted, and balanced products.
	\end{abstract}
	\newpage
	\tableofcontents
	\newpage

\section{Introduction}

In recent years, quantum low-density parity check (qLDPC) codes have emerged as a promising alternative to surface codes for quantum error correction \cite{Breuckmann2021LDPCrev}. While the surface code is relatively easy to implement with its low-weight checks that can be arranged on a 2D planar architecture, it suffers from a vanishing rate and $\sqrt{n}$ distance scaling. 

There are now many examples of families of qLDPC codes with constant rate, distance scaling at least that of the surface code \cite{Breuckmann2021Balanced,Hastings2021Fiber,Leverrier2022Tanner}, and even asymptotically `good' codes \cite{Panteleev2022GoodQuantum,Dinur2023GoodLinear}. The benefit of many of these codes is that they have the potential to significantly reduce the space overhead of quantum error correction relative to the surface code, which could help usher in the era of fault tolerant quantum computation sooner. The price paid for this reduced space overhead however is the introduction of long range connections and checks that are often higher weight than those of the surface code. While both of these factors make the implementation of qLDPC codes on actual quantum hardware harder, particularly for those architectures with limited connectivity, this is a challenge that is actively being addressed by hardware teams \cite{Paetznick2024demoLogical,Reichardt2024ftqcNeutral,Wang2025BBdemonstration}.

Despite the limitations of the surface code, methods for performing logic on it and surface code architectures for performing universal quantum computation are well developed \cite{Horsman2012surgery,Litinski2019gameofSC,Bartolucci2023Fusion}. While the same cannot be said of qLDPC codes, there has been considerable recent progress in this area. For example, there are recent developments in generalised lattice surgery addressing the problem of joint measurements of logical operators within and between arbitrary qLDPC codes \cite{Cowtan2024surgeryUniversal,Cowtan2024ssip,Cross2024qldpcsurgery} with space overheads reduced from $O(d^2)$ \cite{Cohen2022def} to $O(d\text{log}^2(d))$ \cite{Williamson2024gaugingLogicals,Swaroop2025adapters,Ide2025homMeasure}. Progress has also been made on addressing the problem of parallel measurements \cite{Li2025timeEfficient,Cowtan2025parallellmeasure}. Alongside lattice surgery, there have been advances in using automorphisms to perform logic \cite{Sayginel2025auts,Malcolm2025computingEfficiently,Berthusen2025autGadgets}. There are now proposals for architectures that either combine a qLDPC memory with surface code computation blocks \cite{Xu204reconfigAtom,Stein2024hetArch} or use qLDPC codes directly for computation \cite{He2025extractorsArch,Yoder2025TdG}. Recent advances in transversal non-clifford gates on qLDPC codes \cite{He2025asymptoticallyNonCliff,Golowich2025TransversalNonCliff,Xu2025Fast,Jacob2025trivariate,Menon2025magicTricycles} also offer the potential for a purely qLDPC based architecture that can perform universal quantum computation. 

Bivariate bicycle (BB) codes were first proposed by Bravyi et. al. \cite{Bravyi2024BB} and are particularly interesting because there exist relatively small BB code instances with high rate and reasonably high distances, e.g. the so-called gross code with parameters $[[144,12,12]]$. The BB codes have weight 6 checks, with each qubit requiring 4 local and 2 non-local connections. Since their inception, there has been a considerable body of work studying the bivariate bicycle codes \cite{Eberhardt2025foldBB,Shaw2025Morphing,Wang2025coprimeBB,Wang2024RateBB,Voss2025Multivar,Eberhardt2024Pruning}, including a recent proposal for a BB code based architecture \cite{Yoder2025TdG}. Despite the interest in BB codes however, there has been relatively little progress made in establishing a \textit{general} family of BB codes. Recent work from Liang et. al. was able to generate a subset of the BB codes by considering an algebraic generalisation of the toric code \cite{Liang2025gtc} and further work by Chen et. al. was able to establish a topological classification of all such BB codes \cite{ChenAnyonToricLayout}. Although these constructions form a particular subset of BB codes, they contain many, if not all, of the known interesting examples of BB codes. 

In this work, we explore connections between instances of BB codes via the language of covering graphs, specifically those generated by graph lifts (derived graphs). We find that given a BB code $Q$ with Tanner graph $T(Q)$, it is possible to generate an infinite sequence of new BB codes $\tilde{Q}$ corresponding to covering graphs of $T(Q)$. We refer to these codes as $h$\textit{-cover codes} when their Tanner graph $T(\tilde{Q})$ is a $h$-fold covering graph of $T(Q)$. 

We derive purely algebraic conditions on the defining parameters and polynomials of the BB codes that guarantee a Tanner graph is a cover of a given base BB code's Tanner graph. Practically speaking, this allows a search over a considerably reduced polynomial space to generate new BB codes that appear to have interesting parameters. We provide both numerical evidence and theoretical guarantees to support the claim that the space of cover codes contains interesting examples of BB codes. In particular, our sequences of cover codes include all of the codes in the original BB codes paper \cite{Bravyi2024BB}; the $[[144,12,12]]$ gross code can be generated as a double cover of the $[[72,12,6]]$ code, for example\footnote{Shortly prior to publication, the authors were made of aware of a talk by V. Guemard that stated the gross code is a double cover of the [[72,12,6]] code \cite{Guemard2025QIPtalk}.}. We are able to reproduce a sequence of $[[18q,8,\leq2q]]$ balanced product codes from \cite{Tiew2025Dehn} and in principle all of the generalised toric codes \cite{Liang2025gtc} as sequences of cover codes. We also apply our method to search for BB codes with weight 8 checks and find sequences of codes with $k=8,10,12,14$. These codes haver higher $kd^2/n$ than the weight 6 codes and include examples such as $[[24,10,4]]$, $[[64,14,8]]$, and $[[144,14,14]]$ codes that have $kd^2/n$ of 6.7, 14 and 19.1 respectively.

Under certain conditions, we are able to prove bounds on the number of logical qubits and distance of a cover code. Given a base BB code with parameters $[[n,k,d]]$, we prove that an $h$-cover code has parameters $[[n_h = hn, k_h, d_h]]$ obeying $k_h \geq k$ for any $h$, and $d_h \leq  hd$ when $h$ is odd. When $h$ is odd and $k_h=k$, we can lower bound the distance as $d_h \geq d$. For even $h$, we empirically observe the same distance bounds and therefore conjecture that any $h$-cover code obeys the distance bounds $d \leq d_h \leq hd$.

Besides generating new instances of BB codes, the covering graph relation lets us use the graph covering map to induce a map between the (co)homologies of the corresponding BB code chain complexes that can be used to project logical operators from the cover to the base code. We also define a lifting map that similarly induces a map on (co)homology and lifts logical operators from the base to the cover code. This serves as a useful tool for finding logical operators of large codes given a logical basis for a smaller, base code. We define a weight-preserving lift and show how it can sometimes be used to upper bound the distance of a cover code as $d_h \leq d$. This is useful when searching for new codes to rule out instances where the distance does \textit{not} increase relative to the base code. Finally, we show that automorphisms also lift and that, under certain conditions, the logical action of the base and lifted automorphisms is the same.

We also show that our methods extend to qudit BB codes. In particular, we are able to search for new qudit BB codes given an initial base qudit BB code in the same way as for qubit BB codes. Our other results on the behavior of $k_h$ and $d_h$ likewise extend naturally and in fact, we find that our distance inequalities hold for more values of $h$. More precisely, when the field is $\mathbb{F}_q$ with $q = p^n$ for an integer $n$ and prime $p$, we find that $k_h \geq k$ for any $h$. We also show that $d_h \leq hd$ if $h \neq 0\ \text{mod}\ p$, and $d \leq d_h$ if $k_h = k$ also holds.

A paper by V. Guemard \cite{Guemard2025lifts} also uses coverings but in the context of simplicial complexes. Guemard takes the Tanner graph of an \textit{arbitrary} CSS code and constructs a simplicial complex from it, referred to as the \textit{Tanner cone complex}. By generating simplicial complexes that are coverings of the initial simplicial complex and then restricting to a Tanner graph, Guemard is able to construct new CSS codes. In contrast to Tanner graph coverings, which in general may not themselves be valid Tanner graphs, Guemard's method is guaranteed to generate valid codes. While our work uses similar ideas in the context of BB codes, we avoid the issues associated with graph coverings as we show that our graph covers, under the right conditions, always correspond to instances of BB codes and are therefore valid QEC codes. We are also able to exploit the fact that we have simple algebraic conditions that guarantee a Tanner graph is a cover in order to more efficiently search for new codes instead of explicitly generating covering graphs. Alongside this, we are able to derive bounds on code parameters under certain conditions and extend the covering graph ideas beyond code construction to analyse logical operators\footnote{In \cite{Guemard2025Moderate}, Guemard also considers the lifting of logical operators using the transfer homomorphism from covering space theory.}. 

In Section \ref{sec:Prelims}, we review the necessary background in graph theory, homological algebra, and the BB codes for the constructions in the rest of the paper. Section \ref{sec:SequenceBB} contains the proof of our main theorem that establishes a correspondence between the algebraic conditions relating a cover and base BB code and the existence of a Tanner cover graph of the base code's Tanner graph. The graph covering map induces chain maps between the two BB codes that we show in Section \ref{sec:ProjAndLift} can be used to lift and project logical operators between a cover and base code and bound their $k$ and $d$ parameters in certain cases. Section \ref{sec:GraphCodeAuts} proves that they can also be used to lift automorphisms from a base to a cover code and that when $h$ is odd and $k_h=k$, the logical action on the base and lifted logical bases is the same. In Section \ref{sec:CoverCodeNumerics}, we present numerical data on a variety of $h$-cover codes, including several sequences of BB codes with weight 6 and weight 8 checks.

\section{Preliminaries}
\label{sec:Prelims}
In this section, we briefly summarize the necessary background in graph theory, homological algebra, and the bivariate bicycle (BB) codes relevant to this work. We assume familiarity with standard stabilizer code theory and basic algebraic notions such as rings, fields, and modules. Those readers familiar with this material can skip directly to Section \ref{sec:SequenceBB}. 

\subsection{Graph Theory}
The constructions in this paper are inspired by the general notion of covering spaces in topology. Intuitively, a covering space $\tilde{X}$ of a topological space ${X}$ has the property that $\tilde{X}$ and $X$ look the same locally while possibly having different global topological structure.

The idea of a covering space carries over to graphs, which will be necessary to consider when applying covering spaces in the context of quantum error correcting codes. We review here those concepts from graph theory that are necessary for understanding the notion of graph covers. We henceforth denote all graphs as $G$ or $G = (V,E)$, where $V$ is the vertex set and $E$ the edge set, and write edges $e \in E$ as pairs of vertices $e = [v_1, v_2]$. 

Given two graphs, the proper notion of a map between them that preserves the edge structure is as follows: 

\begin{dfn}[\textbf{Graph Homomorphism and Isomorphism}]
\label{dfn:graph-hom-iso}
    Let $G(V,E)$ and $H(V',E')$ be two graphs. A graph homomorphism is a function $f \colon V \rightarrow V'$ such that if $[u,v] \in E$, then $[f(u), f(v)] \in E'$. If its inverse $f^{-1}$ is also a graph homomorphism, then $f$ is a graph isomorphism. Equivalently, $f \colon V \rightarrow V'$ is a graph isomorphism if it is a bijection such that for all $u,v \in V$, $[u, v] \in E$ if and only if $[f(u), f(v)] \in E'$. 
\end{dfn}


We are often interested in looking at portions of a graph around a vertex: 

\begin{dfn}[\textbf{Neighbourhood of a Vertex}]
Given a graph $G(V,E)$, the neighbourhood $N(v)$ of a vertex $v \in V$ is the set of all other vertices in the graph connected to $v$ by an edge, $N(v) = \{v_i \in V\ |\ [v, v_i] \in E \}$. 
\label{dfn:neighbour}
\end{dfn}

We can now define the notion of a covering graph, a concept that is used extensively in this work.

\begin{dfn}[\textbf{Covering Graph}]
A covering graph or a \textit{lift} of a graph $G(V,E)$ is a graph $C(\tilde{V}, \tilde{E})$ equipped with a surjective map $p \colon \tilde{V} \to V$ satisfying the condition that $\forall \tilde{v} \in \tilde{V}$, $p$ restricts to a bijection between $N(\tilde{v})$ and $N(p(\tilde{v}))$. $p$ is then said to be a \textit{covering map}. If $p$ is such that $p^{-1}(v)$ has cardinality $n$ for all $v \in V$, then $C(\tilde{V}, \tilde{E})$ is an $n$\textit{-sheeted cover}.  
\label{dfn:cover-graph}
\end{dfn}

\begin{figure}
    \centering
    \includegraphics[width=0.25\linewidth]{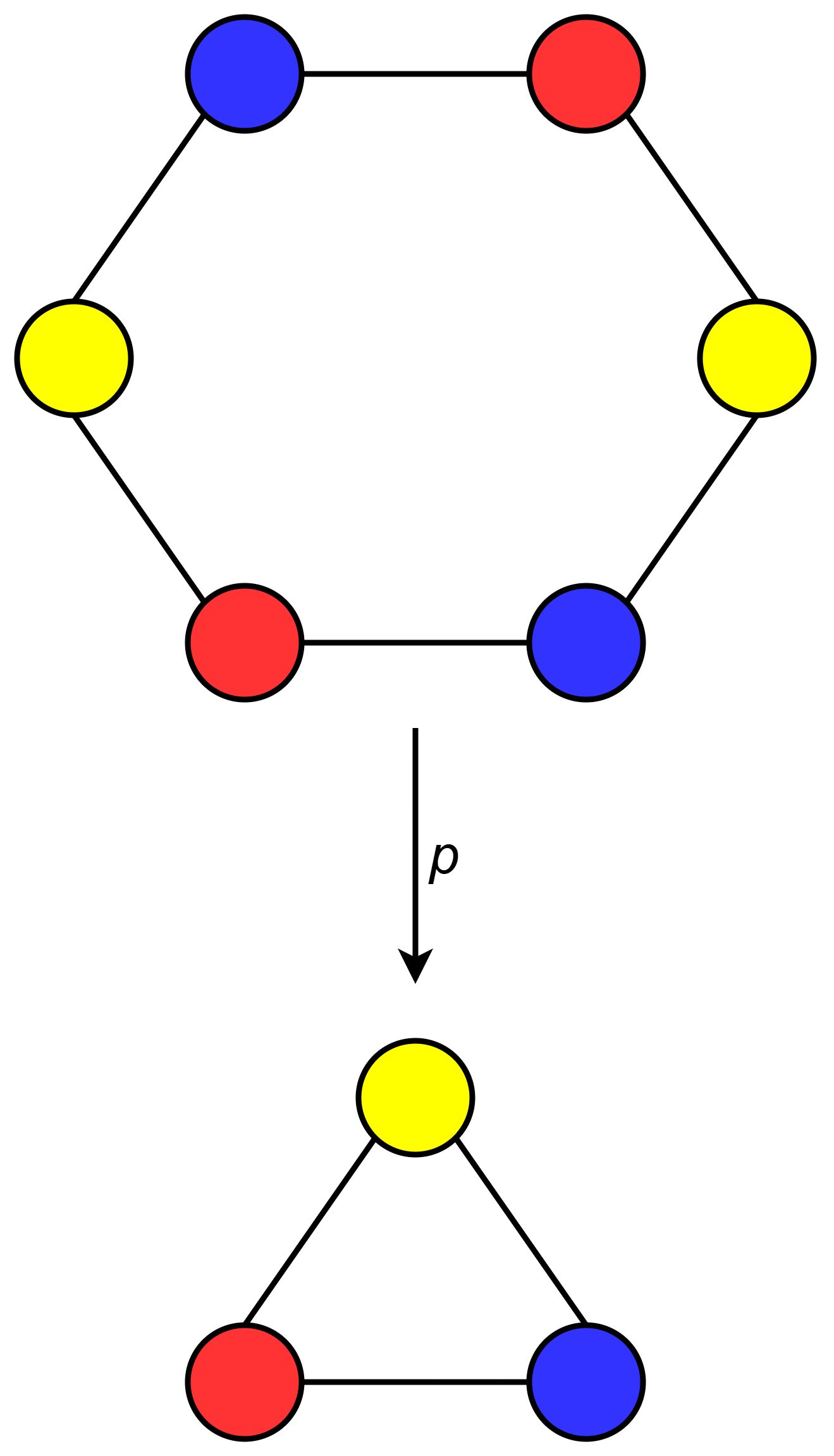}
    \caption{An example of a double cover (top) of a base graph (bottom). The colourings of vertices indicate which vertices are mapped to which under the covering map $p$.}
    \label{fig:covering-graph-ex}
\end{figure}

Figure \ref{fig:covering-graph-ex} shows an example of a double cover of a triangle graph. For each coloured vertex in the base graph, there are two vertices of the same colour in the covering graph that map to the base vertex under the covering map $p$. 

Given two graphs, it can be difficult to verify that one is a cover of the other. However, there is a common way of explicitly generating instances of covering graphs of a given graph using  what are called voltage graphs \cite{Gross1974Voltage}.

\begin{dfn}[\textbf{Voltage Assignment}]
     Let $G(V,E)$ be a graph\footnote{A voltage assignment is usually defined on a directed graph but can be extended to undirected graphs by replacing every undirected edge with pairs of oppositely ordered directed edges and requiring that these edges have labels that are inverse to each other in $\Gamma$. As the directionality of edges is unimportant for us, we leave this implicit.} and $\Gamma$ be a group. A voltage assignment is a function $\nu \colon E \rightarrow \Gamma$ that assigns to every edge $e \in E$ an element of $\Gamma$.
\end{dfn}

\begin{dfn}[\textbf{Voltage Graph and Right-Derived Graph}]
    A $\Gamma$-voltage graph is a pair $$(G(V,E), \ \nu \colon E \rightarrow \Gamma).$$ The \textit{right-derived graph} or \textit{right }$\Gamma$\textit{-lift} of $G$ is the graph $\tilde{G}(\tilde{V}, \tilde{E})$ with vertices $\tilde{V} = V \times \Gamma$ and edges $\tilde{E} = E \times \Gamma$. The vertices of $\tilde{G}$ are given by tuples $(v,\gamma)$ with $\gamma \in \Gamma$ while its edges have the form $\tilde{e} = [(u,\gamma), (v,\gamma\nu(e))]$, where $\nu(e) \in \Gamma$.
\label{dfn:volt-and-right-derived-graph}
\end{dfn}


It can be proven that (right) derived graphs of a base graph $G$ with respect to a group $\Gamma$ are in fact $|\Gamma|$-sheeted covers of $G$ (see \cite{TopGraphTheory}), a fact we will make use of later. Note that a left-derived graph involves multiplication by $\nu(e)$ on the left as opposed to the right. However, as all groups used in this work are Abelian, the left and right derived graphs are equivalent and we therefore leave this implicit throughout.

The choice of group $\Gamma$ and voltage assignment $\nu$ will determine the exact details of the covering graph produced. For example, consider the voltage group $\Gamma = \mathbb{Z}_2$ which is order 2 and therefore produces double covers. The choice of voltage assignment in Figure \ref{fig:volt-assign-ex} generates a derived graph that is exactly the double cover shown in Figure \ref{fig:covering-graph-ex}.

\begin{figure}[H]
    \centering
    \includegraphics[width=0.20\linewidth]{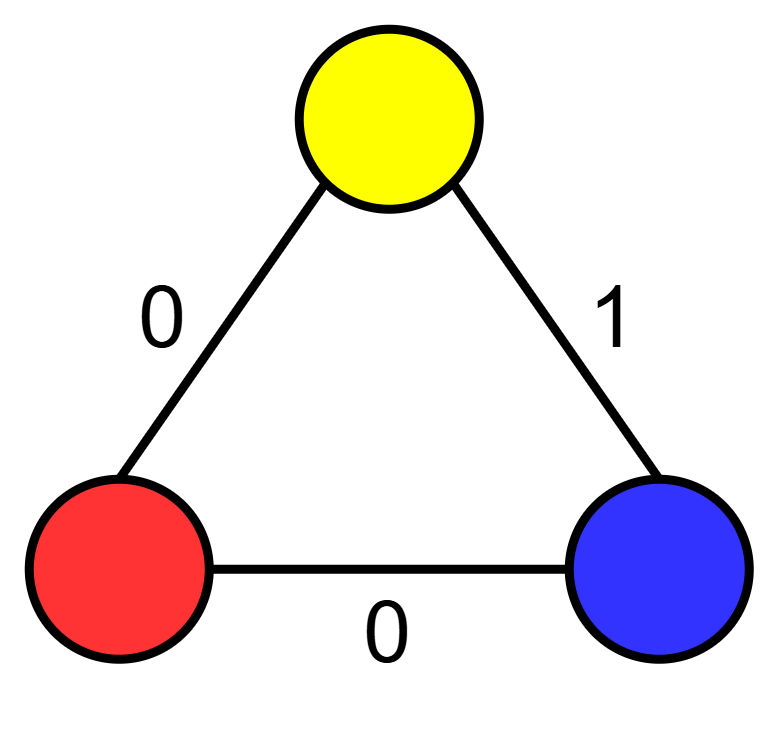}
    \caption{A choice of voltage assignment $\nu : E \to \mathbb{Z}_2$ that produces the double cover in Figure \ref{fig:covering-graph-ex}.}
    \label{fig:volt-assign-ex}
\end{figure}

\subsection{Homological Algebra}

In this section, we let $R$ denote a commutative ring with identity. 

\begin{dfn}[\textbf{Chain and Cochain Complexes}]
    A \textit{chain complex} is a sequence of $R$-modules $C_i$ and $R$-module homomorphisms $\partial_i$ called \textit{boundary maps} 
\begin{equation}
\label{eq:chaincomplex}
    ... \xrightarrow[]{\partial_{i+2}} C_{i+1} \xrightarrow[]{\partial_{i+1}} C_i \xrightarrow[]{\partial_i} C_{i-1} \xrightarrow[]{\partial_{i-1}} ... 
\end{equation}

with the property that $\partial_{i} \circ \partial_{i+1} = 0$ for all $i$. We use the notation $(C_{\bullet}, \partial_{\bullet})$ to denote a chain complex, or $C_{\bullet}$ when the boundary maps are implicit.  

A \textit{cochain complex}, denoted as $(C^{\bullet}, \delta^{\bullet})$ or $C^{\bullet}$, is defined analogously but with the boundary maps $\delta^i$ going in the direction of increasing indices. \end{dfn}
We will sometimes use the term \textit{complex} when it is clear from the context whether a chain or cochain complex is meant. The $i$-th position in a (co)chain complex is referred to as the $i$\textit{-th} \textit{degree}. If a (co)chain complex has only 0's above degree $j$ and below degree $i$, it is said to have \textit{length} $l=|j-i|$. 

The defining condition of a chain complex $\partial_i \circ \partial_{i+1} = 0$ is equivalent to $\Ima \partial_{i+1} \subseteq \ker \partial_i$. If the reverse containment $\Ima \partial_{i+1} \supseteq \ker \partial_i$ also holds at a particular $i$, the chain complex $C_{\bullet}$ is said to be \textit{exact at degree $i$}. If it holds at every $i$, the complex is said to be \textit{exact}. 


More generally, we have the following concept: 
\begin{dfn}[\textbf{Homology and Cohomology}]
    The $i$\textit{-th homology} of $C_{\bullet}$ is defined as the quotient module $$H_i(C_{\bullet}) = \frac{\ker \partial_i}{\Ima \partial_{i+1}}.$$

    The $i$\textit{-th cohomology} of $C^{\bullet}$ is defined as $$H^i(C^{\bullet}) = \frac{\ker \delta^{i+1}}{\Ima \delta^i}.$$
\end{dfn}

$H_i(C_{\bullet})$ ($H^i(C^{\bullet})$) is non-zero in general but is 0 if and only if $\Ima \partial_{i+1} = \ker \partial_i$ (if and only if $\Ima \delta^i = \ker \delta^{i+1}$), i.e. if the complex is exact at degree $i$. (Co)homology thus quantifies the failure of exactness at a particular degree in the (co)chain complex. 

Given two chain complexes $C_{\bullet}$ and $D_{\bullet}$, there are a variety of ways in homological algebra for combining them to build other chain complexes. The most important for error correction applications is the following: 

\begin{dfn}[\textbf{Tensor Product of Chain Complexes}]
Let $(A_{\bullet},\partial_{\bullet}^A)$ and $(B_{\bullet},\partial_{\bullet}^B)$ be chain complexes of $R$-modules. The tensor product of $A_{\bullet}$ and $B_{\bullet}$ is the chain complex $A_{\bullet} \otimes_R B_{\bullet}$ consisting of the $R$-modules $$(A_{\bullet} \otimes_R B_{\bullet})_{n} = \bigoplus_{i+j=n} A_i \otimes_R B_j$$ and boundary maps $$\partial_n^{A_{\bullet} \otimes_R B_{\bullet}} = \bigoplus_{i+j=n} \partial^A_i \otimes_R I_{B_j} + (-1)^i I_{A_i} \otimes_R \partial_j^B,$$ where $I_{A_i}$ ($I_{B_j}$) refers to the identity map on $A_i$ ($B_j$).
\end{dfn}

Given two (co)chain complexes, we also want a suitable notion of a `map' between them that preserves the (co)chain complex property that successive (co)boundary maps compose to $0$. This is given as follows:

\begin{dfn}[\textbf{Chain Map}]
    Let $(C_{\bullet},\partial_{\bullet})$ and $(D_{\bullet},\tilde{\partial}_{\bullet})$ be chain complexes. A \textit{chain map} or \textit{morphism of chain complexes} between $C_{\bullet}$ and $D_{\bullet}$ is denoted as $f_{\bullet} \colon C_{\bullet} \rightarrow D_{\bullet}$ and consists of a collection of $R$-module homomorphisms $f_i \colon C_i \rightarrow D_i$ at each degree $i$ such that the following diagram commutes:
\begin{equation}
\label{commdiag}
    \begin{tikzcd}
        \ldots \arrow[r, ""] 
        & C_{i+1} \arrow[r, "\partial_{i+1}"] \arrow[d, "f_{i+1}"]
        & C_{i} \arrow[d, "f_{i}"] \arrow[r, "\partial_i"]
        & C_{i-1} \arrow[d, "f_{i-1}"] \arrow[r,""]
        & \ldots  \\
        \ldots\arrow[r, ""]
        & D_{i+1} \arrow[r, "\tilde{\partial}_{i+1}"]
        & D_{i}  \arrow[r, "\tilde{\partial}_i"] 
        & D_{i-1} \arrow[r,""]
        & \ldots
    \end{tikzcd}
\end{equation}
The commuting condition is equivalent to demanding that each square in the above diagram commutes for all $i$, i.e. $$f_i \circ \partial_{i+1} = \tilde{\partial}_{i+1} \circ f_{i+1} \ \ \forall i.$$
\end{dfn}

The notion of a cochain map of cochain complexes is analogous. 

\begin{rmk}
\label{rmk:InducedHomMap}
    Chain maps of chain complexes automatically induce well-defined maps between the homologies of the complexes defined as follows for any degree $i$ \cite{Weibel_1994}: 
\begin{align}
    \hat{f}_i \coloneqq H_i(f_{\bullet}) &\colon H_i(C_{\bullet}) \rightarrow H_i(D_{\bullet}) \nonumber \\
    \hat{f}_i([x]) &\coloneqq [f_i(x)],
\end{align}
where $[x] \in H_i(C_{\bullet})$ is an equivalence class. 
An analogous definition holds for the induced maps between the cohomologies of two cochain complexes given a cochain map. 

$H_i$ also preserves the identity chain map $I_{\bullet} \colon C_{\bullet} \rightarrow C_{\bullet}$, where the maps $I_i \colon C_i \rightarrow C_i$ are simply the identity map $I_i(x) = x$, and the composition of chain maps: 
\begin{align}
    I_{H_i} \coloneqq H_i(I_{\bullet}) &\colon H_i(C_{\bullet}) \rightarrow H_i(C_{\bullet}) \nonumber \\
    I_{H_i}([x]) &= [I_i(x)] = [x],
\end{align} 
and 
\begin{equation}
    H_i(g_{\bullet} \circ f_{\bullet}) = H_i(g_{\bullet}) \circ H_i(f_{\bullet}),
\end{equation} 
where $f_{\bullet} \colon B_{\bullet}  \rightarrow C_{\bullet}$ and $g_{\bullet} \colon C_{\bullet} \rightarrow D_{\bullet}$ are two chain maps. The same holds for $H^i$ and the composition of cochain complexes.\hfill $\blacksquare$ 
\end{rmk}

\subsection{Homological QEC and Tanner Graphs}

In this section, we specialize to $R = \bbF_2$ and chain complexes $C_{\bullet}$ of finite-dimensional $\bbF_2$-vector spaces of length at most 2, i.e. where in the notation above $C_i = 0$ for all $i > 2$ and $i < 0$. We also assume all vector spaces are equipped with a basis. 




\begin{rmk}
    A binary $[[n,k,d]]$ CSS code $Q$ is specified by two classical codes that are the kernels $\ker H_X$ and $\ker H_Z$ of two binary parity check matrices $H_X$ and $H_Z$ respectively. Both matrices have $n$ columns corresponding to the $n$ qubits of the code. The $n_X$ rows of $H_X$ represent the $X$ stabilizers of the code and the $n_Z$ rows of $H_Z$ represent the $Z$ stabilizers. 
    
    These stabilizers are all required to pairwise commute, which is equivalent to the condition that $H_Z H_X^T = 0$ or, equivalently, $H_X H_Z^T = 0$. We can therefore represent this code as a length 2 chain complex\footnote{We can also define the chain complex with $H_X^T$ first and $H_Z$ second.}
\begin{equation}
\label{eq:CSScodecomplex}
    Q_{\bullet} \colon 0 \xrightarrow[]{} Q_2 = \bbF_2^{n_Z} \xrightarrow[]{H_Z^T} Q_1 = \bbF_2^{n} \xrightarrow[]{H_X} Q_0 = \bbF_2^{n_X} \xrightarrow[]{} 0.
\end{equation}

Conversely, any chain complex $C_{\bullet}$ of $\bbF_2$-vector spaces can be used to define a CSS code by selecting any three consecutive terms $C_{i+1}, C_i, C_{i-1}$ to be the space of $X$ checks, qubits, and $Z$ checks respectively and defining $H_Z^T \coloneqq \partial_{i+1}$ and $H_X \coloneqq \partial_i$. \hfill $\blacksquare$ 
\end{rmk}

\begin{rmk}
\label{rmk:kofCSScode}
    Note that by definition, the $Z$ logical operators of $Q$ commute with the $X$ stabilizers but are not themselves $Z$ stabilizers. The first condition implies that they are in $\ker H_X$ and the second condition that they are \textit{not} in $\Ima H_Z^T$. Thus, they are elements of $H_1(Q_{\bullet}) = \ker H_X/\Ima H_Z^T$, where the $[0]$ equivalence class in $H_1(Q_{\bullet})$ represents the stabilizers and the non-zero equivalence classes the non-trivial logical operators.

Taking the transpose of the maps in \eqref{eq:CSScodecomplex} gives the cochain complex
\begin{equation}
    Q^{\bullet} \colon 0 \xleftarrow[]{} Q^2 = \bbF_2^{n_Z} \xleftarrow[]{H_Z} Q^1 = \bbF_2^{n} \xleftarrow[]{H_X^T} Q^0 = \bbF_2^{n_X} \xleftarrow[]{} 0.
\end{equation}
By a similar logic, $H^1(Q^{\bullet}) = \ker H_Z/\Ima H_X^T$ represents the vector space of $X$ logical operators. It can be shown that $H^1(Q^{\bullet}) \cong H_1(Q_{\bullet})$.\footnote{This isomorphism is not true in general, but holds when working with complexes of finite dimensional vector spaces.} Then $k$ in fact equals $$k = \dim H^1(Q^{\bullet}) = \dim H_1(Q_{\bullet}).$$ Let $$d_X = \min \{ |c| \in \bbF_2^n \ | \ [c] \in H^1(Q^{\bullet}), [c] \neq [0] \}, \hspace{1em} d_Z = \min \{ |c| \in \bbF_2^n \ | \ [c] \in H_1(Q_{\bullet}), [c] \neq [0] \}$$ denote the $X$ and $Z$-\textit{distances} respectively. The \textit{distance} of the code is $d = \min\{d_X, d_Z\}$. \hfill $\blacksquare$ 
\end{rmk}
\begin{rmk}
    A correspondence also holds between CSS codes and certain bipartite graphs called \textit{Tanner graphs}. We denote them by $T(Q) = (V_C \cup V_Q, E_{CQ})$, where $V_Q$ is the vertex set of qubits, $V_C$ is the vertex set of checks, and an edge in $E_{CQ}$ exists between a qubit and check vertex if the check acts on that qubit in the code. $V_C$ can be further decomposed as $V_C = V_X \cup V_Z$ into disjoint vertex sets of $X$ and $Z$ checks. We often use the notation $T$ when the code $Q$ is clear from the context. \hfill $\blacksquare$ 
\end{rmk} 
The Tanner graph representation of a CSS code allows us to consider covering graphs $\tilde{T}$ of a Tanner graph $T$ that may define new QEC codes, as in Figure \ref{fig:double-cover}. 
\begin{figure}
    \centering
    \includegraphics[width=0.20\linewidth]{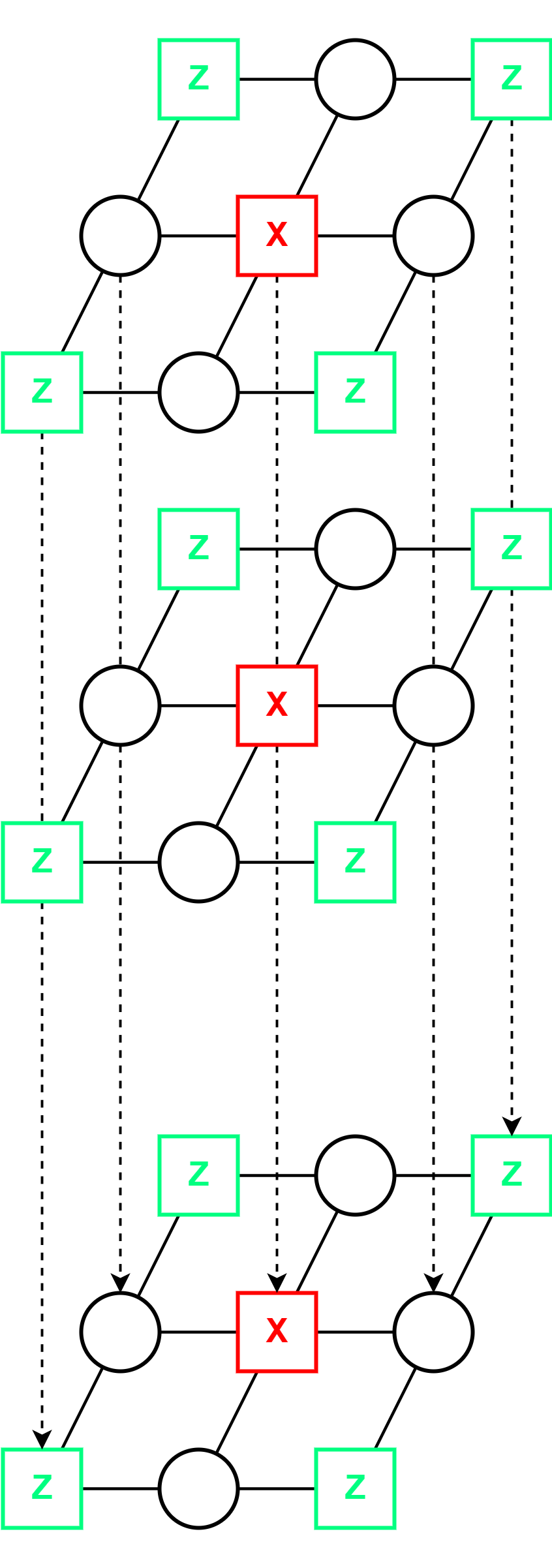}
    \caption{An example of a trivial (disconnected) double cover of a Tanner graph. The bottom graph is the base and the top two graphs together form the covering graph. Vertical dashed lines show how checks and qubits in the cover project to checks and qubits in the base.}
    \label{fig:double-cover}
\end{figure}

\begin{dfn}[\textbf{Cover Code and Base Code}]
    Let $\tilde{Q}$ be a CSS code with a Tanner graph $T(\tilde{Q})$ that is an $h$-sheeted covering of a Tanner graph $T(Q)$ of a code $Q$. We say $\tilde{Q}$ is an $h$\textit{-cover code} of $Q$ and that $Q$ is the \textit{base code}.
\end{dfn}

\subsection{Bivariate Bicycle Codes}

We now recall the basics of the bivariate bicycle code construction; for the full details see \cite{Bravyi2024BB}. The BB codes are group algebra codes for the product group $(\bbZ_l \times \bbZ_m, +)$, where $(\bbZ_l, +)$ and $(\bbZ_m, +)$ are cyclic groups of order $l$ and $m$ respectively. The group algebra $R = \bbF_2[\bbZ_l \times \bbZ_m]$ is ring isomorphic to 
\begin{equation}
\label{eq:bbcodeAlgisos}
    R = \bbF_2[\bbZ_l \times \bbZ_m] \cong \bbF_2[\bbZ_l] \otimes_{\bbF_2} \bbF_2[\bbZ_m] \cong \frac{\bbF_2[x]}{(x^l-1)} \otimes_{\bbF_2} \frac{\bbF_2[y]}{(y^m-1)} \cong \frac{\bbF_2[x,y]}{(x^l-1, y^m-1)}.
\end{equation}
We note that $R$ is, in particular, an $\bbF_2$-vector space of dimension $lm$ since it has the monomial basis
\begin{equation}
\label{eq:monbasis}
    M = \{x^i y^j \ | \ 0 \leq i \leq m-1, \ 0 \leq j \leq l-1\}.
\end{equation} 

Let $a, b \in R$ and consider $R$ as a (rank 1) $R$-module. Consider the following chain complexes $$R \xrightarrow[]{\cdot a} R$$ $$R \xrightarrow[]{\cdot b} R$$ determined by multiplication by $a$ and $b$. Taking the tensor product of these chain complexes and using the standard isomorphism $R \otimes_R R \cong R$ gives the following complex of $R$-modules\footnote{Since the coefficient ring of the polynomials is $\bbF_2$, the sign that appears in the definition of the boundary map for the tensor product of two chain complexes is irrelevant.} 
\begin{equation}
\label{eq:groupAlgComplex}
    R \xrightarrow[]{\begin{pmatrix} b \\ a \end{pmatrix}} R^2 \xrightarrow[]{\begin{pmatrix} a & b \end{pmatrix}} R.
\end{equation} 




Via the vector space isomorphism $R \cong \bbF_2^{lm}$ mentioned above, every element $r$ of $R$ determines a linear map from $\bbF_2^{lm} \rightarrow \bbF_2^{lm}$ corresponding to ``left multiplication by $r$". We have the following characterization for what the corresponding linear maps look like: 

\begin{prop}
\label{thm:groupAlgMatrixRep}
    Let $I_l$ be the $l \times l$ identity matrix. Let $v = (0,1,\ldots,0)^T \in \bbF_2^k$ and $C_k$ be the $k \times k$ circulant matrix with entries $(C_k)_{ij} = v_{(i-j) \ \text{mod} \ k}$ where $v_l$ refers to the $l$-th component of $v$. Then $$\varphi \colon R \rightarrow \bbF_2^{lm \times lm}$$ defined by $$x \mapsto C_l \otimes I_m, \ \ \  y \mapsto I_l \otimes C_m$$ is an injective ring homomorphism.
\end{prop}

\begin{proof}
    Any ring $R$ acting on its underlying Abelian group $(R, +)$ via left-multiplication defines a canonical injective ring homomorphism between $R$ and the endomorphisms of $(R,+)$ (see Proposition 2.7 of Ch. 3 in \cite{algCh0}). We claim that upon choosing a suitable basis, this ring homomorphism is given by that in the theorem statement.  
    
    Since $(R,+) \cong \bbF_2^{lm} \cong \bbF_2^l \otimes_{\bbF_2} \bbF_2^m$ in our case, it suffices to determine how left-multiplication by $x$ and $y$ act on each piece of the tensor product, i.e. on $\bbF_2^l \cong \bbF_2[x]/(x^l-1)$ and $\bbF_2^m \cong \bbF_2[y]/(y^m-1)$ respectively, to determine the action of any polynomial of $R$ on $\bbF_2^l \otimes_{\bbF_2} \bbF_2^m$. Picking the basis $\{1,x,\ldots,x^{l-1}\}$ for $\bbF_2^{l}$ and $\{1,y,\ldots,y^{m-1}\}$ for $\bbF_2^m$ shows that left-multiplication by $x$ and $y$ respectively cyclically permutes the basis elements. The matrices corresponding to this action are precisely those in the theorem statement, verifying the claim. 
\end{proof}

We will often implicitly identify polynomials in $R$ with their binary matrix representations via the above proposition.  

\begin{rmk}
\label{rmk:BBcodeDefn}
    The complex in \eqref{eq:groupAlgComplex} can be written with the preceding results as
\begin{equation}
\label{eq:BBcodeComplex}
    Q_{\bullet} \colon Q_2 = \bbF_2^{lm} \xrightarrow[]{\begin{pmatrix} B \\ A \end{pmatrix}} Q_1 = \bbF_2^{lm} \oplus \bbF_2^{lm} \xrightarrow[]{\begin{pmatrix} A & B \end{pmatrix}} Q_0 = \bbF_2^{lm},
\end{equation} 
where $\varphi(a) = A$ and $\varphi(b) = B$, when restricting to a complex of $\bbF_2$-vector spaces. 
    This complex defines the BB code when the polynomials $A$ and $B$ are specifically chosen to have 3 terms 
\begin{equation}
    A = A_1 + A_2 + A_3, \ B = B_1 + B_2 + B_3.
\end{equation}
Each of the terms $A_i, B_j$ is a monomial in the basis $M$ in \eqref{eq:monbasis}. 

Via the correspondence between CSS codes and chain complexes of $\bbF_2$-vector spaces, the parity check matrices of a BB code are then defined using the polynomials $A$ and $B$ as
\begin{equation}
    H_X = \left(A \ B \right),\ H_Z^T = \begin{pmatrix}
B \\
A 
\end{pmatrix}.
\end{equation}
This construction gives rise to codes that have weight 6 checks because there are 3 terms in each of $A$ and $B$. The codes can be laid out in such a way that 4 of the connections on every qubit are local and 2 are non-local \cite{Bravyi2024BB}. It was proved that $d = d_X = d_Z$ and $\text{rank } H_X = \text{rank } H_Z$ in \cite{Bravyi2024BB} for such codes. 

Note that since circulant matrices are orthogonal, the transpose of $A$ and $B$ when interpreted as matrices corresponds to taking the additive inverse of the powers of $x$ and $y$ modulo their respective orders when interpreted as polynomials. For example, $x^{-1} \coloneqq x^{l-1}$ and $y^{-1} \coloneqq y^{m-1}$. 

The integers $l$ and $m$ for the orders of the cyclic groups $\bbZ_l$ and $\bbZ_m$ are referred to as the \textit{lattice parameters} of the BB code, since we can view the Tanner graph of the BB code as a $2l \times 2m$ 2D lattice with qubits and checks on vertices and periodic boundaries (akin to the toric code). \textit{We will denote a specific BB code via the notation} $Q(A,B,l,m)$ \textit{indicating the defining polynomials $A,B$ and the lattice parameters} $l,m$. \hfill $\blacksquare$
\end{rmk}

\begin{rmk}
    The term in the middle of the chain complex in \eqref{eq:BBcodeComplex} is a direct sum of two $\bbF_2$-vector spaces of dimension $lm$. This implies the BB code qubits have a 2-block structure, i.e. they can be split into two sets of `left' ($L$) and `right' ($R$) qubits each of size $lm$ so that $Q_1=Q_1^L \oplus Q_1^R=\mathbb{F}_2^{lm} \oplus \mathbb{F}_2^{lm}$.

Each of the sets of $L$ and $R$ qubits, $X$ checks, and $Z$ checks are in correspondence with the monomials $r \in M$, so we will occasionally use the following notation (borrowed from \cite{Bravyi2024BB}):
\begin{equation}
\begin{split}
    &q_{r}^L \in Q_1^L,\ q_{r}^R \in Q_1^R, \\
    &Z_{r} \in Q_2,\ X_{r} \in Q_0.
\end{split}
\end{equation}
The sets $\{X_r\}$ and $\{Z_r\}$ generate a basis for the $X$ and $Z$ stabilisers respectively, and the form of the qubit supports of these checks are
\begin{equation}
    \text{Supp}(X_r) = (r(A_1 + A_2 + A_3) , r (B_1 + B_2 + B_3)),
\label{eq:bb-x-supp}
\end{equation}
\begin{equation}
    \text{Supp}(Z_r) = (r(B^T_1 + B^T_2 + B^T_3) , r (A^T_1 + A^T_2 + A^T_3)),
\label{eq:bb-z-supp}
\end{equation}

where the supports in equations \ref{eq:bb-x-supp} and \ref{eq:bb-z-supp} are elements of $Q_1^L \oplus Q_1^R$. 

When discussing logical operators, it will be useful to denote tensor products of Pauli $X$ or $Z$ operators as $X(r,r')$ and $Z(r,r')$ respectively where the terms in the polynomials $r$ and $r'$ index respectively the left and right qubits that the Pauli $X$ or $Z$ operators act on. If $r=1+y$ for example, then the logical operator $X(r, 0)$ would correspond to operator $X_{1} \otimes X_{y}$, i.e. the Pauli $X$ acting on the left qubits indexed by $1$ and $y$ and identity on all other qubits, where we have suppressed the identity operators for clarity.  \hfill $\blacksquare$
\end{rmk}

\begin{rmk}
\label{rmk:BBcodeNotation}
The Tanner graphs $T$ of BB codes likewise have extra structure. The 2-block structure of the BB codes mentioned in the previous remark means that the qubit vertices can be split into two equal size sets of left and right qubits, $V_Q = V_{Q_L} \cup V_{Q_R}$. The vertex set of $T$ therefore decomposes as $V = \cup_i V_i$ where $V_i \in \{V_X, V_Z, V_{Q_L}, V_{Q_R}\}$ and each set $V_i$ contains exactly $lm$ elements. We label the elements of $V_i$ by the elements of the set $\mathbb{Z}_l \times \mathbb{Z}_m$  so that each $V_i \cong \mathbb{Z}_l \times \mathbb{Z}_m$.

Note that the elements of $\mathbb{Z}_l \times \mathbb{Z}_m$ also label the powers of a monomial in the monomial basis set $M$, e.g. $(a,b) \in \mathbb{Z}_l \times \mathbb{Z}_m$ for $x^ay^b \in M$. We specifically denote the powers of the monomials $A_i, B_j$ in the defining polynomials $A$ and $B$ respectively as the tuples $\alpha_i = (\alpha_{i1}, \alpha_{i2}), \beta_j = (\beta_{j1}, \beta_{j2})$ in $\mathbb{Z}_l \times \mathbb{Z}_m$.

The edges in the Tanner graph of a BB code are determined by the Abelian group $(\mathbb{Z}_l \times \mathbb{Z}_m, +)$ and the choice of defining polynomials $A, B$. Figure \ref{fig:BB-TG} shows a subgraph of a BB code Tanner graph with directed edges labelled by monomials $A_i, B_j$ of the defining polynomials $A, B$. Travelling along an edge labelled by $A_i$ in the forward direction corresponds to the addition of $\alpha_i$ while travelling in the reverse direction corresponds to the addition of $-\alpha_i$. This is because addition in $(\mathbb{Z}_l\times \mathbb{Z}_m , +)$ corresponds to multiplication of monomials in the ring $R$, since the multiplication of monomials corresponds to the addition of their powers.

Consider for example an $X$ check vertex labelled by $c = (c_1, c_2) \in (\mathbb{Z}_l \times \mathbb{Z}_m, +)$. The qubits in the neighbourhood of an edge are given by $$q \in \{c + \alpha_1, \ c+ \alpha_2, \ c+ \alpha_3, \ c+ \beta_1, \ c+\beta_2, \ c+\beta_3 \}.$$ Similarly the qubits linked to a $Z$ check vertex labelled $c$ are $$q \in \{ c-\alpha_1, \  c-\alpha_2, \ c-\alpha_3, \  c-\beta_1, \ c-\beta_2, \ c-\beta_3 \}.$$ Thus we can write all edges in a BB code Tanner graph in the form $e = [c,q] = [c, c+\alpha_i]$ or $e = [c,q] = [c, c+\beta_i]$. Note that the addition above is interpreted as modulo $l$ or $m$ depending on if the terms involve $\alpha_i$ or $\beta_i$ respectively. 

When $\{\alpha_i,\beta_j : i,j \in \{1,2,3\}\}$ is a generating set for the group $(\mathbb{Z}_l \times \mathbb{Z}_m,+)$, the Tanner graph is a Cayley graph. Note the edge labelling in Figure \ref{fig:BB-TG} should not be confused with the voltage assignment used later when generating derived graphs. \hfill $\blacksquare$
\end{rmk}

\begin{figure}[H]
    \centering
    \includegraphics[width=0.40\linewidth]{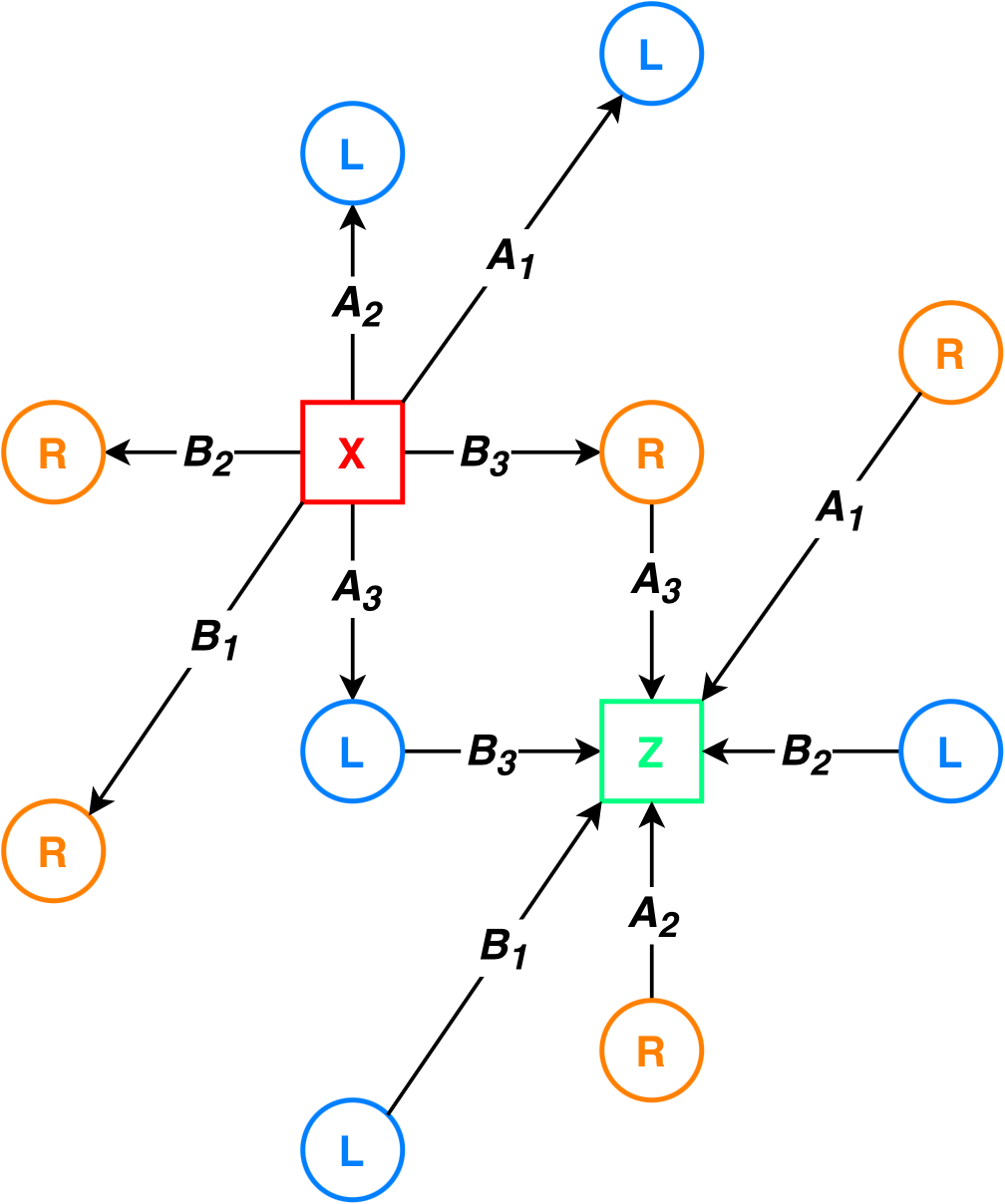}
    \caption{A diagram showing part of a BB code Tanner graph including a single X and Z check. Edges are labelled with elements from the defining polynomials $A,B$. Short edges represent local connections while long edges show non-local connections. See \cite{Bravyi2024BB} for more detail.}
    \label{fig:BB-TG}
\end{figure}

\section{Sequences of Bivariate Bicycle Codes}
\label{sec:SequenceBB}
In this section we will show that under the right conditions, the Tanner graph of one BB code can be viewed as a covering of the Tanner graph of a smaller BB code. This will allow us to construct a sequence of BB codes that are successively larger covers of some smaller base code. All of the BB codes in the original paper \cite{Bravyi2024BB} can be viewed as members of some sequence of so-called cover codes. For example, the Tanner graph of the gross code $[[144,12,12]]$ is a double cover of the Tanner graph of $[[72,12,6]]$ code. 

Given the Tanner graph of a BB code $Q(A,B,l,m)$, an appropriate voltage group $\Gamma$, and a voltage assignment $\nu$, we show that the generated derived graph (covering graph) is a Tanner graph of another BB code. We find that for the Tanner graph of a BB code $\tilde{Q}(\tilde{A}, \tilde{B}, \tilde{l}, \tilde{m})$ to be a cover of another BB code $Q(A,B,l,m)$, the powers of the monomials in $\tilde{A}, \tilde{B}$ must be congruent to those of $A,B$ modulo $l$ and $m$, and that $\tilde{l} = ul$ and $\tilde{m} = tm$. The simple algebraic conditions on the defining polynomials and lattice parameters in Theorem \ref{thm:derived-Tanner-iso} also enable an efficient search over a reduced space of polynomials that meet these conditions (see Section \ref{sec:CoverCodeNumerics}). 

In what follows, we use the notation $\Mod(a,l)$ to denote the remainder after the division of $a$ by $l$. Thus $a = ql + \Mod(a,l)$ where $q = \lfloor a/l \rfloor$. Expressions such as $\Mod(a+b,l) = \Mod(a,l) + \Mod(b,l)$ are interpreted as modulo $l$ on the right hand side. We will sometimes use the notation `$\text{mod } l$' at the end of an equality to make this explicit as needed. 

We now state our main theorem as follows:

\begin{thm}
\label{thm:derived-Tanner-iso}
    With the notation as in Remarks \ref{rmk:BBcodeDefn} and \ref{rmk:BBcodeNotation}, let $Q(A,B,l,m)$ and $\tilde{Q}(\tilde{A}, \tilde{B}, \tilde{l}, \tilde{m})$ be two BB codes with Tanner graphs $T:=T(Q)$ and $\tilde{T} := T(\tilde{Q})$ respectively, where the latter has the vertex set $\tilde{V} = \cup_{i=1}^4 \tilde{V}_i = \tilde{V}_{Q_L} \cup \tilde{V}_{Q_R} \cup \tilde{V}_{X} \cup \tilde{V}_Z$ with $\tilde{V}_i \cong \mathbb{Z}_{\tilde{l}} \times \mathbb{Z}_{\tilde{m}}$ for each $i$. 
    
    Let $V'=\cup_i V'_i = (V_{Q_L} \cup V_{Q_R} \cup V_{X} \cup V_Z) \times \Gamma$, where $V'_i \cong \mathbb{Z}_l \times \mathbb{Z}_m \times \Gamma$ for each $i$ and $\Gamma = \bbZ_u \times \bbZ_t$, be a vertex set and suppose the following conditions are met:
    \begin{enumerate}
        \item $\tilde{l} = ul$ 
        \item $\tilde{m} = tm$
        \item Mod$(\tilde{\alpha}_{i1}, l) = \alpha_{i1}$ and Mod$(\tilde{\alpha}_{i2}, m) = \alpha_{i2}$, where $\tilde{\alpha}_i = (\tilde{\alpha}_{i1}, \tilde{\alpha}_{i2})$ are the powers of $x$ and $y$ in the monomials $\tilde{A_i} = x^a y^b$ in $\tilde{A} = \tilde{A}_1 + \tilde{A}_2 + \tilde{A}_3$, and similarly for $\alpha_i = (\alpha_{i1}, \alpha_{i2})$ and the monomials $A_i$ in $A$. 
        \item Mod$(\tilde{\beta}_{i1}, l) = \beta_{i1}$ and Mod$(\tilde{\beta}_{i2}, m) = \beta_{i2}$, where $\tilde{\beta}_i = (\tilde{\beta}_{i1}, \tilde{\beta}_{i2})$ are the powers of $x$ and $y$ in the monomials $\tilde{B_i} = x^a y^b$ in $\tilde{B} = \tilde{B}_1 + \tilde{B}_2 + \tilde{B}_3$ and similarly for $\beta_i = (\beta_{i1}, \beta_{i2})$ and the monomials $B_i$ in $B$.
    \end{enumerate}
    Consider the derived graph $G(V',E') = D(T,\Gamma)$ of $T$ with vertex set $V'$, voltage group $\Gamma = (\mathbb{Z}_u \times \mathbb{Z}_t, +)$, and voltage assignment $$v(e) = (\lfloor (c_1+\tilde{\alpha}_{i1})/l \rfloor, \lfloor (c_2 + \tilde{\alpha}_{i2})/m \rfloor)$$ on the edges $e=[c, c+\alpha_i] = [(c_1, c_2), (c_1, c_2) + (\alpha_{i1},\alpha_{i2})]$ of $T$ and similarly for the edges $[c, c+\beta_j]$. Then the Tanner graph $\tilde{T}$ is graph isomorphic to the derived graph $G = D(T,\Gamma)$ of $T$. 
\end{thm}

Before proving this theorem, we will first define a bijection between the vertex sets of $\tilde{T}$ and $G$. 

\begin{lem}
\label{lem:vertex-iso}
    With notation as in the statement of the above theorem, let $f \colon \tilde{V}_i \rightarrow V'_i$ be defined by
    \begin{equation}
        f((\tilde{a}, \tilde{b})) = (\text{Mod}(\tilde{a}, l), \text{Mod}(\tilde{b}, m), s_1, s_2),
    \label{eq:graph-hom}
    \end{equation}
    where $s_1 = \lfloor \tilde{a}/l \rfloor $ and $s_2 = \lfloor \tilde{b}/m \rfloor$. Then $f$ is a bijection with inverse $f^{-1} \colon V'_i \rightarrow \tilde{V}_i$ defined by 
    \begin{equation}
    \label{eq:inverse-graph-hom}
        f^{-1}(a, b, s_1, s_2) = (a + ls_1, b + ms_2).
    \end{equation}
\end{lem}
\begin{proof}
It suffices to show that $f^{-1}$ is a two-sided inverse of $f$ as claimed. 

Note that, 
\begin{align*}
    (f \circ f^{-1})(a,b, s_1, s_2) &= f(a + ls_1, b + m s_2) \\
    &= (\Mod(a + ls_1,l), \Mod(b + ms_2,m), \lfloor (a + ls_1)/l \rfloor, \lfloor (b + ms_2)/m \rfloor \\
    &= (a,b,s_1,s_2). 
\end{align*}
Similarly,
\begin{align*}
    (f^{-1} \circ f)(\tilde{a},\tilde{b}) &= f^{-1}((\text{Mod}(\tilde{a}, l), \text{Mod}(\tilde{b}, m), \lfloor \tilde{a}/l \rfloor, \lfloor \tilde{b}/m \rfloor)) \\
    &=(\text{Mod}(\tilde{a}, l) + l\lfloor \tilde{a}/l \rfloor, \text{Mod}(\tilde{b}, m) + m \lfloor \tilde{b}/m \rfloor) \\
    &=(\tilde{a}, \tilde{b}). \qedhere
\end{align*}
\end{proof}

\begin{proof}[Proof of Theorem \ref{thm:derived-Tanner-iso}]
By Definition \ref{dfn:graph-hom-iso}, a graph isomorphism requires an edge preserving bijection between vertices. Lemma \ref{lem:vertex-iso} shows that $f$ is a bijection between the vertices of $\tilde{T}$ and $G$, so we now prove that this map preserves edges. Note that in the following, we consider only edges of the form $[c, c+\alpha_i]$ but the same logic holds for the remaining edges of the form $[c, c+\beta_j]$.

Consider an edge $\tilde{e} = [\tilde{c}, \tilde{q}]$ in $\tilde{T}$ with $\tilde{c} \in \tilde{V}_X, \tilde{V}_Z$ and $\tilde{q} \in \tilde{V}_{Q_L}, \tilde{V}_{Q_R}$. By Remark \ref{rmk:BBcodeNotation}, we can write $\tilde{e} = [\tilde{c}, \tilde{c} + \tilde{\alpha}_i]$ where $\tilde{c} = (\tilde{c}_1, \tilde{c}_2)$ and $\tilde{\alpha}_i = (\tilde{\alpha}_{i1}, \tilde{\alpha}_{i2})$ are in $\mathbb{Z}_{\tilde{l}} \times \mathbb{Z}_{\tilde{m}}$. Now consider an arbitrary edge $e'$ in the derived graph $G$ which has the form
\begin{equation*}
\begin{split}
    e' &= [(a, \gamma) , (b, \gamma + \nu(e))] \\ 
    &= [(a_1, a_2, \gamma_1, \gamma_2), (a_1 + \alpha_{i1}, a_2 + \alpha_{i2}, \gamma_1 + g_1, \gamma_2 + g_2)],
\end{split}
\end{equation*}
where $\nu(e) = g = (g_1, g_2) \in \mathbb{Z}_u \times \mathbb{Z}_t$, $\gamma = (\gamma_1,\gamma_2) \in \mathbb{Z}_u \times \mathbb{Z}_t$, and $\alpha_i = (\alpha_{i1}, \alpha_{i2}) \in \bbZ_l \times \bbZ_m$. We first show that applying the map $f$ to the vertices in an edge $\tilde{e}$ maps it to an edge in the derived graph of the above form. By the definition of $f$, 
\begin{equation*}
    f(\tilde{c}) = (\text{Mod}(\tilde{c}_1, l), \text{Mod}(\tilde{c}_2, m), s_1, s_2)
\end{equation*}
with $$s_1 = \lfloor \tilde{c}_1/l \rfloor, \ s_2 = \lfloor \tilde{c}_2/m \rfloor.$$ Similarly,
\begin{equation*}
    f(\tilde{q}) = (\text{Mod}(\tilde{c}_1+\tilde{\alpha}_{i1}, l), \text{Mod}(\tilde{c}_2+\tilde{\alpha}_{i2}, m), s'_1, s'_2)
\end{equation*}
with $$s'_1 = \lfloor (\tilde{c}_1 + \tilde{\alpha}_{i1})/l \rfloor, \ s'_2 = \lfloor (\tilde{c}_2 + \tilde{\alpha}_{i2})/m \rfloor.$$ 

Define $c_1 = \text{Mod}(\tilde{c_1},l), \ c_2 = \text{Mod}(\tilde{c}_2,m)$. Then $$\tilde{c}_1 = s_1 l + c_1, \hspace{1em} \tilde{c}_2 = s_2 m + c_2,$$ from which it follows that $c_1 < l$ and $c_2 < m$. Since $\tilde{c}_1 < \tilde{l} = ul$, where the last equality is from the first assumption in the theorem statement, we have $s_1 l + c_1 < ul$. This implies $s_1 < u$ and therefore that $s_1 \in \bbZ_u$. A similar reasoning implies $s_2 \in \bbZ_t$ since $\tilde{m} = tm$ by assumption. Thus $f(\tilde{c}) \in \bbZ_l \times \bbZ_m \times \Gamma$. 

Next, note that 
\begin{equation*}
\begin{split}
    &\text{Mod}(\tilde{c}_1 + \tilde{\alpha}_{i1}, l) = (c_1 + \alpha_{i1}) \text{ mod } l,\\
    &\text{Mod}(\tilde{c}_2 + \tilde{\alpha}_{i2}, m) = (c_2 + \alpha_{i2}) \text{ mod } m, 
\end{split}
\label{eq:dg-8}
\end{equation*}
where we used the third and fourth assumptions in the theorem statement and the fact that Mod$(\tilde{c}_1 + \tilde{\alpha}_{i1}, l) = \text{Mod}(\tilde{c}_1, l) + \text{Mod}(\tilde{\alpha}_{i1}, l)$, with the sum on the RHS taken to be modulo $l$.

Finally, note that 
\begin{align*}
    s'_1 &= \lfloor (\tilde{c}_1 + \tilde{\alpha}_{i1})/l \rfloor \\
    &= \lfloor (s_1 l + c_1 + \tilde{\alpha}_{i1})/l \rfloor \\
    &= \lfloor(c_1 + \tilde{\alpha}_{i1})/l \rfloor + s_1 = g_1 + s_1,
\end{align*}
by the definition of our voltage assignment. A similar argument shows that $s'_2 = s_2 + g_2$. Thus, $f(\tilde{q}) = (c_1 + \alpha_{i1}, c_2 + \alpha_{i2}, s_1 + g_1, s_2 + g_2)$, showing that $[f(\tilde{c}),f(\tilde{q})]$ is an edge in the derived graph $G$. 
 
In order to complete the proof, we show the map $f^{-1}$ in \eqref{eq:inverse-graph-hom} is also a graph homomorphism. Let 
\begin{align*}
    e &= [(c,\gamma), (c + \alpha_i, \gamma + g)] \\ 
    &= [(c_1, c_2, \gamma_1, \gamma_2), (\Mod(c_1 + \alpha_{i1},l), \Mod(c_2 + \alpha_{i2},m), \Mod(\gamma_1 + g_1, u), \Mod(\gamma_2 + g_2,t)].
\end{align*}
be an edge in the derived graph $G$, where $g = (g_1, g_2) = \nu(e)$ and we make explicit the integer over which the sums in the last equality are taken modulo of. We must show that $f^{-1}((c,\gamma))$ and $f^{-1}((c+\alpha_i, \gamma + g))$ form an edge in $\tilde{T}$. To do so, it suffices to show that $$f^{-1}((c+\alpha_i, \gamma + g)) = f^{-1}((c,\gamma)) + \tilde{\alpha}_i,$$ where $\tilde{\alpha}_i = (\tilde{\alpha}_{i1}, \tilde{\alpha}_{i2})$. We make the argument for one of the vertex components since the argument for the other is analogous. 

First, let $\tilde{\alpha}_{i1} = \alpha_{i1} + l\lambda$. By the definition of $f^{-1}$,
\begin{align*}
    f^{-1}((c,\gamma)) &= (c_1 + l \gamma_1, c_2 + m \gamma_2) \\
    f^{-1}((c + \alpha_i, \gamma + g)) = (&\Mod(c_1 + \alpha_{i1},l) + \Mod(\gamma_1 + g_1,u)l, \\
    &\Mod(c_2 + \alpha_{i2},m) + \Mod(\gamma_2 + g_2,t)m). \nonumber
\end{align*}
By the definition of our voltage assignment, $$g_1 = \lfloor (c_1+\tilde{\alpha}_{i1})/l \rfloor = \lfloor (c_1+\alpha_{i1} + l\lambda)/l \rfloor = \lfloor (c_1+\alpha_{i1})/l \rfloor + \lambda.$$ We first write 
\begin{align*}
    \Mod(c_1 + \alpha_{i1},l) &= c_1 + \alpha_{i1} - l\lfloor (c_1+\alpha_{i1})/l \rfloor \\
    \Mod(\gamma_1 + g_1,u) &= \gamma_1 + g_1 - us
\end{align*}
where $s = \lfloor (\gamma_1 + g_1)/u \rfloor$. 
This implies 
\begin{align*}
\Mod(c_1 + \alpha_{i1},l) + \Mod(\gamma_1 + g_1,u)l &= c_1 + \alpha_{i1} - l\lfloor (c_1+\alpha_{i1})/l \rfloor + \gamma_1 l + g_1 l - usl \mod \tilde{l}\\
&= c_1 + \alpha_{i1} + \gamma_1 l + \lambda l \mod \tilde{l}\\
&= c_1 + \gamma_1 l + \tilde{\alpha}_1 \mod \tilde{l}, 
\end{align*}
where we used $\tilde{l} = ul$ and the expression for $g_1$ to arrive at the second equality above. Repeating an analogous argument for the other components shows that $f^{-1}((c+\alpha_i, \gamma + g)) = f^{-1}((c,\gamma)) + \tilde{\alpha}_i$ as desired. 
\end{proof}

\begin{thm}
\label{thm:graphProjectionThm}
    Under the conditions in Theorem \ref{thm:derived-Tanner-iso}, the Tanner graph $\tilde{T}$  equipped with the natural projection map $\pi : \mathbb{Z}_{\tilde{l}} \times \mathbb{Z}_{\tilde{m}} \rightarrow \mathbb{Z}_{l} \times \mathbb{Z}_{m}$ defined as $\pi((\tilde{a}, \tilde{b})) = (\text{Mod}(\tilde{a}, l), \text{Mod}(\tilde{b}, m))$ is a $|\Gamma|-$covering of the Tanner graph $T$.
\label{thm:Tanner-cover}
\end{thm}

\begin{proof}
    All derived graphs $D(T,\Gamma)$ are $|\Gamma|$-fold covering graphs of $T$ under the canonical projection map $c \colon V' \rightarrow V$ sending $(v,g) \mapsto v$ (see \cite{TopGraphTheory}). The claim then follows from Theorem \ref{thm:derived-Tanner-iso} upon defining $\pi \coloneqq c \circ f$.
\end{proof}

\begin{thm}
    Given a base BB code $Q = (A,B,l,m)$ with parameters $[[n, k, d]]$, there exists an infinite sequence of $h$-cover codes $\tilde{Q}_h = (\tilde{A}, \tilde{B}, \tilde{l}=ul, \tilde{m} = tm)$ with $h:=ut$, defining polynomials satisfying $\pi(\tilde{\alpha}_i) = \alpha_i$ and $\pi(\tilde{\beta}_j) = \beta_j$ $\forall i,j \in \{1,2,3\}$, and parameters $[[n_h = hn, k_h, d_h]]$.
\label{thm:code-seqs}
\end{thm}

\begin{proof}
    This follows from Theorems \ref{thm:derived-Tanner-iso} and \ref{thm:Tanner-cover}.
\end{proof}

We let the parameters of the base code be $n_1 = n$, $d_1 = d$, and $k_1 = k$ in what follows. We will give certain theoretical guarantees on the behaviour of $k_h$ and $d_h$ in relation to the base code parameters $k$ and $d$ in the subsequent sections. 

Covering graphs found by arbitrary graph lifts are not guaranteed to produce graphs that are still valid Tanner graphs in general. In particular, the checks associated to an arbitrary covering graph of a given Tanner graph are likely to not commute. Guemard addressed this issue in \cite{Guemard2025lifts} by constructing a simplicial complex from the Tanner graph and examining covering complexes. The extra conditions imposed on a covering complex as opposed to a general graph cover ensure that the resulting checks still commute. 

We avoid the need for this as we have shown in Theorem \ref{thm:derived-Tanner-iso} that with the right choices of a group and voltage assignment, the derived graph is guaranteed to be the Tanner graph of a BB code. This means that while we also consider covering graphs, we are restricted by construction to those graph covers that define valid codes. 

Suppose we start with a BB code with lattice parameters $l,m$ and check polynomials of the form 
\begin{equation}
    A = \sum_{i=1}^3 x^{\alpha_{i1}} y^{\alpha_{i2}},     \hspace{1em} B =  \sum_{i=1}^3 x^{\beta_{i1}} y^{\beta_{i2}}.
\end{equation}
We can pick the check polynomials $\tilde{A}$ and $\tilde{B}$ of an $h$-cover code, where $h = ut$, as follows:
\begin{equation}
    \tilde{A} = \sum_{i=1}^3 x^{\alpha_{i1} + la_{i1}} y^{\alpha_{i2} + ma_{i2}},     \hspace{1em} \tilde{B} =  \sum_{i=1}^3 x^{\beta_{i1}+lb_{i1}} y^{\beta_{i2}+mb_{i2}},
\label{eq:cover-poly-explicit}
\end{equation}
where $a_{i1}, b_{i1} \in \{0, \ldots, u-1\}$ and $a_{i2}, b_{i2} \in \{0,\ldots,t-1\}$. It is easy to see that any cover code polynomial of this form satisfies the conditions of the preceding theorem.

We will now show some illustrative examples of BB cover codes.

\begin{ex}
\label{ex:grossCodeCover}
    The gross code $[[144,12,12]]$ with $$(\tilde{A} = x^3 + y + y^2, \ \tilde{B} = y^3 + x + x^2, \ \tilde{l} = 12,\ \tilde{m}=6)$$ is a double cover of the $[[72,12,6]]$ code with $$(A = x^3 + y + y^2,\ B = y^3 + x + x^2,\ l = 6,\ m = 6).$$ The polynomials defining the two codes are the same, so the fact that the gross code is a double cover simply follows from the fact that $\tilde{l}/l = 2$. \hfill $\blacksquare$
\end{ex}

\begin{ex}
\label{ex:18-8-2-code-ex}
    Consider the $[[18,8,2]]$ BB code defined by $$(A = 1 + y+ y^2, \ B = 1 + x + x^2,\ l=3,\ m=3).$$ The code defined by $$(\tilde{A} = x^3 + y + y^2, \ \tilde{B} = 1 + x + x^2,\ \tilde{l} = 6,\ \tilde{m} = 3)$$ is a 2-cover code with parameters $[[36,8,4]]$. This is because $\tilde{l}/l = 2$ and because the defining polynomials only differ in a single term $\tilde{A}_1 = x^3$ such that $$\pi(\tilde{\alpha}_1) = \pi((3,0)) = (\text{Mod}(3,3), 0) = (0,0) = \alpha_1,$$ the label for $1$ in the monomial basis ($x^0 y^0 = 1)$. \hfill $\blacksquare$
\end{ex}

\begin{ex}
    The covering map allows for some freedom, especially as $l,m$ increase. Another valid 2-cover of the $[[18,8,2]]$ code in Example \ref{ex:18-8-2-code-ex} is the $[[36,8,4]]$ code defined by $$(\tilde{A} = x^3 + x^3 y + y^2, \ \tilde{B} = 1 + x + x^2, \ \tilde{l}=6, \ \tilde{m}=3).$$ Similarly to the previous example, we see that $\pi(\tilde{\alpha}_1) = \alpha_1$ and $$\pi(\tilde{\alpha}_2) = \pi((3,1)) = (\text{Mod}(3,3), \text{Mod(1,3)}) = (0,1) = \alpha_2,$$ the label for the monomial $y$ ($x^0 y^1 = y$). \hfill $\blacksquare$
\end{ex}

\section{Projection and Lifting of Logical Operators}
\label{sec:ProjAndLift}

We now extend the graph covering map of Theorem \ref{thm:graphProjectionThm} to a variety of linear maps between the free $\bbF_2$-vector spaces on the base and cover graph vertex sets. 

In Section \ref{sec:ProjLogical}, we use the graph covering map to define a (co)chain map between two BB codes where the Tanner graph of one is a covering graph of the other under the conditions of Theorem \ref{thm:derived-Tanner-iso}. This (co)chain map induces a map on (co)homology that allows for the projection of logical operators in the cover code those of the base code. 

In Section \ref{sec:LiftLogical}, we define a ``lifting" chain map from the base to the cover code and show that this map is the dual of the projection (co)chain map. The lifting map is the more important of the two, as it can be used to lift logical operators of the base code to those of the cover code. We also use the properties of the projection and lifting maps to prove that the number of logical qubits in the cover code must be at least the number of qubits in the base code, i.e. $k_h \geq k$. A consequence of this is that when $h$ is odd and $k_h = k$, a complete logical basis for the base code will lift to a complete, but not necessarily minimum weight, logical basis for the cover code.

We then prove distance bounds under certain conditions in Section \ref{sec:BoundDist}. When $h$ is odd, we have the upper bound $d_h \leq hd$ for the cover code; if $k_h=k$ is also satisfied, then we also lower bound the distance as $d_h \geq d$. Finally, we define a weight-preserving lift (co)chain map that depends on the choice of a cover graph vertex in the pre-image of each base graph vertex under the covering map. Since this choice is arbitrary, this map does not necessarily define a (co)chain map between the base and cover codes. If such a (co)chain map can be found however, it can be used to find weight $d$ logical operators of the cover code and therefore serves as a useful numerical tool for ruling out cover codes with distance equal to the base code. 

Lastly, we conclude in Section \ref{sec:Qudits} with some observations about extending our results to the case of qudit BB codes. 

\subsection{Projecting Logical Operators from a Cover to a Base Code}
\label{sec:ProjLogical}

The graph covering map $\pi$ in Theorem \ref{thm:graphProjectionThm} induces a surjective linear map $p \colon \mathbb{F}_2^{\tilde{l}\tilde{m}} \rightarrow \mathbb{F}_2^{lm}$ between the free vector spaces of the cover and base graph vertex sets, where the bases of the two vector spaces are the monomial bases $\tilde{M}$ and $M$ respectively of \eqref{eq:monbasis}. Its abstract description is as follows:

\begin{dfn}[\textbf{Projection Map}]
\label{dfn:projectionMap}
    Let $\pi \colon \tilde{V} \rightarrow V$ from $C(\tilde{V},\tilde{E})$ to $B(V, E)$ be a graph covering map. Let $e_{\tilde{v}}$ be a basis vector in $\bbF_2^{|\tilde{V}|}$ corresponding to a vertex $v \in \tilde{V}$. Define a surjective linear map $p \colon \bbF_2^{|\tilde{V}|} \rightarrow \bbF_2^{|V|}$ by
    \begin{equation}
        p(e_{\tilde{v}}) \coloneqq e_{\pi(\tilde{v})} 
    \end{equation}
    and extending linearly, i.e. $p(e_{\tilde{v}} + e_{\tilde{v}'}) \coloneqq p(e_{\tilde{v}}) + p(e_{\tilde{v}'})$. We call $p$ the \textit{projection map}. 
\end{dfn}

\begin{rmk}
\label{rmk:projMonomial}
    In our context, this translates to defining $p$ on the monomial basis elements $x^{\tilde{a}}y^{\tilde{b}} \in \tilde{M}$ as
    \begin{equation}
    \label{eq:projMonomial}
        p(x^{\tilde{a}} y^{\tilde{b}}) \coloneqq x^{\text{Mod}(\tilde{a},l)}y^{\text{Mod}(\tilde{b},m)}
    \end{equation}
    and extending linearly, i.e. $p(\tilde{r} + \tilde{r}') \coloneqq p(\tilde{r}) + p(\tilde{r}')$ with $\tilde{r}, \tilde{r}' \in \tilde{M}$. \hfill $\blacksquare$     
\end{rmk}

The significance of this map is that it allows us to define a chain map of two BB codes. This implies that the logical operators of the cover code project to logical operators of the base code.

\begin{thm}
    Let $\tilde{Q}_{\bullet}$ and $Q_{\bullet}$ be the chain complexes as in Remark \ref{rmk:BBcodeDefn} associated to BB codes $\tilde{Q}(\tilde{A}, \tilde{B}, \tilde{l}, \tilde{m})$ and $Q(A, B, l, m)$ with the same properties as in the statement of Theorem \ref{thm:derived-Tanner-iso}. Define the cochain complexes $Q^{\bullet}, \tilde{Q}^{\bullet}$ as the dual of $Q_{\bullet}$ and $\tilde{Q}_{\bullet}$ respectively, as in Remark \ref{rmk:kofCSScode}. Let $p \colon \mathbb{F}_2^{\tilde{l}\tilde{m}} \rightarrow \mathbb{F}_2^{lm}$ be the map in \eqref{eq:projMonomial}.
    
    \begin{enumerate}[label={(\alph*)}]
        \item Define a map $p_{\bullet} \colon \tilde{Q}_{\bullet} \rightarrow Q_{\bullet}$ as follows 
            \begin{equation}
            \label{eq:Z-proj_diag}
                \begin{tikzcd}[row sep=2em, column sep = 5em]
                    \tilde{Q}_2 = \bbF_2^{\tilde{l}\tilde{m}} \arrow[r, "\tilde{H}_Z^T = \begin{pmatrix} \tilde{B} \\ \tilde{A} \end{pmatrix}"] \arrow[d, "p_2"]
                    & \tilde{Q}_1 = \bbF_2^{\tilde{l}\tilde{m}} \oplus \bbF_2^{\tilde{l}\tilde{m}} \arrow[d, "p_1"] \arrow[r, "\tilde{H}_X = \begin{pmatrix} \tilde{A} \ \tilde{B} \end{pmatrix}"]
                    & \tilde{Q}_0 = \bbF_2^{\tilde{l}\tilde{m}} \arrow[d, "p_0"] \\
                    Q_2 = \bbF_2^{lm} \arrow[r, "H_Z^T = \begin{pmatrix} B \\ A \end{pmatrix}"]
                    & Q_1 = \bbF_2^{lm} \oplus \bbF_2^{lm} \arrow[r, "H_X = \begin{pmatrix} A \ B \end{pmatrix}"] 
                    & Q_0 = \bbF_2^{lm}
                \end{tikzcd}
            \end{equation}
            where $p_2 \coloneqq p$, $p_0 \coloneqq p$, $p_1 \coloneqq p \oplus p$, and $p$ is as in Remark \ref{rmk:projMonomial}. Then
            $p_{\bullet}$ is a chain map and induces a map $\hat{p}_1 \coloneqq H_1(p_{\bullet}) \colon H_1(\tilde{Q}_{\bullet}) \rightarrow H_1(Q_{\bullet})$. 
        \item Define a map $p^{\bullet} \colon \tilde{Q}^{\bullet} \rightarrow Q^{\bullet}$ as
        \begin{equation}
        \label{eq:X-proj_diag}
            \begin{tikzcd}[row sep=2em, column sep = 5em]
                \tilde{Q}^0 = \bbF_2^{\tilde{l}\tilde{m}} \arrow[r, "\tilde{H}_X^T = \begin{pmatrix} \tilde{A}^T \\ \tilde{B}^T \end{pmatrix}"] \arrow[d, "p^0"]
                & \tilde{Q}^1 = \bbF_2^{\tilde{l}\tilde{m}} \oplus \bbF_2^{\tilde{l}\tilde{m}} \arrow[d, "p^1"] \arrow[r, "\tilde{H}_Z = \begin{pmatrix} \tilde{B}^T \ \tilde{A}^T \end{pmatrix}"]
                & \tilde{Q}^2 = \bbF_2^{\tilde{l}\tilde{m}} \arrow[d, "p^2"] \\
                Q^0 = \bbF_2^{lm} \arrow[r, "H_X^T = \begin{pmatrix} A^T \\ B^T \end{pmatrix}"]
                & Q^1 = \bbF_2^{lm} \oplus \bbF_2^{lm} \arrow[r, "H_Z = \begin{pmatrix} B^T \ A^T \end{pmatrix}"] 
                & Q^2 = \bbF_2^{lm}
            \end{tikzcd}
        \end{equation}
        where $p^2 \coloneqq p$, $p^0 \coloneqq p$, and $p^1 \coloneqq p \oplus p$. Then $p^{\bullet}$ is a cochain map and induces a map $\hat{p}^1 \coloneqq H^1(p^{\bullet}) \colon H^1(\tilde{Q}^{\bullet}) \rightarrow H^1(Q^{\bullet})$.
    \end{enumerate}
    \label{thm:proj-chain-map}
\end{thm}
\begin{proof}
    \begin{enumerate}[label={(\alph*)}]
        \item We need to show the first diagram above commutes, i.e. that
    \begin{align*}
        p_1 \circ \tilde{H}_Z^T &= H_Z^T \circ p_2 \\
        p_0 \circ \tilde{H}_X &= H_X \circ p_1.
    \end{align*}
    Substituting in the definitions of the above maps, this is equivalent to showing 
    \begin{equation*}
        \begin{pmatrix} p & 0 \\ 0 & p \end{pmatrix}    \begin{pmatrix} \tilde{B} \\ \tilde{A} \end{pmatrix} = \begin{pmatrix} B \\ 
        A \end{pmatrix} p, \hspace{1em} \text{i.e.} \hspace{1em} \begin{pmatrix} p \circ \tilde{B} \\ p \circ \tilde{A} \end{pmatrix} 
         = \begin{pmatrix} B \circ p \\ 
        A \circ p \end{pmatrix}  
    \end{equation*} and
    \begin{equation*}
        p \begin{pmatrix} \tilde{A} & \tilde{B} \end{pmatrix} = \begin{pmatrix} A & B \end{pmatrix} 
        \begin{pmatrix} p & 0 \\ 0 & p \end{pmatrix} 
        , \hspace{1em} \text{i.e.} \hspace{1em} \begin{pmatrix} p \circ \tilde{A} & p \circ \tilde{B} \end{pmatrix} 
         = \begin{pmatrix} A \circ p & B \circ p \end{pmatrix}.  
    \end{equation*}
    It therefore suffices to show that 
    \begin{align}
        p \circ \tilde{A} &= A \circ p \\
        p \circ \tilde{B} &= B \circ p. 
    \end{align} 
    We show the first of these, as the other follows from an identical argument. By linearity of $p$, it suffices to consider the case where 
    \begin{equation*}
        \tilde{A} = x^{\tilde{\alpha}_{1}}y^{\tilde{\alpha}_{2}}, \ A = x^{\alpha_{1}}y^{{\alpha}_{2}} 
    \end{equation*}
    i.e. when the defining polynomials of the codes are single monomials. Without loss of generality, consider an arbitrary monomial $x^{\tilde{a}} y^{\tilde{b}}$ in $\bbF_2^{\tilde{l}\tilde{m}}$. Then
    \begin{equation*}
        (p \circ \tilde{A})(x^{\tilde{a}} y^{\tilde{b}}) = p(x^{\tilde{\alpha}_1 + \tilde{a}} y^{\tilde{\alpha}_2+\tilde{b}}) 
        = 
         x^{\Mod(\tilde{\alpha}_1 + \tilde{a},l)} y^{\Mod(\tilde{\alpha}_2 + \tilde{b},m)}
    \end{equation*}
    On the other hand
    \begin{equation*}
        (A \circ p)(x^{\tilde{a}} y^{\tilde{b}}) = A(x^{\Mod(\tilde{a},l)} y^{\Mod(\tilde{b},m)})
        = x^{\alpha_1 + \Mod(\tilde{a},l)} y^{\alpha_2 + \Mod(\tilde{b},m)}.
    \end{equation*}
    Equality of $(p \circ \tilde{A})(x^{\tilde{a}} y^{\tilde{b}})$ and $(A \circ p)(x^{\tilde{a}} y^{\tilde{b}})$ then follows from
    \begin{align*}
        &\Mod(\tilde{\alpha}_1 + \tilde{a},l) = \Mod(\tilde{\alpha}_1,l) + \Mod(\tilde{a},l) = \alpha_1 + \Mod(\tilde{a},l) \mod l \\
        &\Mod(\tilde{\alpha}_2 + \tilde{a},m) = \Mod(\tilde{\alpha}_2,m) + \Mod(\tilde{a},m) = \alpha_2 + \Mod(\tilde{b},m) \mod m, 
    \end{align*}
    where the last equality follows from our assumptions on the two BB codes. Thus $p_{\bullet}$ is a chain map and, by Remark \ref{rmk:InducedHomMap}, induces a well-defined map $\hat{p}_1 \colon H_1(\tilde{Q}_{\bullet}) \rightarrow H_1(Q_{\bullet})$. 
    \item By a similar reasoning as the previous point, it suffices to show that 
    \begin{align}
        p \circ \tilde{A}^T &= A^T \circ p \\
        p \circ \tilde{B}^T &= B^T \circ p. 
    \end{align} 
    We again show the first of these, as the other follows from an identical argument. Since the transpose operation corresponds to taking the additive inverse of the powers of $x$ and $y$ modulo their respective orders, it suffices to consider
    \begin{equation*}
        \tilde{A}^T = x^{\tilde{l} - \tilde{\alpha}_{1}}y^{\tilde{m} - \tilde{\alpha}_{2}}, \ A^T = x^{l-\alpha_{1}}y^{m-{\alpha}_{2}}. 
    \end{equation*}
    Again Without loss of generality, consider an arbitrary $x^{\tilde{a}} y^{\tilde{b}}$ in $\bbF_2^{\tilde{l}\tilde{m}}$. Then $$(p \circ \tilde{A}^T)(x^{\tilde{a}} y^{\tilde{b}}) = p(x^{\tilde{l} + \tilde{a} - \tilde{\alpha}_1} y^{\tilde{m}+\tilde{b}-\tilde{\alpha}_2}) = x^{\Mod(\tilde{l} + \tilde{a} - \tilde{\alpha}_1,l)}y^{\Mod(\tilde{m}+\tilde{b}-\tilde{\alpha}_2,m)}.$$ On the other hand, $$(A^T \circ p)(x^{\tilde{a}} y^{\tilde{b}}) = A^T(x^{\Mod(\tilde{a},l)}y^{\Mod(\tilde{b},m)}) = x^{l-\alpha_1+\Mod(\tilde{a},l)}y^{m-\alpha_2 + \Mod(\tilde{b},m)}.$$ Equality of $(p \circ \tilde{A}^T)(x^{\tilde{a}} y^{\tilde{b}})$ and $(A^T \circ p)(x^{\tilde{a}} y^{\tilde{b}})$ then follows from 
    \begin{align*}
        &\Mod(\tilde{l} + \tilde{a} - \tilde{\alpha}_1,l) = \Mod(\tilde{l} - \tilde{\alpha}_1,l) + \Mod(\tilde{a},l) = l - \alpha_1 + \Mod(\tilde{a},l) \mod l \\
        &\Mod(\tilde{m} + \tilde{b} - \tilde{\alpha}_2,m) = \Mod(\tilde{m} - \tilde{\alpha}_2,m) + \Mod(\tilde{b},m) = m - \alpha_2 + \Mod(\tilde{b},m) \mod m,
    \end{align*}
    where the last equality follows from our assumptions on the two BB codes. Thus $p^{\bullet}$ is a chain map and, by Remark \ref{rmk:InducedHomMap}, induces a well-defined map $\hat{p}^1 \colon H^1(\tilde{Q}^{\bullet}) \rightarrow H^1(Q^{\bullet})$. \qedhere
    \end{enumerate}
\end{proof}

The first part of the preceding theorem therefore implies that we can project $Z$ logical operators from the cover code to those of the base code; the second part implies that we can project $X$ logical operators from the cover code to those of the base code. 
    
Note that at each degree of the two chain maps, $p_i = p^i$. Since $H^1(\tilde{Q}^{\bullet}) \cong H_1(\tilde{Q}_{\bullet})$ (see Remark \ref{rmk:kofCSScode}) and similarly for $Q$, the matrix representations of $\hat{p}^1$ and $\hat{p}_1$ are equivalent up to a change of basis. The maps induced in (co)homology by $p_{\bullet}$ ($p^{\bullet}$) are not necessarily surjective and we will later show under what conditions this is guaranteed to be the case (see Lemma \ref{lem:PcomposeT}) and the proof of Theorem \ref{thm:hOddLogical}).  

\begin{ex}
\label{ex:proj-logical}
Continuing from Example \ref{ex:grossCodeCover}, a logical basis for the gross code is 
\begin{equation}
\begin{split}
    &\tilde{X}_i = X(r_i f,0),\ \tilde{Z}_i = Z(r'_ig^T,r'_ih^T) \\
    &\tilde{X}'_i = X(r_i g, r_i h), \tilde{Z}'_i = Z(0, r'_if),
\end{split}
\label{eq:n18-log-basis}
\end{equation}
where
\begin{equation}
\begin{split}
    &f = 1 + x + x^2 + x^3 + xy^3 + x^5y^3 + x^6 + x^7 + x^8 + x^9 + x^7y^3 + x^{11}y^3,\\
    &g = y^2 + y^4 + x + xy^2 + x^2y + x^2y^3, \\
    &h = 1 + y + y^2 + y^3 + xy + xy^3.
\end{split}
\label{eq:144-basis-polys}
\end{equation}

The monomials $r_i$ and $r'_j$ are drawn from the ordered sets $$r_i \in \{1, y, x^2y, x^2y^5, x^3y^2, x^4 \}$$ and $$r'_j \in \{ y, y^5, xy, 1, x^4, x^5y^2 \},$$ where $\tilde{X}_i$ and $\tilde{Z}_j$ (and correspondingly $\tilde{X}'_i$ and $\tilde{Z}'_j$) anticommute when $i=j$ and commute otherwise \cite{Bravyi2024BB}. 

When we apply the projection map, we find that $\tilde{X}_i$ and $\tilde{Z}'_j$ project to zero for all $i,j$. For the simplest case of $\tilde{X}_1$ where $r_1 = 1$, the projected logical operator is $X(p_1(f), 0)$ where
\begin{equation}
    p_1(f) = 1 + x + x^2 + x^3 + xy^3 + x^5y^3 + 1 + x + x^2 + x^3 + xy^3 + x^5y^3 = 0.
\end{equation}
While $p_1(f) = f' + f' = 0$, $X(f', 0)$ is in fact a non-trivial logical operator of the base code. Thus the cover logical $\tilde{X}_1$ projects to a sum of two copies of the same base code logical operator. 

The remaining logical operators project to weight 12 non-trivial logical operators of the $[[72,12,6]]$ code. For example, the projection of $\tilde{X}'_1$ is the logical operator $X(p_1(g), p_1(h))$, where it is clear that $p_1(g) = g$ and $p_1(h) = h$. \hfill $\blacksquare$
\end{ex}

\subsection{Lifting Logical Operators from a Base to a Cover Code}
\label{sec:LiftLogical}

While the ability to project logical operators from a cover code to a smaller base code can be useful, more interesting is the possibility of doing the reverse, i.e. lifting logical operators from a base to a cover code. 

\begin{dfn}[\textbf{Lifting Map}]
\label{dfn:liftingMap}
    Let $\pi \colon \tilde{V} \rightarrow V$ from $C(\tilde{V},\tilde{E})$ to $B(V, E)$ be a  graph covering map. Let $e_{v}$ be a basis vector in $\bbF_2^{|V|}$ corresponding to a vertex $v \in V$. Define a linear map $\tau \colon \bbF_2^{|V|} \rightarrow \bbF_2^{|\tilde{V}|}$ by
    \begin{equation}
        \label{eq:tauAbstract}
        \tau(e_v) = \sum_{\tilde{v} \in \pi^{-1}(v)} e_{\tilde{v}} 
    \end{equation} 
    and extending linearly. We call $\tau$ the \textit{lifting map}. 
\end{dfn}

\begin{rmk}
\label{rmk:tauMonomial}
    In our context, this translates to defining $\tau$ on the monomial basis elements $x^a y^b \in M$ as
    \begin{equation}
    \label{eq:tauMonomial}
        \tau(x^{a} y^{b}) \coloneqq \sum_{j=0}^{u-1} \sum_{k=0}^{t-1} x^{a+lj}y^{b+mk}
    \end{equation}
    and extending linearly. \hfill $\blacksquare$
\end{rmk}
As in the case of the projection map, we will show the lifting map can be used to form a (co)chain map from a base BB code to a cover BB code which therefore induces a corresponding map on (co)homology. This allows logical Pauli operators to be lifted from the base code to the cover code. We first note that there is a close relationship between the projection and lifting maps as follows:
\begin{thm}
    Let $p$ be the projection map in Definition \ref{dfn:projectionMap} and $\tau$ the lifting map in Definition \ref{dfn:liftingMap}. Then $\tau$ is the transpose of $p$ and is injective.
\label{thm:lift-transpose}
\end{thm}

\begin{proof}
    For simplicity, define $W = \bbF_2^{|V|}$ and $\tilde{W} = \bbF_2^{|\tilde{V}|}$. The dual of $p$ is the map $p^* \colon W^* \rightarrow \tilde{W}^*$. Assume we have picked a dual basis $f^v$ of $W^*$ and $f^{\tilde{v}}$ of $\tilde{W}^*$ such that $f^v e_{v'} = \delta_{v,v'}$ and $f^{\tilde{v}} e_{\tilde{v}'} = \delta_{\tilde{v}, \tilde{v}'}$. Note that $p^*(f^v) = f^v \circ p$. We can write $f^v \circ p$ as an $\bbF_2$-linear combination of the basis $f^{\tilde{v}}$, i.e.
    \begin{equation*}
        f^v \circ p = \sum_{\tilde{v} \in \tilde{V}} a_{\tilde{v}} f^{\tilde{v}}
    \end{equation*} 
    where $a_{\tilde{v}} \in \bbF_2$. To determine which $a_{\tilde{v}}$ are non-zero, we can apply an arbitrary basis vector $e_{\tilde{v}'}$ to both sides of the equation to get
    \begin{equation*}
        (f^v \circ p)(e_{\tilde{v}'}) = \delta_{v,\pi(\tilde{v}')} = \sum_{\tilde{v} \in \tilde{V}} a_{\tilde{v}} f^{\tilde{v}} e_{\tilde{v}'} = \sum_{\tilde{v} \in \tilde{V}} a_{\tilde{v}} \delta_{\tilde{v},\tilde{v}'} = a_{\tilde{v}'}
    \end{equation*}
    Thus the non-zero $a_{\tilde{v}}$ are those such that $\pi(\tilde{v}') = v$, i.e. those corresponding to the $\tilde{v}'$ that are in the pre-image $\pi^{-1}(v)$. Thus $$p^*(f^{v}) = \sum_{\tilde{v} \in \pi^{-1}(v)} f^{\tilde{v}}.$$ Let $\varphi \colon W \rightarrow W^*$ and $\psi \colon \tilde{W}^* \rightarrow \tilde{W}$ be the basis isomorphisms sending $e_v \mapsto f^v$ and $f^{\tilde{v}} \mapsto e_{\tilde{v}}$ respectively. Then $\tau = \psi \circ p^* \circ \varphi$ and the matrix representing $\tau$ is simply the transpose of the matrix representing $p$. The transpose of a surjective linear map is injective, so $\tau$ is injective.   
\end{proof}

A straightforward consequence of the above result that $p^T = \tau$ is as follows:

\begin{thm}
\label{thm:lift-chain-map}
    Let $\tilde{Q}_{\bullet}$ and $Q_{\bullet}$ be the chain complexes as in Remark \ref{rmk:BBcodeDefn} associated to BB codes $\tilde{Q}(\tilde{A}, \tilde{B}, \tilde{l}, \tilde{m})$ and $Q(A, B, l, m)$ with the same properties as in the statement of Theorem \ref{thm:derived-Tanner-iso}. Define the cochain complexes $Q^{\bullet}, \tilde{Q}^{\bullet}$ as the dual of $Q_{\bullet}$ and $\tilde{Q}_{\bullet}$ respectively, as in Remark \ref{rmk:kofCSScode}. Let $\tau \colon \mathbb{F}_2^{lm} \rightarrow \mathbb{F}_2^{\tilde{l}\tilde{m}}$ be the map in \eqref{eq:tauMonomial}.
    
    \begin{enumerate}[label={(\alph*)}]
        \item Define a map $\tau_{\bullet} \colon Q_{\bullet} \rightarrow \tilde{Q}_{\bullet}$ as follows 
            \begin{equation}
            \label{eq:Z-lift_diag}
                \begin{tikzcd}[row sep=2em, column sep = 5em]
                    \tilde{Q}_2 = \bbF_2^{\tilde{l}\tilde{m}} \arrow[r, "\tilde{H}_Z^T = \begin{pmatrix} \tilde{B} \\ \tilde{A} \end{pmatrix}"] 
                    & \tilde{Q}_1 = \bbF_2^{\tilde{l}\tilde{m}} \oplus \bbF_2^{\tilde{l}\tilde{m}} \arrow[r, "\tilde{H}_X = \begin{pmatrix} \tilde{A} \ \tilde{B} \end{pmatrix}"]
                    & \tilde{Q}_0 = \bbF_2^{\tilde{l}\tilde{m}} \\
                    Q_2 = \bbF_2^{lm} \arrow[r, "H_Z^T = \begin{pmatrix} B \\ A \end{pmatrix}"] \arrow[u,"\tau_2"]
                    & Q_1 = \bbF_2^{lm} \oplus \bbF_2^{lm} \arrow[r, "H_X = \begin{pmatrix} A \ B \end{pmatrix}"] \arrow[u,"\tau_1"]
                    & Q_0 = \bbF_2^{lm} \arrow[u,"\tau_0"]
                \end{tikzcd}
            \end{equation}
            where $\tau_2 \coloneqq \tau$, $\tau_0 \coloneqq \tau$, $\tau_1 \coloneqq \tau \oplus \tau$, and $\tau$ is as in Remark \ref{rmk:tauMonomial}. Then $\tau_{\bullet} = (p^{\bullet})^T$ 
            is a chain map and induces a map $\hat{\tau}_1 \coloneqq H_1(\tau_{\bullet}) \colon H_1(Q_{\bullet}) \rightarrow H_1(\tilde{Q}_{\bullet})$ such that $\hat{\tau}_1 = (\hat{p}^1)^T$. 
        \item Define a map $\tau^{\bullet} \colon Q^{\bullet} \rightarrow \tilde{Q}^{\bullet}$ as
        \begin{equation}
        \label{eq:X-lift-diag}
            \begin{tikzcd}[row sep=2em, column sep = 5em]
                \tilde{Q}^0 = \bbF_2^{\tilde{l}\tilde{m}} \arrow[r, "\tilde{H}_X^T = \begin{pmatrix} \tilde{A}^T \\ \tilde{B}^T \end{pmatrix}"]
                & \tilde{Q}^1 = \bbF_2^{\tilde{l}\tilde{m}} \oplus \bbF_2^{\tilde{l}\tilde{m}} \arrow[r, "\tilde{H}_Z = \begin{pmatrix} \tilde{B}^T \ \tilde{A}^T \end{pmatrix}"]
                & \tilde{Q}^2 = \bbF_2^{\tilde{l}\tilde{m}} \\
                Q^0 = \bbF_2^{lm} \arrow[r, "H_X^T = \begin{pmatrix} A^T \\ B^T \end{pmatrix}"] \arrow[u, "\tau^0"]
                & Q^1 = \bbF_2^{lm} \oplus \bbF_2^{lm} \arrow[r, "H_Z = \begin{pmatrix} B^T \ A^T \end{pmatrix}"] \arrow[u, "\tau^1"] 
                & Q^2 = \bbF_2^{lm} \arrow[u, "\tau^2"] 
            \end{tikzcd}
        \end{equation}
        where $\tau^2 \coloneqq \tau$, $\tau^0 \coloneqq \tau$, and $\tau^1 \coloneqq \tau \oplus \tau$. Then $\tau^{\bullet} = (p_{\bullet})^T$ is a cochain map and induces a map $\hat{\tau}^1 \coloneqq H^1(\tau^{\bullet}) \colon H^1(Q^{\bullet}) \rightarrow H^1(\tilde{Q}^{\bullet})$ such that $\hat{\tau}^1 = (\hat{p}_1)^T$. 
    \end{enumerate}
\end{thm}

\begin{proof}
    \begin{enumerate}[label={(\alph*)}]
        \item Note that this diagram is simply the dual of the diagram in \eqref{eq:X-proj_diag}, so $\tau_{\bullet} = (p^{\bullet})^T$. The commutativity of this diagram therefore follows from taking the transpose of the commutativity conditions for $p^{\bullet}$. The duality map on vector spaces preserves cokernels and kernels and therefore commutes with (co)homology.\footnote{More formally, the duality functor is an exact functor on vector spaces and exact functors commute with (co)homology.} Thus $H_1(\tau_{\bullet}) = H^1((p^{\bullet})^T) = (H^1(p^{\bullet}))^T$.
        \item The same proof as above holds when applied to taking the dual of the diagram in \eqref{eq:Z-proj_diag}. \qedhere  
    \end{enumerate} 
\end{proof}

Theorem \ref{thm:lift-chain-map} implies that we can lift logical Pauli $Z$ and $X$ operators from a base code to a cover code. This is useful in finding logical operators of particularly large codes if the logical operators of a smaller, base code are known. At each degree of the two chain maps, $\tau_i = \tau^i$. Since $H^1(\tilde{Q}^{\bullet}) \cong H_1(\tilde{Q}_{\bullet})$ (see Remark \ref{rmk:kofCSScode}) and similarly for $Q$, the matrix representations of $\hat{\tau}^1$ and $\hat{\tau}_1$ are equivalent up to a change of basis. 

Note however that the induced maps in homology $\hat{\tau}_1$ and $\hat{\tau}^1$ are generally not guaranteed to be surjective or injective. This implies that even if a logical basis is known for the base code, it may only lift to a set of logical operators for the cover code that do not span the logical space in the cover code, i.e. the number of linearly independent lifted logical operators may be less than $k_h$. Lifting logical operators from a base code may not give minimum weight logical operators in the cover code either. This is because the definition of the lifting map implies the weight of the lifted logical operator dilates by a factor of $h$ while, as we will see in the next section, the distance may not. We will give conditions in the next section for when the map induced in (co)homology by $\tau_{\bullet}$ ($\tau^{\bullet}$) is injective (see Lemma \ref{lem:PcomposeT}) and the proof of Theorem \ref{thm:hOddLogical}). 

\begin{ex}
    Consider the $[[72,12,6]]$ code $$Q(x^3 + y + y^2, \ y^3 + x + x^2, \ 6, \ 6)$$ and let
    \begin{equation}
        f = 1 + x + x^2 + x^3 + xy^3 + x^5y^3.
    \label{eq:lift-ex1-f-poly}
    \end{equation}
    Then $X(rf,0)$ is a logical operator of the code for all $r$ in the base code's monomial basis $M$. Each of these lifts to a logical operator of the $[[144,12,12]]$ code $$\tilde{Q}(x^3 + y + y^2, \ y^3 + x + x^2, \ 12, 6).$$ For simplicity, we consider only the $r=1$ case here. Applying the lifting map in \eqref{eq:tauMonomial} to the polynomial in \eqref{eq:lift-ex1-f-poly} gives 
    \begin{equation}
        \tau_1(f) = 1 + x + x^2 + x^3 + xy^3 + x^5y^3 + x^6 + x^7 + x^8 + x^9 + x^7y^3 + x^{11}y^5,
    \end{equation}
    which is exactly the polynomial given in the original BB code paper \cite{Bravyi2024BB} and also in \eqref{eq:144-basis-polys} in Example \ref{ex:proj-logical}. Thus $X(\tau_1(f), 0)$ is a logical operator of the $[[144,12,12]]$ code. 
    
    When taking base code logical operators $X_i = X(r_i f, 0)$ with $r_i \in \{1, y, x^2y, x^2y^5, x^3y^2, x^4 \}$, each $X_i$ lifts to exactly the $\tilde{X}_i$ logical operators of the gross code given in Example \ref{ex:proj-logical}. These are the gross code $X$ logical operators that project to zero. \hfill $\blacksquare$
\end{ex} 

We can use the induced projection and lifting maps in (co)homology of Theorems \ref{thm:proj-chain-map} and \ref{thm:lift-chain-map} to show that in certain cases, the number of logical qubits in the cover code must be at least that of the base code. 

We first show the following simple lemma: 

\begin{lem}
\label{lem:PcomposeT}
    Let $Q(A,B,l,m)$ and $\tilde{Q}(\tilde{A},\tilde{B},\tilde{l},\tilde{m})$ be two BB codes satisfying the conditions of Theorem \ref{thm:derived-Tanner-iso} so that $\tilde{Q}$ is an $h$-cover code of $Q$. Consider the projection map $p$ in \eqref{eq:projMonomial} and lifting map $\tau$ in \eqref{eq:tauMonomial}. When $h$ is odd, $p \circ \tau = I$. When $h$ is even, $p \circ \tau = 0$. 
\end{lem}

\begin{proof}
    By definition, we have $$(p \circ \tau)(x^a y^b) = p\left(\sum_{j=0}^{u-1} \sum_{k=0}^{t-1} x^{a+lj}y^{b+mk}\right) = \sum_{j=0}^{u-1} \sum_{k=0}^{t-1} p(x^{a+lj}y^{b+mk}) = \sum_{j=0}^{u-1} \sum_{k=0}^{t-1} x^{\Mod(a+lj,l)}y^{\Mod(b+mk,m)}$$ Since $\Mod(a+lj,l) = a$ as $a < l$ and $\Mod(b+mk,m) = m$ as $b < m$, we have $$(p \circ \tau)(x^a y^b) = \sum_{j=0}^{u-1} \sum_{k=0}^{t-1} x^{\Mod(a+lj,l)}y^{\Mod(b+mk,m)} = \sum_{j=0}^{u-1} \sum_{k=0}^{t-1} x^a y^b = ut x^a y^b = h x^a y^b$$ Since we are working over $\bbF_2$, $(p \circ \tau)(x^a y^b) = x^a y^b$ when $h$ is odd and $(p \circ \tau)(x^a y^b) = 0$ when $h$ is even.  
\end{proof}

\begin{thm}
\label{thm:hOddLogical}
    With the same conditions as in the preceding lemma, if $h$ is odd then $k_h \geq k$. 
\end{thm}

\begin{proof}
From the preceding lemma, it follows that $p_{\bullet} \circ \tau_{\bullet} = I_{\bullet}$ when $h$ is odd, where $p_{\bullet}$ and $\tau_{\bullet}$ are the chain maps in \eqref{eq:Z-proj_diag} and $\eqref{eq:Z-lift_diag}$. Applying homology as in Remark \ref{rmk:InducedHomMap} gives that $H_1(p_{\bullet}) \circ H_1(\tau_{\bullet}) = I$. 

Because $H_1(p_{\bullet}) \colon H_1(\tilde{Q}_{\bullet}) \rightarrow H_1(Q_{\bullet})$ and $H_1(\tau_{\bullet}) \colon H_1(Q_{\bullet}) \rightarrow H_1(\tilde{Q}_{\bullet})$ are simply linear maps over $\bbF_2$, this immediately implies that $H_1(p_{\bullet})$ is surjective because it has a right-inverse and $H_1(\tau_{\bullet})$ is injective because it has a left-inverse. Surjectivity of $H_1(p_{\bullet})$ (or injectivity of $H_1(\tau_{\bullet})$) as a linear map immediately implies $\dim H_1(\tilde{Q}_{\bullet}) = k_h \geq \dim H_1(Q_{\bullet}) = k$. 
\end{proof}

A consequence of Theorem \ref{thm:hOddLogical} is that when $h$ is odd and $k_h = k$, then $\hat{\tau}_1$ is in fact an isomorphism and a complete logical basis for the base code will lift to a complete logical basis for the cover code. Below we show an example of a case where a minimum weight logical basis does lift to a minimum weight basis for the cover code (as noted before, this is not guaranteed in general).

\begin{ex}
\label{ex:logical-lift-2}
The $[[54,8,6]]$ code $$\tilde{Q}(x^3+y+y^2, \ 1+x+x^2, \ 6, \ 3)$$ is a triple cover of the $[[18,8,2]]$ code $$Q(1+y+y^2, \ 1+x+x^2, \ 3, \ 3).$$ The base code has a logical basis given by
\begin{equation}
\begin{split}
    &X_i = X(r_i f,0), Z_i = Z(r'_ig,0) \\
    &X'_i = X(0, r'_ig), Z'_i = Z(0, r_if),
\end{split}
\end{equation}
where $f = 1+x$ and $g = 1+y$. Anticommuting pairs of logical operators are given by using monomials from the ordered sets $r_i \in \{ 1, y^2, x, xy^2 \}$ and $r'_j \in \{ 1, y, x^2, x^2y\}$, where the operators $X_i, Z_j$ (or $X'_i, Z'_j$) anticommute for $r_i, r'_j$ when $i=j$ and commute otherwise.

The corresponding lifted basis is given by,
\begin{equation}
\begin{split}
    &\tilde{X}_i = X(\tau_1(r_i f),0), \tilde{Z}_i = Z(\tau_1(r'_i g),0) \\
    &\tilde{X}'_i = X(0, \tau_1(r'_i g)), \tilde{Z}'_i = Z(0, \tau_1(r_i f)),
\end{split}
\label{eq:n18-lift-log-basis}
\end{equation}
where $r_i, r'_j$ are drawn from the same sets as above. We don't enumerate all of the lifted polynomials here but consider $r_1=1$ and $r'_1 =1$ for example. Then $$\tau_1(r_1 f) = 1 + x + x^3 + x^4 + x^6 + x^7$$ and $$\tau_1(r'_1 g) = 1 + y + x^3 + x^3y + x^6 + x^6y.$$ It is straightforward to check that the lifted logical operators $\tilde{X}_1$ and $\tilde{Z}_1$ anticommute with these choices of polynomials. \hfill $\blacksquare$
\end{ex}

We will now strengthen the result of the preceding theorem and show that $k_h \geq k$ for \textit{any} $h$. 

\begin{thm}
\label{thm:kboundanyh}
    Let $Q(A,B,l,m)$ and $\tilde{Q}(\tilde{A},\tilde{B},\tilde{l},\tilde{m})$ be two BB codes satisfying the conditions of Theorem \ref{thm:derived-Tanner-iso} so that $\tilde{Q}$ is an $h$-cover code of $Q$. Then $k_h \geq k$.
\end{thm}

\begin{proof}
    Let $R$ be the ring in \eqref{eq:bbcodeAlgisos} and $\tilde{R}$ the ring $$\tilde{R} = \frac{\bbF_2[x,y]}{(x^{\tilde{l}} - 1, y^{\tilde{m}} -1)}.$$ Consider the surjective map $p$ as defined in \eqref{eq:projMonomial}. It in fact extends to a surjective ring homomorphism $p \colon \tilde{R} \rightarrow R$ since $$p(1) = p(x^0 y^0) = x^{\text{Mod}(0,l)}y^{\text{Mod}(0,m)} = 1$$ and $$p(x^{\tilde{a}} y^{\tilde{b}}) = x^{\text{Mod}(\tilde{a},l)}y^{\text{Mod}(\tilde{b},m)} = p(x) p(y).$$ From \cite{Bravyi2024BB}, $k = 2(lm - (\mathrm{rank}(H_X))$ for any BB code but note that $$lm - \mathrm{rank}(H_X) = \dim_{\bbF_2} \frac{R}{(A,B)}.$$ Our assumptions on $A,B$ and $\tilde{A}, \tilde{B}$ from Theorem \ref{thm:derived-Tanner-iso} imply that $p((\tilde{A}, \tilde{B})) \subseteq (A,B)$, so $p$ descends to a surjective ring homomorphism $$p \colon \frac{\tilde{R}}{(\tilde{A},\tilde{B})} \rightarrow \frac{R}{(A,B)}.$$ This immediately implies that $$\dim_{\bbF_2} \frac{\tilde{R}}{(\tilde{A},\tilde{B})} \geq \dim_{\bbF_2} \frac{R}{(A,B)},$$ which in turn implies that $k_h \geq k$. 
\end{proof}

\subsection{Bounds on Distance by Lifting and Projecting Logical Operators}
\label{sec:BoundDist}

The induced lifting map $\hat{\tau}_1$ on homology also gives a way to upper bound the distance $d_h$ of an $h$-cover code when $h$ is odd. 

\begin{thm}
\label{thm:hOddDistUpperBound}
    Let $Q(A,B,l,m)$ and $\tilde{Q}(\tilde{A},\tilde{B},\tilde{l},\tilde{m})$ be two BB codes with parameters $[[n,k,d]]$ and $[[n_h=hn,k_h,d_h]]$ respectively that satisfy the conditions of Theorem \ref{thm:derived-Tanner-iso} so that $\tilde{Q}$ is an $h$-cover code of $Q$. If $h$ is odd, $d_h \leq hd$.
\end{thm}

\begin{proof}
    Let $L$ be a weight $d$ representative of a non-zero homology class $[L] \in H_1(Q_{\bullet})$ (recall that for BB codes, $d = d_X = d_Z$; see Remark \ref{rmk:BBcodeDefn}). Then $\tau_1(L) \in [\tau_1(L)] = \hat{\tau}_1([L])$ by Remark \ref{rmk:InducedHomMap}. Since $h$ is odd, $\hat{\tau}_1$ is injective by the proof of Theorem \ref{thm:hOddLogical}. Thus $[\tau_1(L)] \neq 0$, so $\tau_1(L)$ represents a non-trivial logical operator in $\tilde{Q}$. The definition of $\tau_1$ then implies $\tau_1(L)$ has weight $hd$. 
\end{proof}
Empirically, we always observe that the distance of the cover code obeys $d_h \leq hd$  regardless of $h$. As $h$ increases, we find that $d_h$ rarely saturates the upper bound of $hd$. When $k_h = k$ and $h$ is odd though, we can in fact prove a distance lower bound. 

\begin{thm}
\label{thm:hOddDistBounds}
    With the same conditions as in the preceding theorem, suppose $h$ is odd and $k_h = k$. Then $d \leq d_h \leq hd$.
\end{thm}

\begin{proof}
    Let $L'$ denote a representative of a non-trivial class $[L'] \in H_1(\tilde{Q}_{\bullet})$ such that $|L'| < d$. The definition of $p_1$ implies that $|p_1(L')| \leq |L'| < d$, so $p_1(L')$ is a $Z$ stabilizer in $Q_{\bullet}$ (see the diagram in \eqref{eq:Z-proj_diag}). Thus there exists a $z \in Q_2$ such that $H_Z^T(z) = p_1(L')$. Then $[H_Z^T(z)] = [p_1(L')] = \hat{p}_1[L'] = [0]$. Since $\hat{p}_1$ is surjective and $k_h = k$, $\hat{p}_1$ is also injective since a surjective linear map between vector spaces of the same dimension is an isomorphism. This implies that $[L'] = 0$. Thus there exists a $\tilde{z}' \in \tilde{Q}_2$ such that $L' = \tilde{H}_Z^T(\tilde{z}')$, so $L'$ is in fact a $Z$ stabilizer of $\tilde{Q}_{\bullet}$. This implies $d \leq d_h$. 
\end{proof}

When attempting to generalize these arguments to situations where $h$ is even or $k \neq k_h$, we arrive at points in the proof where the injectivity or surjectivity of $p_i$ or $\tau_i$ respectively or their corresponding induced maps in homology are needed. These are clearly not guaranteed in general. While it may be possible to prove general upper and lower bounds for these cover codes, this may require either more specific information about the codes or other kinds of arguments entirely. We therefore conjecture the following: 

\begin{conj}
    \textit{All} $h$\textit{-cover BB codes obey} $d \leq d_h \leq hd$.
\end{conj}

We've used a definition of the lifting map so far that naturally forms a chain map. There are other possible definitions that are useful however. 

\begin{dfn}[\textbf{Section of a Graph Cover}]
    Let $\pi \colon \tilde{V} \rightarrow V$ from $C(\tilde{V},\tilde{E})$ to $B(V, E)$ be a  graph covering map. A \textit{section} of a graph cover is a map $s \colon V \rightarrow \tilde{V}$ such that $\pi \circ s = I$. 
\end{dfn}
A section is essentially defined by selecting for each $v \in V$, a single element $\tilde{v}$ in the pre-image $\pi^{-1}(v)$. The definition above implies that $s$ is necessarily injective. 

\begin{dfn}[\textbf{Weight-Preserving Lift}]
    Let $s \colon V \rightarrow \tilde{V}$, be a section of a graph cover $\pi \colon \tilde{V} \rightarrow V$ and $e_v$ be a basis vector in $\bbF_2^{|V|}$ corresponding to a vertex $v \in V$. Define an injective linear map $\sigma \colon \bbF_2^{|V|} \rightarrow \bbF_2^{|\tilde{V}|}$ by
    \begin{equation}
        \sigma(e_{v}) \coloneqq e_{s(v)}.
    \end{equation}
    and extending linearly. We call $\sigma$ a \textit{weight-preserving lift}.
\end{dfn}

\begin{rmk}
\label{rmk:WgtPreserveMonom}
    This translates to defining $\sigma$ on the monomial basis elements $x^ay^b \in M$ by
    \begin{equation}
        \sigma(x^a y^b) = x^{a+l \gamma_a} y^{b + m\delta_b}
    \end{equation}
    and extending linearly. Both $\gamma_a \in \bbZ_u$ and $\delta_b \in \bbZ_t$ can depend on $a$ and $b$. The condition that $\pi \circ s = I$ also implies that $p \circ \sigma  = I$ as linear maps,\footnote{This follows from the fact that the map sending a finite set $S$ to the free-vector space $\bbF_2^{|S|}$ is a functor and therefore preserves compositions of maps.} where $p$ is as in \ref{eq:projMonomial}. \hfill $\blacksquare$
\end{rmk}

Intuitively, whereas the map $\tau$ sends a given monomial in the base code to a sum over all its preimages in the cover code, $\sigma$ sends it to just one of its preimages residing on a single sheet in the cover. 

As before, we can use $\sigma$ above to possibly define a (co)chain map between a base and a cover code. While it was straightforward to show that the lifting map $\tau$ did in fact form a (co)chain map and therefore induced a map between the (co)homologies of the base and cover codes, this is no longer necessarily the case for the weight-preserving lift $\sigma$. This is because the action of $\sigma$ on input monomials need not act in a consistent way across different input monomials, as noted in the above remark. Given codes $Q_{\bullet}$ and $\tilde{Q}_{\bullet}$ or their duals, if there is a choice of possibly different lifts $\sigma$ at each degree such that they form a chain map, then we have the following distance bound. 

\begin{thm}
 Let $\tilde{Q}_{\bullet}$ and $Q_{\bullet}$ be the chain complexes as in Remark \ref{rmk:BBcodeDefn} associated to BB codes $\tilde{Q}(\tilde{A}, \tilde{B}, \tilde{l}, \tilde{m})$ and $Q(A, B, l, m)$ with parameters $[[n_h=hn,k_h,d_h]]$ and $[[n,k,d]]$  respectively  and with the same properties as in the statement of Theorem \ref{thm:derived-Tanner-iso}. Suppose there is a chain map $\sigma_{\bullet} \colon Q_{\bullet} \rightarrow \tilde{Q}_{\bullet}$
     \begin{equation}
            \label{eq:singlesheet-lift_diag}
                \begin{tikzcd}[row sep=2em, column sep = 5em]
                    \tilde{Q}_2 = \bbF_2^{\tilde{l}\tilde{m}} \arrow[r, "\tilde{H}_Z^T = \begin{pmatrix} \tilde{B} \\ \tilde{A} \end{pmatrix}"] 
                    & \tilde{Q}_1 = \bbF_2^{\tilde{l}\tilde{m}} \oplus \bbF_2^{\tilde{l}\tilde{m}} \arrow[r, "\tilde{H}_X = \begin{pmatrix} \tilde{A} \ \tilde{B} \end{pmatrix}"]
                    & \tilde{Q}_0 = \bbF_2^{\tilde{l}\tilde{m}} \\
                    Q_2 = \bbF_2^{lm} \arrow[r, "H_Z^T = \begin{pmatrix} B \\ A \end{pmatrix}"] \arrow[u,"\sigma_2"]
                    & Q_1 = \bbF_2^{lm} \oplus \bbF_2^{lm} \arrow[r, "H_X = \begin{pmatrix} A \ B \end{pmatrix}"] \arrow[u,"\sigma_1"]
                    & Q_0 = \bbF_2^{lm} \arrow[u,"\sigma_0"]
                \end{tikzcd}
            \end{equation}
            where $\sigma_1 \coloneqq \sigma' \oplus \sigma'$ and $\sigma_1, \sigma', \sigma_0$ need not be the same maps. Then $d_h \leq d$.
\end{thm}

\begin{proof}
Under the assumptions in the theorem statement, there is an induced map on homology,
\begin{equation}
	\hat{\sigma}_1 \coloneqq H_1(\sigma_{\bullet}) \colon H_1(Q_{\bullet}) \rightarrow H_1(\tilde{Q}_{\bullet}).
\end{equation}
Let $L$ be a weight $d$ representative of a non-trivial class $[L] \in H_1(Q)$. From Remark \ref{rmk:WgtPreserveMonom}, $p_{\bullet} \circ \sigma_{\bullet} = I$, from which it follows that $\hat{\sigma}_1$ is injective and $\hat{\sigma}_1[L] = [\sigma_1(L)] \neq 0$. Thus $\sigma_1(L)$ is a non-trivial logical operator of the cover code $\tilde{Q}$. By definition, $|\sigma(L)| = |L| = d$, which implies $d_h \leq d$.  
\end{proof}

This result can be used to numerically determine instances of cover codes where distance \textit{does not} increase relative to the base code regardless of $h$. In practice, there are combinatorially many ways to choose a weight-preserving lift. As there are $h$ elements in the fibre $\pi^{-1}(v)$ for each vertex $v$ in the base Tanner graph, there are a total of $h^d$ possible choices for a weight $d$ logical operator. While this scales unfavourably, it often scales better than a brute force search for a weight $d$, length $n$ vector that might be a logical operator which scales as $\binom{n}{d}$. For example, searching over all possible weight 6 candidate logical operators for an $n=144$ double cover of the $[[72,12,6]]$ code would require searching $\binom{144}{6}$ vectors. Searching over all weight-preserving lifts of a weight 6 base logical operator requires 
only $h^d = 2^6 = 64$ lifts by contrast. 

Checking the weight-preserving lifts of minimum weight base code logical operators is therefore a potentially more efficient method of finding a weight $d$ logical operator in the cover code and thus upper bounding $d_h \leq d$. \textit{Failure to find a weight-preserving lift does not in general rule out that $d_h \leq d$ however}. The checking of weight-preserving lifts should be treated as a numerical tool for efficiently ruling out low distance codes as opposed to a method for guaranteeing that distance has increased relative to the base code. 

Lastly, we note that while it is possible in principle to derive conditions on each $\sigma_i$ in \eqref{eq:singlesheet-lift_diag} that suffice to make the diagram commute, the fact that any $\sigma$ depends on an arbitrary choice of a cover code monomial in the pre-image of each base code monomial makes this theoretically unfeasible except for trivial choices of $\sigma_i$, e.g. $\sigma(x^a y^b) = x^a y^b$ or $\sigma(x^a y^b) = x^{a+lc_1} x^{b + mc_2}$ where $c_1 \in \bbZ_u, c_2 \in \bbZ_t$ are constants. 

\begin{rmk}
    The BB codes inherit logical operators from their underlying classical codes. More specifically, let $Q(A,B,l,m)$ be a BB code and suppose it has a logical operator defined entirely on either the left or right qubits, e.g. $Z(t,0)$ for some polynomial $t$. Since $Z(t,0)$ is in $\text{ker } H_X$, it follows that $At = 0$, i.e. that $t \in \text{ker } A$. We can therefore treat $A$ as the parity check matrix of a classical code,
    \begin{equation}
        C_{\bullet} \colon C_1 = \bbF_2^{lm} \xrightarrow[]{\partial_1=A} C_0 = \bbF_2^{lm}.
    \end{equation}
    Similarly, if $Z(0,q)$ is a logical operator of the BB code, then it can be seen as having been inherited from a classical code with $B$ as a parity check matrix. 
    
    There are instances where only a subset of the base code's logical operators lift by a weight-preserving lift to cover code logical operators. Often, but not always, this subset is exactly the subset of logical operators that are inherited from the classical codes. In these cases, it appears the weight-preserving lift does not form a chain map as in \eqref{eq:singlesheet-lift_diag} but instead forms a chain map $\sigma_{\bullet} \colon C_{\bullet} \rightarrow \tilde{C}_{\bullet}$
    \begin{equation}
    	\label{eq:classical_sheet_lift_diag}
    	\begin{tikzcd}
    		\tilde{C}_1 = \bbF_2^{\tilde{l}\tilde{m}} \arrow[r, "\tilde{\partial}_1 = \tilde{A}"]
    		& \tilde{C}_0 = \bbF_2^{lm} \\
    		  C_1 = \bbF_2^{lm}  \arrow[u, "\sigma_{1}"] \arrow[r, "\partial_1 = A"] 
    		& C_0 = \bbF_2^{lm} \arrow[u, "\sigma_{0}"]
    	\end{tikzcd}
    \end{equation}
    \end{rmk}
of the relevant classical codes. \hfill $\blacksquare$
\newpage
\begin{ex}
There is a weight-preserving lift for the single block logical operators between the $[[18,8,2]]$ and $[[36,8,2]]$ codes of Example \ref{ex:18-8-2-code-ex}. The logical operators $X_r = X(0, r(1+y))$ with $r \in \{ 1, y, x^2, x^{2}y\}$ are half of the base code's logical operators. These logical operators satisfy $A^T(r(1+y)) = 0$ and can therefore be viewed as logical operators of the classical code with check matrix $A^T$. 

This set of logical operators lifts to cover code logical operators under a chain map with weight-preserving lifts $s_1(x^ay^b) = s_0(x^ay^b) = x^ay^b$. It is straightforward to check that $\tilde{A}^T(r(1+y)) = 0$. Interestingly, the remaining base code's logical $X$ operators fail to lift under any weight-preserving lifting map and the cover code has a mix of weight 2 and weight 4 logical operators. \hfill $\blacksquare$
\end{ex}

\subsection{Generalization to Qudits}
\label{sec:Qudits}

While we have focused on qubit codes, there has been recent work \cite{Spencer2026quditldpc} showing that the BB codes can be defined on qudits. The main algebraic difference is that the underlying field is now $\mathbb{F}_q$, where $q = p^n$ for some $n$ and $p$ is a prime number, rather than $\mathbb{F}_2$. We refer readers to \cite{Spencer2026quditldpc} for the construction of such codes.  

We note that the lifting and projection maps in Definitions \ref{dfn:projectionMap} and \ref{dfn:liftingMap} extend straightforwardly to this setting as maps that are, in particular, $\bbF_q$-linear. The main subtlety concerns the generalization of Lemma \ref{lem:PcomposeT}, since its proof relies on the characteristic of $\bbF_2$ in an essential way. The result is as follows:

\begin{lem}
\label{lem:PcomposeTinFQ}
    Let $Q(A,B,l,m)$ and $\tilde{Q}(\tilde{A},\tilde{B},\tilde{l},\tilde{m})$ be two BB codes satisfying the conditions of Theorem \ref{thm:derived-Tanner-iso} so that $\tilde{Q}$ is an $h$-cover code of $Q$. Consider the projection map $p$ in \eqref{eq:projMonomial} and lifting map $\tau$ in \eqref{eq:tauMonomial} as maps that are, in particular, $\bbF_q$-linear. When $h \neq 0 \text{ mod } p$, $p \circ \tau = hI$. When $h = 0 \text{ mod } p$, $p \circ \tau = 0$. 
\end{lem}

\begin{proof}
    The same argument as in Lemma \ref{lem:PcomposeT} establishes that $p \circ \tau = hI$ in general. Considering $h$ modulo $p$ yields the two cases we need. 
\end{proof}

The maps induced in (co)homology by $p_{\bullet}$ and $\tau_{\bullet}$ have the follow property:

\begin{thm}
\label{thm:h0ModpTauandProj}
    With the same conditions as in the preceding theorem, suppose $h \neq 0 \text{ mod } p$. Then the induced maps in first degree (co)homology $\hat{p}_1$ $(\hat{p}^1)$ and $\hat{\tau}_1$ $(\hat{\tau}^1)$ are surjective and injective respectively. 
\end{thm}

\begin{proof}
    We focus on the induced maps in homology, since the cohomology case follows from an identical argument. 

    To prove that $\hat{\tau}_1$ is injective, let $[L]$ be a logical in the base code and suppose that $\hat{\tau}_1([L]) = [0]$. Then applying $\hat{p}_1$ to both sides and using the previous lemma, we get 
    \begin{align}
        &(\hat{p}_1 \circ \hat{\tau}_1)([L]) = \hat{p}_1([0]) \\
        \implies &[p(\tau(L))] = [0] \\
        \implies &[hL] = h[L] = [0].
    \end{align}
    Since $h \neq 0 \text{ mod } p$, it follows that $[L] = [0]$. 
    
    To prove that $\hat{p}_1$ is surjective, let $[L]$ be a logical of the base code and consider $h^{-1} \hat{\tau}_1([L])$, where $h^{-1}$ denotes the unique multiplicative inverse of $h$ that exists in $\bbF_q$ (specifically in its prime subfield $\bbF_p$) by our assumption that $h \neq 0 \text{ mod } p$. Then by the previous lemma, $$(\hat{p}_1(h^{-1}\hat{\tau}_1([L]))) = h^{-1}[p(\tau(L))] = [L],$$ showing that $\hat{p}_1$ is surjective. 
\end{proof}

An analogous argument holds for all other degrees as well. 

Since our previous bounds on distance depend only these properties of the induced maps in (co)homology when $h \neq 0 \text{ mod } p$, they will also apply to qudit BB codes. The case where $h = 0 \text{ mod } p$ remains difficult to prove concrete distance results about, but these are now a smaller proportion of all qudit cover codes compared to the $\bbF_2$ case. We conjecture however that our proposed distance bounds hold here as well. 

\begin{conj}
    \textit{Any qudit} $h$\textit{-cover BB code obeys} $d \leq d_h \leq hd$. 
\end{conj}

On the other hand, the argument in Theorem \ref{thm:kboundanyh} that $k_h \geq k$ for any $h$ immediately carries over to the qudit setting when replacing $\bbF_2$ with $\bbF_q$, since the equality $k = n - 2(\text{rank}(H_X))$ also holds over other fields as shown in \cite{Spencer2026quditldpc}. 

\section{Graph and Code Automorphisms}
\label{sec:GraphCodeAuts}

We now discuss a series of increasingly general notions of graph and code transformations that preserve our covering graph structures. In Section \ref{sec:Deck}, we discuss the notions of a deck transformation and the more general notion of the lifting of a graph automorphism from a base to a cover graph. We extend these notions to chain complexes in Section \ref{sec:CSScodeAuts} and give examples of how they can be used to lift the actions of logical CNOT and Hadamard operators from a base to a cover code. 

\subsection{Lifts of Graph Automorphisms}
\label{sec:Deck}

\begin{dfn}[\textbf{Deck Transformation}]
    Let $\pi \colon \tilde{V} \rightarrow V$ be a covering map between graphs $C(\tilde{V},\tilde{E})$ and $B(V,E)$. A deck transformation of $\pi$ is a graph isomorphism $\varphi \colon \tilde{V} \rightarrow \tilde{V}$ such that $\pi \circ \varphi = \pi$. 
\end{dfn}

The deck transformations of a covering map form a group under composition and act `vertically' on fibers of the covering graph, i.e. they permute elements within a given fiber as opposed to between fibers.  

Consider in particular the derived graph $G(V',E')$ of Theorem \ref{thm:derived-Tanner-iso} with the canonical projection $c$ to the base Tanner graph $T$ of Theorem \ref{thm:graphProjectionThm}. It can be shown that $\operatorname{Deck}(c) \cong \Gamma$. An element $\gamma \in \Gamma$ corresponds to sending all vertices $v' = (v, g) \mapsto (v, \gamma + g)$ and edges $e'=[(u,g), (v, g + \nu(e))] \mapsto [(u, \gamma+g), (v, \gamma + g + \nu(e))]$. It then follows that deck transformations of $G$ correspond to multiplying the defining polynomials $\tilde{A}$ and $\tilde{B}$ of a BB cover code by the \textit{same} monomial $x^{il}y^{jm}$, where $0 \leq i < u$ and $0 \leq j < t$. Since a deck transformation is a graph isomorphism by definition, a cover code with new defining polynomials of the form $x^{il}y^{jm}\tilde{A}, \ x^{il}y^{jm}\tilde{B}$ gives rise to identical Tanner graphs and therefore identical BB codes.

It is known though that multiplying the defining polynomials $A,B$ of a BB code $Q$ by \textit{different} monomials, for example $x^ay^bA$ and $x^cy^dB$ where $a \neq c$ and $b \neq d$, also gives a code with the same parameters \cite{Liang2025gtc}. These transformations are more general than deck transformations as they need not act vertically on fibers. In order for a cover code to remain a cover code under such a transformation, $\tilde{A}$ and $\tilde{B}$ can only be multiplied by monomials that respect the projection condition of Theorem \ref{thm:derived-Tanner-iso}, that is $p(x^ay^b\tilde{A}) = A$ and $p(x^cy^d\tilde{B}) = B$. 

This motivates the following generalization of a deck transformation:
\begin{dfn}[\textbf{Lift of a Graph Automorphism}]
    Let $T$ be a graph, $(\tilde{T}, p)$ a covering graph, and $\sigma \in \operatorname{Aut}(T)$. Then $\tilde{\sigma} \in \operatorname{Aut}(\tilde{T})$ is a lift of $\sigma$ if the following holds
    \begin{equation}
        \sigma \circ p = p \circ \tilde{\sigma}.
    \label{eq:aut-lift}
    \end{equation}    
\end{dfn}

As a particular example of this, let $T = (V,E)$ be the Tanner graph of a BB code $Q$ and $\tilde{T}$ the Tanner graph of a cover BB code $\tilde{Q}$. Suppose there is a Tanner graph automorphism $\sigma \in \operatorname{Aut}(T)$ such that $\sigma$ swaps qubits $a,b \in V$, i.e. $\sigma(a) = b$ and $\sigma(b) = a$. We denote the elements in the fibers of $a$ and $b$ as $\tilde{a}_i \in p^{-1}(a)$ and $\tilde{b}_j \in p^{-1}(b)$ respectively. Then an example of an automorphism $\tilde{\sigma}$ of $\tilde{T}$ which will satisfy Equation \ref{eq:aut-lift} is a fiber permutation which permutes the fibers above the base elements that are swapped, i.e. $\tilde{\sigma}(\tilde{a}_i) = \tilde{b}_j$ and $\tilde{\sigma}(\tilde{b}_j) = \tilde{a}_i$ where $i$ need not equal $j$. This extends trivially to automorphisms that swap more than two qubits by simply swapping all the relevant fibers. 

\subsection{Lifts of CSS Code Automorphisms}
\label{sec:CSScodeAuts}

The graph theoretic notions of automorphisms that preserve a covering graph structure discussed above extend readily to the CSS code chain complex setting \cite{Berthusen2025autGadgets,Breuckmann2024foldtransversal}. 

\begin{dfn}[\textbf{Automorphism of a CSS Code}]
     Let $C_{\bullet}$ be a chain complex representing a CSS code. An automorphism of $C_{\bullet}$ is a chain map $A_{\bullet} \colon C_{\bullet} \rightarrow C_{\bullet}$
    \begin{equation}
        \begin{tikzcd}
            C_2 = \bbF_2^{n_Z} \arrow[r, "H_Z^T"] \arrow[d, "W_Z"]
            & C_1 = \bbF_2^n \arrow[d, "f"] \arrow[r, "H_X"]
            & C_0 = \bbF_2^{n_X} \arrow[d, "W_X"] \\
            C_2 = \bbF_2^{n_Z} \arrow[r, "H_Z^T"]
            & C_1 = \bbF_2^n \arrow[r, "H_X"] 
            & C_0 = \bbF_2^{n_X}
        \end{tikzcd}
    \end{equation}
    where $f \in S_n$, $W_X \in \text{GL}_{n_X}(\bbF_2)$, $W_Z \in \text{GL}_{n_Z}(\bbF_2)$.\footnote{While we can have $f \in \text{GL}_n(\bbF_2)$, this is unphysical as it would include operations that can send a qubit to a product of qubits.}
\end{dfn}
In other words, an automorphism of a CSS code is a permutation of the qubits $f \in S_n$ such that the following equations are satisfied
\begin{equation}
\begin{split}
    &W_X H_X = H_X f\\
    &H_Z^T W_Z = f H_Z^T.
\end{split}
\label{eq:css-code-aut}
\end{equation}
Code automorphisms are important because they can have non-trivial logical actions and can therefore be used to implement transversal or swap-transversal logic \cite{Sayginel2025auts,Berthusen2025autGadgets}. 

Since invertible matrices include elementary transformations that add rows and columns together, the matrices $W_X,W_Z$ can send stabilisers to products of stabilisers. A more restrictive notion of an automorphism where this does not happen is as follows.

\begin{dfn}[\textbf{Tanner Graph Automorphism}]
     Let $C_{\bullet}$ be a chain complex representing a CSS code. A Tanner graph automorphism of $C_{\bullet}$ is a chain map 
    \begin{equation}
        \begin{tikzcd}
            C_2 = \bbF_2^{n_Z} \arrow[r, "H_Z^T"] \arrow[d, "f_2"]
            & C_1 = \bbF_2^n \arrow[d, "f_1"] \arrow[r, "H_X"]
            & C_0 = \bbF_2^{n_X} \arrow[d, "f_0"] \\
            C_2 = \bbF_2^{n_Z} \arrow[r, "H_Z^T"]
            & C_1 = \bbF_2^n \arrow[r, "H_X"] 
            & C_0 = \bbF_2^{n_X}
        \end{tikzcd}
    \end{equation}
    where $f_2,f_1,f_0 \in S_n$.
\end{dfn}

Tanner graph automorphisms are therefore a subgroup of CSS code automorphisms that are restricted to sending qubits to qubits and checks to checks. 

We can now define the lift of a CSS code automorphism as follows.
\begin{dfn}[\textbf{Lift of a CSS Code Automorphism}]
    Let $h_{\bullet} \colon \tilde{C}_{\bullet} \rightarrow C_{\bullet}$ be a chain map between CSS codes and $f_{\bullet} \colon C_{\bullet} \rightarrow C_{\bullet}$ be an automorphism. We say that an automorphism $\tilde{f}_{\bullet} \colon \tilde{C}_{\bullet} \rightarrow \tilde{C}_{\bullet}$ is a lift of $f$ if and only if 
    \begin{equation}
        h_{\bullet} \circ \tilde{f}_{\bullet} = f_{\bullet} \circ h_{\bullet}.
    \end{equation}
\end{dfn}

The above definition immediately implies by Remark \ref{rmk:InducedHomMap} that
\begin{equation}
    H_1(h_{\bullet}) \circ H_1(\tilde{f}_{\bullet}) = H_1(f_{\bullet}) \circ H_1(h_{\bullet}). 
\label{eq:aut-tau-ind}
\end{equation}
This equality raises the question of whether the logical action of an automorphism is preserved under a lifting. As the logical action of an automorphism depends on the choice of a logical basis, this question is well posed only when it is possible to define a consistent basis between the two codes. While this is usually not possible in general as there is no canonical logical basis, there is one situation in our BB code context where it is. 

\begin{prop}
    Let $Q(A,B,l,m)$ and $\tilde{Q}(\tilde{A},\tilde{B},\tilde{l},\tilde{m})$ be two BB codes with parameters $[[n,k,d]]$ and $[[n_h=hn,k_h,d_h]]$ respectively that satisfy the conditions of Theorem \ref{thm:derived-Tanner-iso} so that $\tilde{Q}$ is an $h$-cover code of $Q$. Let $f_{\bullet} \colon Q_{\bullet} \rightarrow Q_{\bullet}$ and $\tilde{f}_{\bullet} \colon \tilde{Q}_{\bullet} \rightarrow \tilde{Q}_{\bullet}$ be automorphisms such that $\tau_{\bullet} \circ f_{\bullet} = \tilde{f}_{\bullet} \circ \tau_{\bullet}$ and $\tau^{\bullet} \circ f_{\bullet} = \tilde{f}_{\bullet} \circ \tau^{\bullet}$. If $h$ is odd and $k_h = k$, then the logical action of $\tilde{f}_{\bullet}$ is the same as $f_{\bullet}$. 
\end{prop}

\begin{proof}
Consider a basis $X_i, Z_j$ for a base BB code $Q$ where $[X_i, Z_j] = 0$ when $i\neq j$ and they anti-commute for all $i,j$ otherwise. By the proof of Theorem \ref{thm:hOddLogical}, $\hat{p}_1$, $\hat{\tau}_1$, and their transposes are isomorphisms when $h$ is odd and $k_h = k$. Thus the basis for $Q$ lifts to a basis $\tilde{X}_i, \tilde{Z}_j$ for a cover BB code $\tilde{Q}$ where $\tilde{X}_i = \hat{\tau}^1(X_i)$ and $\tilde{Z}_j = \hat{\tau}_1(Z_j)$.

Without loss of generality, suppose that the base automorphism acts as $\hat{\sigma}_1(Z_i) = U_i$ where $U_i \in H_1(Q)$ and that the lifted automorphism induces the action $\hat{\tilde{f}}_1(\tilde{Z}_i) = \tilde{V}_i$ where $\tilde{V}_i \in H_1(\tilde{Q})$. Then
\begin{equation*}
     (\hat{\tilde{f}}_1 \circ \hat{\tau}_1) (Z_i) =  \hat{\tilde{f}}_1(\tilde{Z}_i) = \tilde{V}_i
\end{equation*}
and
\begin{equation*}
    (\hat{\tau}_1 \circ \hat{f}_1) (Z_i) = \hat{\tau}_1(U_i) = \tilde{U}_i. 
\end{equation*}

By the definition of a lift of a CSS code automorphism, $\tilde{V}_i = \tilde{U}_i$. An analogous argument holds for the $X_i$ and $\tilde{X}_i$ operators.  
\end{proof}

In the qudit BB code case, we instead only need that $k_h = k$ and $h \neq 0 \text{ mod } p$. 

In short, if the conditions of the preceding proposition are met and a base automorphism sends the base logical $Z_i$ operator to a different logical operator $U_i$, then the lifted automorphism will send the corresponding cover logical operator $\tilde{Z}_i$ to a logical operator $\tilde{V}_i$ that is necessarily the lift of $U_i$. 

\begin{ex}
\label{ex:TG-aut}
In Example \ref{ex:logical-lift-2}, we showed that the $[[54,8,6]]$ BB code is a triple cover of the $[[18,8,2]]$ BB code. We use the base and lifted logical basis given there in what follows.

Define the matrix
\begin{equation}
        f = \begin{pmatrix}
        0 & 1 & 0 & 0 & 0 & 0 & 0 & 0 & 0 \\
        1 & 0 & 0 & 0 & 0 & 0 & 0 & 0 & 0 \\
        0 & 0 & 1 & 0 & 0 & 0 & 0 & 0 & 0 \\
        0 & 0 & 0 & 0 & 1 & 0 & 0 & 0 & 0 \\
        0 & 0 & 0 & 1 & 0 & 0 & 0 & 0 & 0 \\
        0 & 0 & 0 & 0 & 0 & 1 & 0 & 0 & 0 \\
        0 & 0 & 0 & 0 & 0 & 0 & 0 & 1 & 0 \\
        0 & 0 & 0 & 0 & 0 & 0 & 1 & 0 & 0 \\
        0 & 0 & 0 & 0 & 0 & 0 & 0 & 0 & 1 \\
    \end{pmatrix}.
\end{equation}

We define a Tanner graph automorphism $f_{\bullet}$ of the base code by letting $f_2=f_0 = f$ and
\begin{equation}
    f_1 = \begin{pmatrix}
        f & 0 \\
        0 & f
    \end{pmatrix}.
\end{equation}
The projection chain map $p_{\bullet}$ consists of the vertical maps $p_2,p_1,p_0$ with matrix representations 
\begin{equation}
    p_2 = p_0 = p = \begin{pmatrix}
        I_{lm} & I_{lm} & I_{lm}
    \end{pmatrix},
\end{equation}
\begin{equation}
    p_1 = \begin{pmatrix}
        p & 0 \\
        0 & p 
    \end{pmatrix}.
\end{equation}
In general, the matrix $p$ will be $h$ blocks of $I_{lm}$ stacked horizontally for an $h$-cover. Let 
\begin{equation}
    \tilde{f} = \begin{pmatrix}
        f & 0 & 0 \\
        0 & f & 0 \\
        0 & 0 & f
    \end{pmatrix}.
\end{equation}
It can be shown that the lifted automorphism $\tilde{f}_\bullet$ is defined as $\tilde{f}_2 = \tilde{f}_0 = \tilde{f}$ and
\begin{equation}
    \tilde{f}_1 = \begin{pmatrix}
        \tilde{f} & 0 \\
        0 & \tilde{f}
    \end{pmatrix}.
\end{equation}
It is straightforward to verify that this lifted automorphism corresponds to swapping entire fibers of the covering map. 

Taking $G_X$ and $G_Z$ as the matrices whose rows are respectively the basis of $X$ and $Z$ logical operators given in Equation \ref{eq:n18-log-basis}, we find that the logical action of the automorphism is to transform the bases as $G'_X = G_Xf^1$ and $G'_Z = G_Z f_1$, where $f^1 = f_1$. We can then determine how the logical operators have changed in order to find the logical action. In this example, the logical action of $f_{\bullet}$ is as $$\text{CNOT}(1,2)\text{CNOT}(3,4)\text{CNOT}(6,5)\text{CNOT}(8,7).$$ Applying the lifted automorphism to the lifted basis for the cover code we find that, as expected, the logical action in the cover code is identical. \hfill $\blacksquare$
\end{ex}

There are also situations where we want an isomorphism between a CSS code and its dual. These correspond to the $ZX$-dualities in \cite{Breuckmann2024foldtransversal} as they swap the $X$ and $Z$ sectors of a CSS code and therefore implement logical Hadamards. We specialize to BB codes as follows (this extends to group algebra codes in general \cite{Eberhardt2025foldBB}):

\begin{dfn}[\textbf{ZX-Duality of a BB Code}]
    \label{dfn:ZXduality}
     Let $Q_{\bullet}$ be a chain complex representing a BB code $Q$. A $ZX$-duality of $Q_{\bullet}$ is a chain map $A_{\bullet} \colon Q_{\bullet} \rightarrow Q_{\bullet}^T$
    \begin{equation}
        \begin{tikzcd}
            Q_2 = \bbF_2^{lm} \arrow[r, "H_Z^T"] \arrow[d, "W_Z"]
            & Q_1 = \bbF_2^{lm} \oplus  \bbF_2^{lm} \arrow[d, "f"] \arrow[r, "H_X"]
            & Q_0 = \bbF_2^{lm} \arrow[d, "W_X"] \\
            Q_0 = \bbF_2^{lm} \arrow[r, "H_X^T"]
            & Q_1 = \bbF_2^{lm} \oplus  \bbF_2^{lm} \arrow[r, "H_Z"] 
            & Q_2 = \bbF_2^{lm}
        \end{tikzcd}
    \end{equation}
    where $f \in S_{2lm}$, $W_X \in \text{GL}_{lm}(\bbF_2)$, $W_Z \in \text{GL}_{lm}(\bbF_2)$.
\end{dfn}
\newpage
As the subsequent example shows, we can also lift $ZX$-dualities from a base BB code to a cover BB code.  

\begin{ex}
\label{ex:zx-duality}
Using the same codes and logical bases as in Example \ref{ex:TG-aut}, let $I_{lm}$ be the $lm \times lm$ identity matrix and consider the base CSS code automorphism defined by $f_2 = f_0 = I_{lm}$ and
\begin{equation}
    f_1 = \begin{pmatrix}
        0 & I_{lm} \\
        I_{lm} & 0 
    \end{pmatrix}.
\end{equation}
Letting
\begin{equation}
    \tilde{f} = \begin{pmatrix}
        I_{lm} & 0 & 0 \\
        0 & I_{lm} & 0 \\
        0 & 0 & I_{lm}
    \end{pmatrix} = I_{\tilde{l}\tilde{m}},
\end{equation}
we find that the lifted automorphism consists of the maps $\tilde{f_2} = \tilde{f}_0 = \tilde{f} = I_{\tilde{l}\tilde{m}}$ and,
\begin{equation}
    \tilde{f}_1 = \begin{pmatrix}
        0 & \tilde{f} \\
        \tilde{f} & 0 
    \end{pmatrix} = \begin{pmatrix}
        0 & I_{\tilde{l}\tilde{m}} \\
        I_{\tilde{l}\tilde{m}} & 0 
    \end{pmatrix}.
\end{equation}
\begin{equation}
    \tilde{f}_1 = \begin{pmatrix}
        0 & I_{\tilde{l}\tilde{m}} \\
        I_{\tilde{l}\tilde{m}} & 0 
    \end{pmatrix}.
\end{equation}
The logical action of the base and lifted automorphisms is identical and is given by the logical circuit in Figure \ref{fig:aut-logical-circ}. This is precisely a SWAP-transversal Hadamard gate as in \cite{Sayginel2025auts}. \hfill $\blacksquare$
\end{ex}

\begin{figure}
    \centering
    \includegraphics[width=0.20\linewidth]{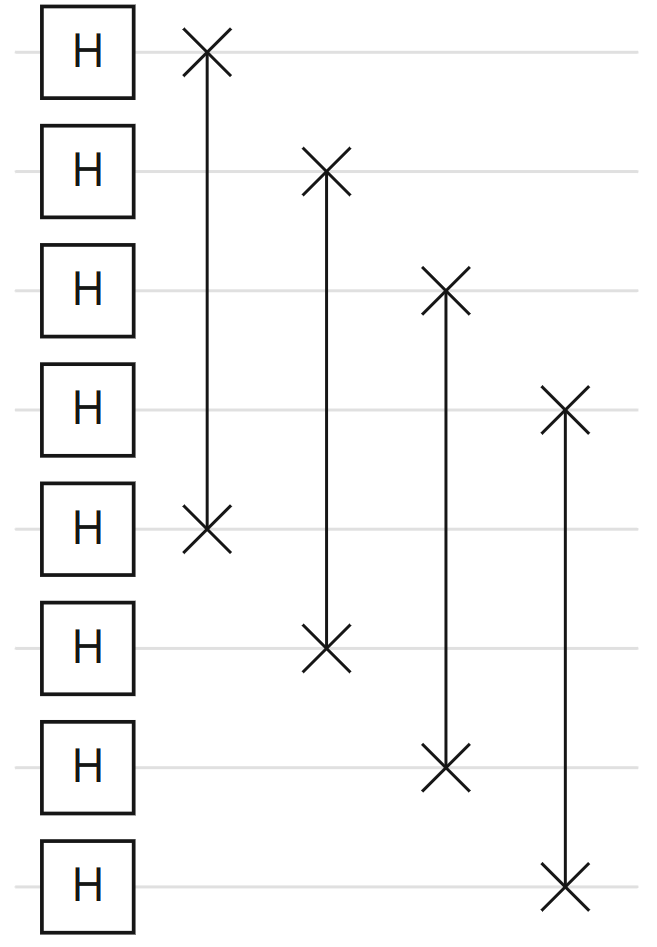}
    \caption{A circuit showing the logical action of the base and lifted automorphisms in Example \ref{ex:zx-duality}.}
    \label{fig:aut-logical-circ}
\end{figure}

\section{Searching for Cover Codes}
\label{sec:CoverCodeNumerics}

Theorem \ref{thm:code-seqs} allows us to start with a small base BB code and search for new codes by taking integer multiples of the base code's lattice parameters $l, m$ and check polynomials with monomials whose powers are congruent to those of the base code's modulo $l$ or $m$. All possible check polynomials for the cover code are of the form given in \eqref{eq:cover-poly-explicit}.

We perform a search for weight 6 and weight 8 BB codes in Sections \ref{sec:weight-6-codes} and \ref{sec:weight-8-codes} respectively. We give more details below but, we find some novel weight 8 codes that have promising parameters. In particular, the $[[64,14,8]]$, $[[144,14,14]]$, $[[192,20,\leq 16]]$ and $[[288,14,\leq 24]]$ codes are some of the best instances we find.

The search space of cover codes is much smaller than the full search space of 6 term polynomials. For a given $\tilde{l}$ and $\tilde{m}$, there are $\tilde{l}\tilde{m} = n_h/2$ possible choices of monomials. As there are 6 terms in the check polynomials, the full search space consists of $(n_h/2)^6$ polynomials. However, as discussed in Section \ref{sec:Deck}, there is an $\tilde{l}\tilde{m}$-fold degeneracy in the choice of both $\tilde{A}$ and $\tilde{B}$, as multiplying either polynomial by any monomial will give an equivalent code. This reduces the search space to $N=(n_h/2)^4$.  

For a $h$-cover code though, as there are $ut=h$ choices of each $\tilde{A}_i, \tilde{B}_j$ that satisfy the conditions in Theorem \ref{thm:code-seqs}, the full search space is $h^6$. As before, there is at least an $h$-fold degeneracy for each $\tilde{A}$ and $\tilde{B}$ due to the fact that each polynomial can independently be multiplied by monomials $x^{il}y^{jm}$ with $0 \leq i < u$ and $0 \leq j < t$ such that the resulting code has identical parameters and still remains a cover code. There are also certain instances where multiplication by other monomials, e.g. $\tilde{A}' = x^ay^b\tilde{A}$ where $a$ and $b$ are not multiples of $l$ and $m$, still produces a valid cover code satisfying  $p(\tilde{A}') = A$. Since whether or not such monomials exist depends on the specific code, we have an upper bound of $N'\leq h^4$ on the number of unique codes in the search space for $h$-cover codes. This is considerably smaller than $N = (n_h/2)^4$. For example, when considering the gross code as a double cover of the $[[72,12,6]]$ code, we find that $N = 72^4$ while $N' = 2^4=16$. We show the sizes of the full and reduced search spaces for the code examples in Tables \ref{tab:k-12-codes}-\ref{tab:k-6-codes} and Tables \ref{tab:k-14-w-8-codes}-\ref{tab:k-8-w-8-codes}, and we find reductions of up to a factor of $10^{9}$ in the size of the search space.

The search space of cover codes is not only significantly smaller than the full search space but also appears to contain interesting instances of BB codes. By interesting, we mean codes that have a non-trivial number of logical qubits and have distance that appears to scale with $h$. 

In Table \ref{tab:numbers-of-covers}, we show $k_h$ for all unique instances of cover codes up to $h=5$ for the $[[72,12,6]]$, $[[18,8,2]]$ and $[[14,6,2]]$ base codes (see Tables \ref{tab:k-12-codes}-\ref{tab:k-6-codes} for polynomials). By unique codes, we mean all codes with polynomials that are inequivalent as described above (recall that check polynomials related by monomial multiplication are equivalent). The data in Table \ref{tab:numbers-of-covers} were gathered by brute-force searching over all unique instances of valid cover codes for a given base code with $\tilde{l},\tilde{m}$ chosen to match the instances in Tables \ref{tab:k-12-codes}-\ref{tab:k-6-codes}. If the number of logical qubits of the cover code increases, we observe that it is often, but not always, because the cover is trivial in the sense that the cover consists of multiple disconnected Tanner graphs i.e. multiple separate, smaller codes. The disconnected Tanner graphs may or may not be copies of the base Tanner graph. In Tables \ref{tab:increasing-k-codes-wt-6} and \ref{tab:increasing-k-codes-wt-8} we show examples of non-trivial cover codes where $k_h > k$. 

\subsection{Weight 6 Code Examples}
\label{sec:weight-6-codes}

In Tables \ref{tab:k-12-codes}, \ref{tab:k-10-codes}, \ref{tab:k-8-codes}, and \ref{tab:k-6-codes}, we show sequences of increasingly large cover codes with fixed $k=12, 10, 8$ and $6$, respectively. In each of the sequences, the $h$-cover codes are covers of the base ($h=1$) code which is the bottom entry in the table. For codes with distances up to and including $d=14$, exact distances were found using the mixed integer programming method of \cite{Bravyi2024BB}. For codes with larger distances, upper bounds were computed using BP-OSD \cite{Roffe2020Decode,Roffe_LDPC_Python_tools_2022} with 16,000 shots. Codes from the original BB codes paper \cite{Bravyi2024BB} are highlighted in bold. Note that while we chose to fix $k_h=k$ in these sequences, there are non-trivial instances of cover codes where $k_h > k$. In particular, the $[[72,12,6]]$ code used as the base code for the $k=12$ sequence is respectively a double and 4-cover of the $[[36,8,4]]$ and $[[18,8,2]]$ codes shown in Table \ref{tab:k-8-codes}. Another interesting example of increasing $k$ is for the $[[756,16,\leq 34]]$ code shown in \cite{Bravyi2024BB}, which can be viewed as a 6-cover code of the $[[126, 4, 12]]$ code $Q(x^3 + y + y^2, \ y^2 + x^3 + x^{19} , \  21, \ 3)$.
While we show codes with specific choices of lattice parameters $\tilde{l}$ and $\tilde{m}$, these are generally not unique and may not be the only lattice parameters that would lead to such codes. 

The sequences of codes shown in Tables \ref{tab:k-12-codes}-\ref{tab:k-6-codes}, appear to have increasing distance. This further supports our claim that the reduced search space of polynomials that corresponds to cover codes contains interesting code instances. The $k=12$ sequence appears to align with the conjectured family in Yoder et al \cite{Yoder2025TdG}. In particular, if the $n=432$ and $n=648$ codes have distances of 24 and 30 respectively, then this would match the conjectured distances. 

We believe that several of the codes in Tables \ref{tab:k-12-codes}-\ref{tab:k-6-codes} are new, particularly the larger instances of the $k=12$ codes. The $k=8$ codes appear to match the $[[18h,8,\leq 2h]]$ balanced product cyclic codes of \cite{Tiew2025Dehn}. In principle, we are also able to recover all of the generalised toric codes in \cite{Liang2025gtc}. As an example, we show in Table \ref{tab:k-10-codes} that several of the $k=10$ codes can be generated as a sequence of covers of the base $[[62,10,6]]$ code. The remaining $k=10$ codes from \cite{Liang2025gtc} can be generated as covers of the $[[42,10,4]]$ code $Q(1+x^2 + x^{13},\ 1+x^{16}+x^{20},\ 21,\ 1)$. We are not restricted to only reproducing the reported codes though as we can also extend the sequence by increasing $h$. We find, for example, the $[[434, 10, d \leq 26]]$ 7-cover code in Table \ref{tab:k-10-codes} that is larger than the $[[372,10,\leq 24]]$ code reported in \cite{Liang2025gtc}.

\vspace{0.5em}
\begin{table}[H]
\centering
\caption*{$\boldsymbol{k=12}$ \textbf{Codes}}
\begin{tabular}{c|c|c|c|c|c|c|c|c}
    Parameters & $l$ & $m$ & $h$ & $A$ & $B$ & $N$ & $N'$ & $kd^2/n$\\
    \hline & & & & & & & & \\[-1em]
    $[[648,12,d \leq 36]]$ & 18 & 18 & 9 & $x^3 + y^{13} + x^{12}y^2$ & $y^3 + x^{7}y^{12} + x^{14}y^6$ & $1.1\times 10^{10}$& 6561 & $\leq 24$\\
    \hline & & & & & & & & \\[-1em]
    $[[576,12, d \leq 32]]$ & 24 & 12 & 8 & $x^{21} + y^7 + y^2$ & $y^9 + x^{13} + x^{14}$ & $6.9\times 10^9$ & 4096 & $\leq 21.3$\\
    \hline & & & & & & & & \\[-1em]
    $[[504,12,d \leq 30]]$ & 42 & 6 & 7 & $x^{27} + y + y^2$ & $y^3 + x^{37} + x^{38}$ & $4.0\times 10^9$ & 2401 & $\leq 21.4$\\
    \hline & & & & & & & & \\[-1em]
    $[[432,12,d \leq 24]]$ & 18 & 12 & 6 & $x^{15} + y^7 + y^2$ & $y^9 + x^7 + x^{14}$ & $2.2\times 10^9$& 1296 & $\leq 16$\\
    \hline & & & & & & & & \\[-1em]
    $\mathbf{[[360,12,24]]}$ & \textbf{30} & \textbf{6} & 5 & $\boldsymbol{x^9 + y + y^2}$ & $\boldsymbol{y^3 + x^{25} + x^{26}}$ & $1.0\times 10^9$ & 625 & 19.2\\
    \hline & & & & & & & & \\[-1em]
    $\mathbf{[[288,12,18]]}$ & \textbf{12} & \textbf{12} & 4 & $\boldsymbol{x^3 + y^7 + y^2}$ & $\boldsymbol{y^3 + x + x^2}$ & $4.3\times 10^8$& 256 & 13.5\\
    \hline & & & & & & & & \\[-1em]
    $[[216,12,12]]$ & 18 & 6 & 3 & $x^3 + y + y^2$ & $y^3 + x + x^2$ & $1.4\times 10^8$& 81 & 8\\
    \hline & & & & & & & & \\[-1em]
    $\mathbf{[[144,12,12]]}$ & \textbf{12} & \textbf{6}& 2 & $\boldsymbol{x^3 + y + y^2}$ & $\boldsymbol{y^3 + x + x^2}$ & $2.7\times 10^7$ & 16 & 12\\
    \hline & & & & & & & & \\[-1em]
    $\mathbf{[[72,12,6]]}$ & \textbf{6} & \textbf{6} & 1 & $\boldsymbol{x^3 + y + y^2}$ & $\boldsymbol{y^3 + x + x^2}$ & 1 & 1 & 6 \\
\end{tabular}
\vspace{0.25em}
\caption{Sequence of $h$-cover codes of the base $[[72,12,6]]$ code with $k=12$. Codes from the original BB codes paper \cite{Bravyi2024BB} are highlighted in bold. $N$ is the size of the full search space of polynomials for the given $l,m$ while $N'$ is the reduced search space of cover code polynomials.}
\label{tab:k-12-codes}
\end{table}


\begin{table}[H]
\centering
\caption*{$\boldsymbol{k=10}$ \textbf{Codes}}
\begin{tabular}{c|c|c|c|c|c|c|c|c}
    Parameters & $l$ & $m$ & $h$ & $A$ & $B$ & $N$ & $N'$ & $kd^2/n$\\
    \hline & & & & & & & & \\[-1em]
    $[[434,10,\leq26]]$ & 31 & 7 & 7 & $1 + x^6y + x^{27}y^4$ & $1 + x^{15}y^6 + x^{24}y^3$ & $2.2\times 10^9$ & 2401 & $\leq 15.6$\\
    \hline & & & & & & & & \\[-1em]
    $[[372,10,\leq24]]$ & 31 & 6 & 6 & $1 + x^6y^2 + x^{27}$ & $1 + x^{15}y + x^{24}y^4$ & $1.2\times 10^9$ & 1296 & $\leq 15.5$\\
    \hline & & & & & & & & \\[-1em]
    $[[310,10,\leq22]]$ & 31 & 5 & 5 & $1 + x^6y^2 + x^{27}y$ & $y^2 + x^{15} + x^{24}$ & $5.8\times 10^8$ & 625 & $\leq 15.6$\\
    \hline & & & & & & & & \\[-1em]
    $[[248,10,\leq18]]$ & 31 & 4 & 4 & $1 + x^6y + x^{27}$ & $y^2 + x^{15}y^3 + x^{24}$ & $2.4\times 10^8$ & 256 & $\leq 13.1$\\
    \hline & & & & & & & & \\[-1em]
    $[[186,10,14]]$ & 31 & 3 & 3 & $y + x^6y^2 + x^{27}$ & $1 + x^{15}y + x^{24}$ & $7.5 \times 10^7$& 81 & 10.5\\
    \hline & & & & & & & & \\[-1em]
    $[[124,10,10]]$ & 31 & 2 & 2 &$y + x^6 + x^{27}$ & $1 + x^{15} + x^{24}$ & $1.7\times 10^7$& 16 & 8.1\\
    \hline & & & & & & & & \\[-1em]
    $[[62,10,6]]$ & 31 & 1 & 1 & $1 + x^6 + x^{27}$ & $1 + x^{15} + x^{24}$ & 1 & 1 & 5.8\\
\end{tabular}
\vspace{0.25em}
\caption{Sequence of $h$-cover codes of the base $[[62,10,6]]$ code with $k=10$. The codes with $h=1,2,...,6$ are some of the $k=10$ generalised toric codes \cite{Liang2025gtc}.}
\label{tab:k-10-codes}
\end{table}

\vspace{-2.5em}
\begin{table}[H]
\centering
\caption*{$\boldsymbol{k=8}$ \textbf{Codes}}
\begin{tabular}{c|c|c|c|c|c|c|c|c}
    Parameters & $l$ & $m$ & $h$ & $A$ & $B$ & $N$ & $N'$ & $kd^2/n$\\
    \hline & & & & & & & & \\[-1em]
    $[[198,8,d\leq 16]]$ & 33 & 3 & 11 & $x^{12} + y + y^2$ & $1 + x + x^8$ & $9.6\times 10^7$ & 1681 & $\leq 10.3$ \\
    \hline & & & & & & & & \\[-1em]
    $[[180,8,d\leq 16]]$ & 15 & 6 & 10 & $x^9 + y + x^{6}y^5$ & $x^6 y^3 + x + x^2$ & $6.6\times 10^7$ & 1156 & $\leq 11.4$\\
    \hline & & & & & & & & \\[-1em]
    $[[162,8, 14]]$ & 27 & 3 & 9 & $1 + y + x^6y^2$ & $1 + x^{25} + x^{20}$ & $4.3\times 10^7$ & 729 & 9.7\\
    \hline & & & & & & & & \\[-1em]
    $[[144,8,12]]$ & 24 & 3 & 8 & $1 + y + x^{21}y^2$ & $1 + x^{22} + x^{17}$ & $2.7\times 10^7$ & 484 & 8\\
    \hline & & & & & & & & \\[-1em]
    $[[126,8,10]]$ & 21 & 3 & 7 & $x^9 + y + y^2$ & $1 + x + x^8$ & $1.6\times 10^7$ & 289 & 6.3\\
    \hline & & & & & & & & \\[-1em]
    $\mathbf{[[108,8,10]]}$ & \textbf{9} & \textbf{6} & 6 & $\boldsymbol{x^3 + y + y^2}$ & $\boldsymbol{y^3 + x + x^2}$ & $8.5\times 10^6$& 168 & 7.4\\
    \hline & & & & & & & & \\[-1em]
    $\mathbf{[[90,8,10]]}$ & \textbf{15} & \textbf{3} & 5 & $\boldsymbol{x^9 + y + y^2}$ & $\boldsymbol{1 + x^7 + x^2}$ & $4.1\times 10^6$ & 81 & 8.9\\
    \hline & & & & & & & & \\[-1em]
    $[[72,8,8]]$ & 12 & 3 & 4 & $x^9 + y + y^2$ & $1 + x^4 + x^{11}$ & $1.7\times 10^6$ & 36 & 7.1\\
    \hline & & & & & & & & \\[-1em]
    $[[54,8,6]]$ & 9 & 3 & 3 & $x^3 + y + y^2$ & $1 + x + x^2$ & $5.3 \times 10^5$& 15 & 5.3\\
    \hline & & & & & & & & \\[-1em]
    $[[36,8,4]]$ & 6 & 3 & 2 &$x^3 + y + y^2$ & $1 + x + x^2$ & $1.0\times 10^5$& 4 & 3.6\\
    \hline & & & & & & & & \\[-1em]
    $[[18,8,2]]$ & 3 & 3 & 1 & $1 + y + y^2$ & $1 + x + x^2$ & 1 & 1 & 1.8\\
\end{tabular}
\vspace{0.25em}
\caption{Sequence of $h$-cover codes of the base $[[18,8,2]]$ code with $k=8$. Codes from the original BB codes paper \cite{Bravyi2024BB} are highlighted in bold.}
\label{tab:k-8-codes}
\end{table}


\vspace{-2.5em}
\begin{table}[H]
\centering
\caption*{$\boldsymbol{k=6}$ \textbf{Codes}}
\begin{tabular}{c|c|c|c|c|c|c|c|c}
    Parameters & $l$ & $m$ & $h$ & $A$ & $B$ & $N$ & $N'$ & $kd^2/n$\\
    \hline & & & & & & & & \\[-1em]
    $[[154,6,\leq 16]]$ & 7 & 11 & 11 & $1 + x^2 y^3 +x^3y^4$  & $1 + x^2 y^8 + x^3 y^7$ & $3.5\times 10^7$ & $1.5 \times 10^4$ & $\leq 10.0$\\
    \hline & & & & & & & & \\[-1em]
    $[[140,6,14]]$ & 7 & 10 & 10 & $1 + x^2 y^5 +x^3y^9$  & $1 + x^2y^6 + x^3 y^3$ & $2.4\times 10^7$ & $1.0 \times 10^4$ & 8.4\\
    \hline & & & & & & & & \\[-1em]
    $[[126,6,14]]$ & 7 & 9 & 9 & $1 + x^2 y^5 +x^3y$  & $1 + x^2 + x^3 y^2$ & $1.6\times 10^7$ & 6561 & 9.3\\
    \hline & & & & & & & & \\[-1em]
    $[[112,6,12]]$ & 7 & 8 & 8 & $1 + x^2 y^5 +x^3y$ & $1 + x^2 y^6 + x^3 y^5$ & $9.8\times 10^6$ & 4096 & 7.7\\
    \hline & & & & & & & & \\[-1em]
    $[[98,6,12]]$ & 7 & 7 & 7 & $1 + x^2 y^2 +x^3y$  & $1 + x^2 + x^3 y^2$ & $5.8\times 10^6$ & 2401 & 8.8\\
    \hline & & & & & & & & \\[-1em]
    $[[84,6,10]]$ & 7 & 6 & 6 & $1 + x^2 y^3 +x^3y^2$  & $1 + x^2 + x^3 y$ & $3.1\times 10^6$ & 1296 & 7.1\\
    \hline & & & & & & & & \\[-1em]
    $[[70,6,8]]$ & 7 & 5 & 5 & $y + x^2 y^4 +x^3y$  & $y^4 + x^2 + x^3$ & $1.5\times 10^6$ & 625 & 5.5\\
    \hline & & & & & & & & \\[-1em]
    $[[56,6,8]]$ & 7 & 4 & 4 &$1 + x^2 + x^3y^2$ & $1 + x^2y + x^3$ & $6.1\times 10^5$ & 256 & 6.9\\
    \hline & & & & & & & & \\[-1em]
    $[[42,6,6]]$ & 7 & 3 & 3 &$1 + x^2 + x^3y$ & $1 + x^2 + x^3y^2$ & $1.9\times 10^5$ & 81 & 5.1\\
    \hline & & & & & & & & \\[-1em]
    $[[28,6,4]]$ & 7 & 2 & 2 &$y + x^2 + x^3$ & $1 + x^2 + x^3$ & $3.8\times 10^4$ & 16 & 3.4\\
    \hline & & & & & & & & \\[-1em]
    $[[14,6,2]]$ & 7 & 1 & 1 & $1 + x^2 + x^3$ & $1 + x^2 + x^3$ & 1 & 1& 1.7\\
\end{tabular}
\vspace{0.25em}
\caption{Sequence of $h$-cover codes of the base $[[14,6,2]]$ code with $k=6$.}
\label{tab:k-6-codes}
\end{table}

It is also possible to choose a sequence of covers where $k_h$ is increasing rather than fixed. In Table \ref{tab:increasing-k-codes-wt-6}, we show a sequence of $h$-covers of the $[[36,8,4]]$ code that all have $k_h > 8$. We only show values of $h$ in Table \ref{tab:increasing-k-codes-wt-6} where $k_h > k$. These occur somewhat sparsely and it is unclear if this is a general feature or specific to the choice of base code. Comparing the codes in Table \ref{tab:increasing-k-codes-wt-6} to those in Tables \ref{tab:k-12-codes}-\ref{tab:k-6-codes}, we see that it is generally better to choose multiple copies of a lower $k$ code from a sequence with fixed $k$. All codes in Table \ref{tab:increasing-k-codes-wt-6} have parameters that align with our conjectured bounds, and we saw no codes that violated those bounds when searching.

\begin{table}[H]
\centering
\caption*{$\boldsymbol{k_h > k}$ \textbf{Codes}}
\begin{tabular}{c|c|c|c|c|c|c}
    Parameters & $l$ & $m$ & $h$ & $A$ & $B$ & $kd^2/n$\\
    \hline & & & & & & \\[-1em]
    $[[288,20,6]]$ & 12 & 12 & 8 & $x^{9}y^{9} + x^{6}y^{7} + y^{8}$ & $1 + x^{7}y^{6} + x^{2}y^{9}$ & 2.5 \\
    \hline & & & & & & \\[-1em]
    $[[288,16, 12]]$ & 12 & 12 & 8 & $x^3y^3 + x^{6}y^{7} + y^8$  & $x^6 + xy^9 + x^2y^{9}$ & $8$ \\
    \hline & & & & & & \\[-1em]
    \hline & & & & & & \\[-1em]
    $[[252,20,4]]$ & 6 & 21 & 7 & $x^3y^{18} + y^{4} + y^{11}$  & $y^{18} + xy^{15} + x^2y^{9}$ & 1.27 \\
    \hline & & & & & & \\[-1em]
    $[[252,14,12]]$ & 6 & 21 & 7 & $x^3 + y^{19} + y^8$  & $y^9 + xy^6 + x^2y^{15}$ & 8 \\
    \hline & & & & & & \\[-1em]
    \hline & & & & & & \\[-1em]
    $[[72,12,6]]$ & 6 & 6 & 2 & $x^3 + y + y^2$  & $y^3 + x + x^2$ & 6 \\
    \hline & & & & & & \\[-1em]
    \hline & & & & & & \\[-1em]
    $[[36,8,4]]$ & 6 & 3 & 1 & $x^3 + y + y^2$ & $1 + x + x^2$ & 3.6 \\
\end{tabular}
\vspace{0.25em}
\caption{Sequence of $h$-cover codes of the base $[[36,8,4]]$ code with increasing $k_h$. We only show codes in the sequence of $h$-covers where  $k_h > 8$.}
\label{tab:increasing-k-codes-wt-6}
\end{table}

\begin{table}[H]
\centering
\begin{tabular}{c|c|c|c|c|c}
    Base Code & $h$ & $\tilde{l}$ & $\tilde{m}$ & $k_h$ & \# Unique Instances \\
    \hline
    $[[18,8,2]]$ & 2 & 6 & 3 & 8 & 3 \\
    $(l=3,m=3)$&  &  &  & 16 & 1 \\
    \cdashline{2-6}
    & 3 & 9 & 3 & 8 & 15 \\
    \cdashline{2-6}
    & 4 & 12 & 3 & 8 & 32 \\
    &  &  &  & 16 & 3 \\
    &  &  &  & 32 & 1 \\
    \cdashline{2-6}
    & 5 & 15 & 3 & 8 & 80 \\
    &  &  &  & 40 & 1 \\ 
    \hline
    $[[72,12,6]]$ & 2 & 12 & 6 & 12 & 16  \\
    \cdashline{2-6}
    $(l=6,m=6)$& 3 & 18 & 6 & 12 & 81 \\
    \cdashline{2-6}
    & 4 & 12 & 12 & 12 & 192 \\
    & & & & 16 & 48 \\
    & & & & 20 & 16 \\
    \cdashline{2-6}
    & 5 & 30 & 6 & 12 & 624 \\
    & & & & 60 & 1 \\
    \hline
    $[[14,6,2]]$ & 2 & 7 & 2 & 6 & 15 \\
    $(l=7,m=1)$ & & & & 12 & 1 \\
    \cdashline{2-6}
    & 3 & 7 & 3 & 6 & 76 \\
    & & & & 10 & 4 \\
    & & & & 18 & 1 \\
    \cdashline{2-6}
    & 4 & 7 & 4 & 6 & 240 \\
    & & & & 12 & 15 \\
    & & & & 24 & 1 \\
    \cdashline{2-6}
    & 5 & 7 & 5 & 6 & 624 \\
    & & & & 30 & 1 \\
\end{tabular}
\caption{Distribution of $k_h$ for all unique instances of $h=2$ to $h=5$ cover codes of 3 different base codes.}
\label{tab:numbers-of-covers}
\end{table}

\subsection{Weight 8 Code Examples}
\label{sec:weight-8-codes}

Until now, we have only considered weight 6 BB codes, i.e. those with 3 terms in each of the defining polynomials. However, our methods extend straightforwardly to BB codes with any number of terms in the defining polynomials. The larger the number of terms in the defining polynomials, the larger the full search space of codes. This makes the reduction in search space by considering coverings even more useful. 

We applied our methods to search for BB codes with weight 8 checks where both defining polynomials $A$ and $B$ have 4 terms. We find that for these weight 8 BB codes, the full search space is $N = (n_h/2)^6$ while the reduced search space of cover codes is upper bounded as $N' \leq h^6$. In Tables \ref{tab:k-14-w-8-codes}-\ref{tab:k-8-w-8-codes}, we show sequences of $k=8,10,12,14$ weight 8 BB codes. While we believe most of these codes are new, the $[[96,10,12]]$ and $[[96,12,10]]$ codes were found previously by Lin and Pryadko \cite{LinPryadko2024Q2BGA}. Our weight 8 codes represent further empirical evidence that our conjectured bounds hold and also demonstrate the utility of our method in searching for new codes with promising parameters.

It is expected that increasing the check weight will improve the parameters of the codes and this is indeed what we observe in Tables \ref{tab:k-14-w-8-codes}-\ref{tab:k-8-w-8-codes}. We find two highly promising sequences of $k=14$ codes in Tables \ref{tab:k-14-w-8-codes} and \ref{tab:k-14-w-8-codes-alt} that have among the best $kd^2/n$ of any of the codes that we found. For example, the $[[64,14,8]]$ and $[[144,14,14]]$ codes have $kd^2/n$ of 14 and 19.1 respectively.

The weight 8 Abelian and non-Abelian two-block group algebra codes found by Lin and Pryadko in \cite{LinPryadko2024Q2BGA} are some of the more closely comparable codes to our weight 8 codes in Tables \ref{tab:k-14-w-8-codes}-\ref{tab:k-8-w-8-codes}. The weight 8 codes we find are as good and often better in terms of $kd^2/n$ than those found in \cite{LinPryadko2024Q2BGA}. In particular, the $[[64,14,8]]$ has better $kd^2/n$ than all of their codes apart from their $[[96,10,12]]$ code and most of our codes with $n>100$ have better $kd^2/n$.

All of the weight 8 codes in Tables \ref{tab:k-14-w-8-codes}-\ref{tab:k-8-w-8-codes} have improved parameters relative to their weight 6 counterparts in Tables \ref{tab:k-12-codes}-\ref{tab:k-6-codes}. The $k=8$ and weight 8 codes in Table \ref{tab:k-8-w-8-codes} have better $kd^2/n$ than the $k=8$ and weight 6 codes in Table \ref{tab:k-8-codes}, and some of these codes even outperform some of the weight 6, $k=12$ codes. For example, the $[[56,8,8]]$ weight 8 code has $kd^2/n = 9.1$ compared to the $[[72,12,6]]$ weight 6 code with $kd^2/n = 6$. Similarly, the weight 8, $k=12$ codes in Table \ref{tab:k-12-w-8-codes} have better parameters than the weight 6 codes in Table \ref{tab:k-12-codes}; the $[[32,12,4]]$ weight 8 code achieves the same $kd^2/n$ as the $[[72,12,6]]$ weight 6 code, and the $[[96,12,10]]$ code has slightly better $kd^2/n$ than the gross code with 48 fewer physical qubits. The $k=10$ codes in Table \ref{tab:k-10-w-8-codes} also outperform the $k=10$ generalised toric codes in \cite{Liang2025gtc}; the $[[48,10,6]]$ achieves distance 6 with just $n=48$ as opposed to $n=62$ for the $[[62,10,6]]$ generalised toric code.


\begin{table}[H]
\centering
\caption*{$\boldsymbol{k=14}$ \textbf{Codes}}
\begin{tabular}{c|c|c|c|c|c|c|c}
    Parameters & $l$ & $m$ & $h$ & Polynomials & $N$ & $N'$ & $kd^2/n$\\
    \hline & & & & & & & \\[-1em]
    $[[288,14,\leq 24]]$ & 12 & 12 & 4 & \multicolumn{1}{l|}{$A = x^{6}y^4 + x^{11}y^4 + x^{3}y^{6} +x^{11}y^3$} & $8.9\times 10^{12}$ & 4096 & $\leq 28$\\
    & & & & \multicolumn{1}{l|}{$B = y^{11} + x^{2}y^7 + x^{5}y^{11} +x^{9}y^4$} & & & \\
    \hline & & & & & & & \\[-1em]
    $[[216,14,\leq 20]]$ & 18 & 6 & 3 & \multicolumn{1}{l|}{$A = y^4 + x^{11}y^4 + x^{3} +x^{11}y^3$} & $1.6\times 10^{12}$ & 729 & $\leq 25.9$\\
    & & & & \multicolumn{1}{l|}{$B = x^{6}y^5 + x^{2}y + x^{5}y^5 +x^{15}y^4$} & & & \\
    \hline & & & & & & & \\[-1em]
    $[[144,14,14]]$ & 12 & 6 & 2 & \multicolumn{1}{l|}{$A = x^6y^4 + x^5y^4 + x^3 +x^{11}y^3$} & $1.4\times 10^{11}$ & 64 & 19.1\\
    & & & & \multicolumn{1}{l|}{$B = y^5 + x^8y + x^5y^5 +x^9y^4$} & & & \\
    \hline & & & & & & & \\[-1em]
    $[[72,14,8]]$ & 6 & 6 & 1 & \multicolumn{1}{l|}{$A = y^4 + x^5y^4 + x^3 +x^5y^3$} & 1 & 1& 12.4\\
    & & & & \multicolumn{1}{l|}{$B = y^5 + x^2y + x^5y^5 +x^3y^4$} & & & \\
\end{tabular}
\vspace{0.25em}
\caption{Sequence of $h$-cover codes of the base $[[72,14,8]]$ code, with weight 8 checks and $k=14$.}
\label{tab:k-14-w-8-codes}
\end{table}

\vspace{-2em}
\begin{table}[H]
\centering
\begin{tabular}{c|c|c|c|c|c|c|c}
    Parameters & $l$ & $m$ & $h$ & Polynomials & $N$ & $N'$ & $kd^2/n$\\
    \hline & & & & & & & \\[-1em]
    $[[256,14,\leq 22]]$ & 16 & 8 & 4 & \multicolumn{1}{l|}{$A = x^{9}y^3 + y^4 + x^{14} +x^{3}y^6$} & $4.4\times 10^{12}$ & 4096 & $\leq 26.5$\\
    & & & & \multicolumn{1}{l|}{$B = x^{14}y + x^4y + x^3 +x^{13}y$} & & & \\
    \hline & & & & & & & \\[-1em]
    $[[192,14,\leq 16]]$ & 8 & 12 & 3 & \multicolumn{1}{l|}{$A = xy^3 + y^8 + x^6y^4 +x^3y^2$} & $7.8\times 10^{11}$ & 729 & $\leq 18.7$\\
    & & & & \multicolumn{1}{l|}{$B = x^6y + x^4y^5 + x^3 +x^5y$} & & & \\
    \hline & & & & & & & \\[-1em]
    $[[128,14,12]]$ & 8 & 8 & 2 & \multicolumn{1}{l|}{$A = xy^3 + y^4 + x^6y^4 +x^3y^6$} & $6.9\times 10^{10}$ & 64 & 15.8\\
    & & & & \multicolumn{1}{l|}{$B = x^6y + x^4y^5 + x^3 + x^5y^5$} & & & \\
    \hline & & & & & & & \\[-1em]
    $[[64,14,8]]$ & 8 & 4 & 1 & \multicolumn{1}{l|}{$A = xy^3 + 1 + x^6 +x^3y^2$} & 1 & 1& 14\\
    & & & & \multicolumn{1}{l|}{$B = x^6y + x^4y + x^3 +x^5y$} & & & \\
\end{tabular}
\vspace{0.25em}
\caption{Sequence of $h$-cover codes of the base $[[64,14,8]]$ code, with weight 8 checks and $k=14$.}
\label{tab:k-14-w-8-codes-alt}
\end{table}

\begin{table}[H]
\centering
\caption*{$\boldsymbol{k=12}$ \textbf{Codes}}
\begin{tabular}{c|c|c|c|c|c|c|c}
    Parameters & $l$ & $m$ & $h$ & Polynomials & $N$ & $N'$ & $kd^2/n$\\
    \hline & & & & & & & \\[-1em]
    $[[160,12,\leq 16]]$ & 8 & 10 & 5 & \multicolumn{1}{l|}{$A = xy^3 + y^2 + x^6y^8 +x^3y^4$} & $2.6\times 10^{11}$ & 15625 & $\leq 19.2$\\
    & & & & \multicolumn{1}{l|}{$B = x^6y^7 + x^4y + x^3 + x^5y^3$} & & & \\
    \hline & & & & & & & \\[-1em]
    $[[128,12, 14]]$ & 8 & 8 & 4 & \multicolumn{1}{l|}{$A = xy^3 + 1 + x^6y^6 +x^3y^2$} & $6.9\times 10^{10}$ & 4096 & $\leq 18.4$\\
    & & & & \multicolumn{1}{l|}{$B = x^6y^7 + x^4y^7 + x^3 + x^5y$} & & & \\
    \hline & & & & & & & \\[-1em]
    $[[96,12,10]]$ & 8 & 6 & 3 & \multicolumn{1}{l|}{$A = xy + y^2 + x^6y^4 +x^3y^4$} & $1.2\times 10^{10}$ & 729 & 12.5\\
    & & & & \multicolumn{1}{l|}{$B = x^6y + x^4y^5 + x^3y^2 +x^5y$} & & & \\
    \hline & & & & & & & \\[-1em]
    $[[64,12,8]]$ & 8 & 4 & 2 & \multicolumn{1}{l|}{$A = xy^3 + 1 + x^6y^2 +x^3$} & $1.1\times 10^{9}$ & 64 & 12\\
    & & & & \multicolumn{1}{l|}{$B = x^6y + x^4y + x^3y^2 +x^5y$} & & & \\
    \hline & & & & & & & \\[-1em]
    $[[32,12,4]]$ & 8 & 2 & 1 & \multicolumn{1}{l|}{$A = xy + 1 + x^6 +x^3$} & 1 & 1& 6\\
    & & & & \multicolumn{1}{l|}{$B = x^6y + x^4y + x^3 +x^5y$} & & & \\
\end{tabular}
\vspace{0.25em}
\caption{Sequence of $h$-cover codes of the base $[[32,12,4]]$ code, with weight 8 checks and $k=12$.}
\label{tab:k-12-w-8-codes}
\end{table}

\vspace{-2.40em}
\begin{table}[H]
\centering
\caption*{$\boldsymbol{k=10}$ \textbf{Codes}}
\begin{tabular}{c|c|c|c|c|c|c|c}
    Parameters & $l$ & $m$ & $h$ & Polynomials & $N$ & $N'$ & $kd^2/n$\\
    \hline & & & & & & & \\[-1em]
    $[[168,10,\leq 18]]$ & 6 & 14 & 7 & \multicolumn{1}{l|}{$A = y^{12} + x^5y^8 + x^3y^4 +x^5y^5$} & $3.5\times 10^{11}$ & $1.2\times 10^5$ & $\leq 19.3$\\
    & & & & \multicolumn{1}{l|}{$B = y^7 + x^2y^7 + x^5y^5 +x^3y^6$} & & & \\
    \hline & & & & & & & \\[-1em]
    $[[144, 10, \leq 16]]$ & 6 & 12 & 6 & \multicolumn{1}{l|}{$A = 1 + x^5y^4 + x^3y^4 +x^5y^3$} & $1.4\times 10^{11}$ & 46656 & $\leq 17.8$\\
    & & & & \multicolumn{1}{l|}{$B = y^5 + x^2y^9 + x^5y^{11} +x^3y^4$} & & & \\
    \hline & & & & & & & \\[-1em]
    $[[120,10, 14]]$ & 6 & 10 & 5 & \multicolumn{1}{l|}{$A = y^6 + x^{5}+ x^3 y^2 +x^5y^3$} & $1.2\times 10^{10}$ & 15625 & $\leq 16.3$\\
    & & & & \multicolumn{1}{l|}{$B = y + x^2 y + x^5 y^9 +x^3 y^2$} & & & \\
    \hline & & & & & & & \\[-1em]
    $[[96,10,12]]$ & 12 & 4 & 4 & \multicolumn{1}{l|}{$A = y^2 + x^{11}+ x^9 +x^5y$} & $1.2\times 10^{10}$ & 4096 & 15\\
    & & & & \multicolumn{1}{l|}{$B = y^3 + x^2 y^3 + x^5 y^3 +x^3 y^2$} & & & \\
    \hline & & & & & & & \\[-1em]
    $[[72,10,8]]$ & 6 & 6 & 3 & \multicolumn{1}{l|}{$A = y^4 + x^5 + x^3 y^4 +x^5y$} & $2.2\times 10^9$ & 729 & 8.9\\
    & & & & \multicolumn{1}{l|}{$B = y^5 + x^2 y^5 + x^5 y +x^3 y^2$} & & & \\
    \hline & & & & & & & \\[-1em]
    $[[48,10,6]]$ & 6 & 4 & 2 & \multicolumn{1}{l|}{$A = y^2 + x^5 + x^3 +x^5y$} & $1.9\times 10^8$ & 64 & 7.5\\
    & & & & \multicolumn{1}{l|}{$B = y^3 + x^2 y^3 + x^5 y^3 +x^3 y^2$} & & & \\
    \hline & & & & & & & \\[-1em]
    $[[24,10,4]]$ & 6 & 2 & 1 & \multicolumn{1}{l|}{$A = 1 + x^5 + x^3 +x^5y$} & 1 & 1& 6.7\\
    & & & & \multicolumn{1}{l|}{$B = y + x^2 y + x^5 y +x^3$} & & & \\
\end{tabular}
\vspace{0.25em}
\caption{Sequence of $h$-cover codes of the base $[[24,10,4]]$ code, with weight 8 checks and $k=10$.}
\label{tab:k-10-w-8-codes}
\end{table}

\vspace{-2.40em}
\begin{table}[H]
\centering
\caption*{$\boldsymbol{k=8}$ \textbf{Codes}}
\begin{tabular}{c|c|c|c|c|c|c|c}
    Parameters & $l$ & $m$ & $h$ & Polynomials & $N$ & $N'$ & $kd^2/n$\\
    \hline & & & & & & & \\[-1em]
    $[[140,8,\leq 16]]$ & 7 & 10 & 5 & \multicolumn{1}{l|}{$A = x^4y^6 + y^5+ x^5y^5 +x^3y^9$} & $1.2\times 10^{11}$ & 15625 & $\leq 14.6$\\
    & & & & \multicolumn{1}{l|}{$B = x^5y^3 + x^3y + x^4y + y^8$} & & & \\
    \hline & & & & & & & \\[-1em]
    $[[112,8,14]]$ & 7 & 8 & 4 & \multicolumn{1}{l|}{$A = x^4y^4 + y^3 + x^5y^7 +x^3y^7$} & $1.2\times 10^{10}$ & 4096 & $\leq 14$\\
    & & & & \multicolumn{1}{l|}{$B = x^5y^7 + x^3y^5 + x^4y^7 + 1$} & & & \\
    \hline & & & & & & & \\[-1em]
    $[[84,8,10]]$ & 7 & 6 & 3 & \multicolumn{1}{l|}{$A = x^4y^4 + y^5 + x^5y^3 +x^3y^5$} & $2.2\times 10^9$ & 729 & 9.5\\
    & & & & \multicolumn{1}{l|}{$B = x^5y + x^3y^3 + x^4y^3 + y^2$} & & & \\
    \hline & & & & & & & \\[-1em]
    $[[56,8,8]]$ & 7 & 4 & 2 & \multicolumn{1}{l|}{$A = x^4y^2 + y + x^5y +x^3y^3$} & $1.9\times 10^8$ & 64 & 9.1\\
    & & & & \multicolumn{1}{l|}{$B = x^5y^3 + x^3y^3 + x^4y^3 + y^2$} & & & \\
    \hline & & & & & & & \\[-1em]
    $[[28,8,4]]$ & 7 & 2 & 1 & \multicolumn{1}{l|}{$A = x^4 + y + x^5y +x^3y$} & 1 & 1& 4.6\\
    & & & & \multicolumn{1}{l|}{$B = x^5y + x^3y + x^4y + 1$} & & & \\
\end{tabular}
\vspace{0.25em}
\caption{Sequence of $h$-cover codes of the base $[[28,8,4]]$ code, with weight 8 checks and $k=8$.}
\label{tab:k-8-w-8-codes}
\end{table}

Finally, we also consider weight 8 codes with increasing $k$. In Table \ref{tab:increasing-k-codes-wt-8} we show a sequence of covers of a base $[[32,8,4]]$ weight 8 code that have $k_h > k$. We find instances of $k_h > k$ far more frequently here than in the weight 6 case, but we again do not know if this is a general property of weight 8 codes or a consequence of the choice of base code. For several values of $h$, we find codes with different values of $k_h$ that satisfy $k_h > k$ and we do not always show all of them. Where a code is clearly worse, we tend not to include it; for example, when $h=6$ we find a $[[192,28,8]]$ code that we do not show as it has the same distance as the $[[192,30,8]]$ code but fewer logical qubits. Where we only have upper bounds on distances however, we typically opt to keep codes that have the same upper bound on $d_h$ but different numbers of logical qubits. All of the codes in Table \ref{tab:increasing-k-codes-wt-8} have parameters in line with our conjectured bounds.

Within Table \ref{tab:increasing-k-codes-wt-8}, we find what appear to be sequences of codes that have constant rate $k/n$ and constant $kd^2/n$. For example the $[[36,8,4]]$, $[[96,24,4]]$, $[[160,40,4]]$ and $[[224,56,4]]$ codes all have a rate of 0.25 and $kd^2/n = 4$. These codes all have connected Tanner graphs i.e. they are single codeblocks but their parameters are equivalent to taking multiple individual copies of the $[[32,8,4]]$ code. We also found a $[[192,24,8]]$ code that we did not include in Table \ref{tab:increasing-k-codes-wt-8} but that appears to form part of a constant rate sequence with $k/n=0.125$, along with the $[[96,12,8]]$ and $[[288,36,\leq 8]]$ codes.

Beyond the constant rate codes, we also find codes with better distances and rates of up to around $k/n=0.15$. We also see some fixed $k$ sequences within Table \ref{tab:increasing-k-codes-wt-8}. For example there is a sequence of $k=20$ codes that have promising parameters. The $[[96,20,8]]$, $[[192,20,\leq 16]]$ and $[[288,20,\leq 22]]$ codes have $kd^2/n$ of 13.3, $\leq 26.7$ and $\leq 33.6$ respectively.

\begin{table}[H]
\centering
\caption*{$\boldsymbol{k_h > k}$ \textbf{Codes}}
\begin{tabular}{c|c|c|c|c|c|c}
    Parameters & $l$ & $m$ & $h$ & $A$ & $B$ & $kd^2/n$\\
    \hline & & & & & & \\[-1em]
    $[[288,36, \leq 8]]$ & 12 & 12 & 8 & $x^{9}y^{3} + x^{8}y^{4} + x^{10} + x^{7}y^{10}$  & $x^{6}y^{9} + x^{4}y + x^{11} + x^{9}y$ & $\leq 8$\\
    \hline & & & & & & \\[-1em]
    $[[288,24, \leq 14]]$ & 12 & 12 & 8 & $x^{9}y^{11} + x^{4}y^{4} + x^{10} + x^{3}y^{10}$  & $x^{10}y^{5} + x^{8}y^{5} + x^{11}y^{8} + xy^{5}$ & $\leq 16.3$\\
    \hline & & & & & & \\[-1em]
    $[[288,20, \leq 22]]$ & 12 & 12 & 8 & $x^{5}y^{3} + x^{4}y^{8} + x^{2}y^{4} + x^{7}y^{6}$  & $x^{10}y + x^{8}y + x^{11}y^{4} + x^{5}y^{5}$ & $\leq33.6$ \\
    \hline & & & & & & \\[-1em]
    \hline & & & & & & \\[-1em]
    $[[224,56,4]]$ & 4 & 28 & 7 & $xy^{7} + y^{4} + x^{2}y^{4} + x^{3}y^{14}$  & $x^{2}y^{25} + y^{25} + x^{3} + xy^{21}$ & 4 \\
    \hline & & & & & & \\[-1em]
    $[[224,32,8]]$ & 4 & 28 & 7 & $xy^{23} + y^{16} + x^{2}y^{4} + x^{3}y^{18}$  & $x^{2}y^{9} + y^{21} + x^{3} + xy^{9}$ & 9.1 \\
    \hline & & & & & & \\[-1em]
    $[[224,26, \leq 14]]$ & 28 & 4 & 7 & $x^{13}y^{3} + x^{16} + x^{10} + x^{11}y^{2}$  & $x^{26}y + x^{16}y + x^{3} + xy$ & $\leq 22.8$ \\
    \hline & & & & & & \\[-1em]
    $[[224,20, \leq 14]]$ & 28 & 4 & 7 & $x^{5}y^{3} +  + x^{26} + x^{7}y^{2}$  & $x^{26}y + x^{16}y + x^{19} + x^{9}y$ & $\leq 17.5$ \\
    \hline & & & & & & \\[-1em]
    \hline & & & & & & \\[-1em]
    $[[192,30,8]]$ & 24 & 4 & 6 & $x^{21}y^{3} + x^{4} + x^{10} + x^{15}y^{2}$  & $x^{22}y + x^{4}y + x^{3} + x^{21}y$ & 10 \\
    \hline & & & & & & \\[-1em]
    $[[192,26, 12]]$ & 24 & 4 & 6 & $x^{9}y^{3} + x^{12} + x^{14} + x^{23}y^{2}$  & $x^{10}y + x^{20}y + x^{11} + xy$ & $19.5$ \\
    \hline & & & & & & \\[-1em]
    $[[192,20,\leq 16]]$ & 24 & 4 & 6 & $x^{9}y^{3} + x^{4} + x^{6} + x^{19}y^{2}$  & $x^{22}y + x^{8}y + x^{7} + x^{5}y$ & $\leq 26.7$ \\
    \hline & & & & & &\\[-1em]
    $[[192,18, \leq 16]]$ & 24 & 4 & 6 & $x^{21}y^3 + x^{20} + x^2 + x^{23}y^2$  & $x^{18}y + y + x^{11} + x^{17}y$ & $\leq 24$ \\
    \hline & & & & & & \\[-1em]
    \hline & & & & & & \\[-1em]
    $[[160,40,4]]$ & 20 & 4 & 5 & $x^{17}y^3 + x^{16} + x^6 +x^{7}y^2$  & $x^{10}y + y + x^{11} +xy$ & 4 \\
    \hline & & & & & & \\[-1em]
    $[[160,24,8]]$ & 20 & 4 & 5 & $xy^3 + x^{16} + x^{14} +x^{19}y^2$  & $x^6y + x^{8}y + x^3 +xy$ & 9.6 \\
    \hline & & & & & & \\[-1em]
    $[[160,16,\leq 14]]$ & 20 & 4 & 5 & $x^{17}y^3 + x^{12} + x^{10} +x^{15}y^2$  & $x^6y + x^{12}y + x^7 +xy$ & $\leq 19.6$ \\
    \hline & & & & & & \\[-1em]
    \hline & & & & & & \\[-1em]
    $[[96,24,4]]$ & 12 & 4 & 3 & $xy^3 + 1 + x^6 +x^{7}y^2$  & $x^{2}y + x^{8}y + x^3 +x^{9}y$ & 4 \\
    \hline & & & & & & \\[-1em]
    $[[96,20,8]]$ & 12 & 4 & 3 & $xy^3 + x^4 + x^2 +x^{11}y^2$  & $x^6y + x^{8}y + x^{11} +x^{9}y$ & 13.3 \\
    \hline & & & & & & \\[-1em]
    $[[96,16,4]]$ & 12 & 4 & 3 & $x^{5}y^3 + x^4 + x^{10} +x^{11}y^2$  & $x^6y + y + x^3 +x^{9}y$ & 2.7 \\
    \hline & & & & & &\\[-1em]
    $[[96,12,8]]$ & 12 & 4 & 3 & $x^{9}y^3 + 1 + x^6 +x^3y^2$  & $x^{10}y + x^{8}y + x^7 +x^{5}y$ & 8  \\
    \hline & & & & & & \\[-1em]
    \hline & & & & & & \\[-1em]
    $[[64,14,8]]$ & 8 & 4 & 2 & $xy^3 + 1 + x^6 +x^3y^2$  & $x^6y + x^4y + x^3 +x^5y$ & 14  \\
    \hline & & & & & & \\[-1em]
    \hline & & & & & & \\[-1em]
    $[[32,8,4]]$ & 4 & 4 & 1 & $xy^3 + 1 + x^2 +x^3y^2$ & $x^{2}y + y + x^3 +xy$ & 4 \\
\end{tabular}
\vspace{0.25em}
\caption{Sequence of $h$-cover codes of the base $[[32,8,4]]$ code, with weight 8 checks and $k_h > 8$. Only a selection of codes with $k_h > 8$ are shown.}
\label{tab:increasing-k-codes-wt-8}
\end{table}

\subsection{Qudit Codes}

To show that our methods can be readily applied to qudit codes, we take the $[[24,4,4]]_3$ code from Table 2 in \cite{Spencer2026quditldpc} as a base code. By using the algebraic rules from Theorem \ref{thm:derived-Tanner-iso}, we find the sequence of $h$-cover codes in Table \ref{tab:qudit-k-4-codes}, where the subscript $3$ means that the code is defined over $\mathbb{F}_3$. The codes all have weight 5 checks; they are far from exhaustive and we only provide a short sequence as a proof of principle to demonstrate that our method is not limited to qubit codes.

We also find, for example, a $[[56,4,8]]_5$ 2-cover of the $[[28,4,5]]_5$ weight 6 code from \cite{Spencer2026quditldpc}. The base code is $Q( 3x^6 + 3x^2 + 4x^3, 3x^5 + x + 1, 7,2)$ and our 2-cover is $Q( 3x^6 + 3x^2 + 4x^3, 3x^5 + x + 1, 14,2)$. The base code has a disconnected Tanner graph however, i.e. the code is in fact multiple smaller codes. This is the case for several of the weight 6 codes in Table 2 of \cite{Spencer2026quditldpc} and our 2-cover in this instance is also disconnected.

\begin{table}[H]
\centering
\begin{tabular}{c|c|c|c|c|c|c}
    Parameters & $l$ & $m$ & $h$ & $A$ & $B$ & $kd^2/n$\\
    \hline & & & & & & \\[-1em]
    $[[72,4,8]]_3$ & 12 & 3 & 3 & $x + x^6$  & $x^3 + 2y + 2y^2$ & 3.56 \\
    \hline & & & & & & \\[-1em]
    $[[48,4,7]]_3$ & 8 & 3 & 2 & $x + x^2$  & $x^3 + 2y + 2y^2$ & 4.08 \\
    \hline & & & & & & \\[-1em]
    $[[24,4,4]]_3$ & 4 & 3 & 1 & $x + x^2$ & $x^3 + 2y + 2y^2$ & 2.67 \\
\end{tabular}
\vspace{0.25em}
\caption{Sequence of $h$-cover codes of the base $[[24,4,4]]_3$ qudit BB code from \cite{Spencer2026quditldpc}.}
\label{tab:qudit-k-4-codes}
\end{table}

\section{Conclusion}

We have shown that when starting with a given BB code, it is possible to generate an infinite sequence of new BB codes via covering graphs. These so-called cover codes all have Tanner graphs that are covering graphs of the base code's Tanner graph. We present simple algebraic requirements on the defining polynomials and lattice parameters of a BB code to be a cover of a smaller BB code. In particular, the base code's lattice parameters must divide the cover code's lattice parameters and the defining polynomials of the cover code must equal the base code's polynomials modulo the base code's lattice parameters.

The cover codes reside in a considerably smaller search space than the full search space of polynomials for a given set of lattice parameters. This is only useful in so far as this search space contains interesting codes that have non-trivial $k$ and $d$. Through a combination of theoretical guarantees and numerical evidence, we have shown that this is indeed the case. For example, we have shown that the $k=12$ and $k=8$ codes from \cite{Bravyi2024BB} can be viewed as cover codes. We have also recovered the $[[18,8h,\leq 2h]]$ balanced product codes of \cite{Tiew2025Dehn} and the generalised toric codes \cite{Liang2025gtc}. 

We have found several new sequences of BB codes with weight 8 checks. Many of these codes have $kd^2/n$ considerably higher than the weight 6 codes, although this comes at the cost of increased complexity of implementation on quantum devices. The $[[64,14,8]]$, $[[144,14,14]]$, $[[192,20,\leq 16]]$ and $[[288,14,\leq24]]$ codes are some of the best instances we find. 

Beyond searching for new BB codes, we have also used the covering graph relationship to explore the connections between the logical operators of different BB codes. By extending the graph covering map to a linear map on vector spaces, we were able to define a projection (co)chain map from a cover code to a base code. The dual of this projection (co)chain map is a lifting (co)chain map. Both of these induce maps on (co)homology which can be used to project and lift logical operators. We use these chain maps to prove that $k_h \geq k$ for any $h$-cover code, and $d_h \leq hd$ when $h$ is odd. Thus the number of logical qubits cannot decrease for any cover code, and the distance cannot increase by more than a factor of $h$ when $h$ is odd. When $h$ is odd and $k_h=k$, we are able to lower bound the distance as at least that of the base code, i.e. $d_h \geq d$. We have empirically observed the same distance behaviour for codes with even $h$ and conjecture that the distance of an $h$-cover code satisfies $d \leq d_h \leq hd$ for any $h$. All the numerical evidence we have gathered to date supports this conjecture.

We also use the idea of a section of a graph cover to define a weight-preserving-lift (co)chain map. When such a (co)chain map exists, we show it can be used to lift a weight $d$ base code logical operator to a weight $d$ cover code logical operator and therefore upper bound the distance of the cover code as $d_h \leq d$. This serves primarily as a useful numerical tool to rule out cover codes where the distance does not increase relative to the base code. 

We further show how to lift automorphisms that induce non-trivial logical actions from a base code to a cover code. The maps induced in first degree homology by the lifting and projection chain maps are isomorphisms when $k_h = k$ and $h$ is odd, which implies that a logical basis of the base code will lift to a logical basis of the cover code. In this setting, where there is a notion of a consistent basis for the base and cover codes, we show that automorphisms of a base code not only lift to the cover code but that their logical action is preserved.

Our theoretical results readily extend to qudit BB codes when replacing the condition that $h$ be odd with the more general condition that $h \neq 0 \text{ mod } p$, where $p$ is the characteristic of the underlying field. We apply them to a $[[24,4,4]]_3$ base code from \cite{Spencer2026quditldpc} and get a 2-cover $[[48,4,7]]_3$ code and a 3-cover $[[72,4,8]]_3$ code. We only performed a limited search for qudit codes as a proof of principle that our method extends to qudits. We leave a more extensive search for future work.

This work raises several outstanding open questions. Our proof techniques for establishing bounds on $d$ break down when $h = 0 \text{ mod } p$. While we have numerical evidence that all $h$-cover BB codes obey the conjectured bounds $d \leq d_h \leq hd$ regardless of $h$, establishing this rigorously may require different tools from those employed here. 

The weight 8 codes we found have, perhaps unsurprisingly, better parameters than the weight 6 codes. While the improvement in parameters is due in part to the increased check weight, it is likely that the non-locality of the checks also plays a role. Recent work develops a subsystem variant of BB codes that have weight 4 checks that are less local than the surface code but are not non-local, and so are easier to implement in practice \cite{liang2026topologicalsubsystembivariatebicycle}. These codes have worse parameters than the standard BB codes but better than those of the surface code, so a compelling research direction is to develop a robust theoretical understanding of the trade-offs between the code parameters and both the weight and the non-locality of the checks in BB codes. 

The drastically reduced search space for cover codes means the remaining bottleneck to finding useful instances of BB codes, especially at larger $n$, is computing the distance. While we have made some progress in using the properties of a base code to bound the distance of a cover code, the question remains whether these tools can be extended to explicitly compute the distances of cover code instances. A recent paper by Hopkins et. al. also proves that any family of BB codes with connected Tanner graphs has distance scaling obeying the general BPT bound for dimension $2 \leq D \leq w-2$, where $w$ is the check weight \cite{hopkin2026TTIqLDPC}. If the check polynomials are fixed, then it is expected that the distance will obey the BPT bound for $D=2$. Exceeding this bound will likely require that the check polynomials need to change to ensure that range of the checks increases as $n$ increases. Since our techniques give explicit constructions for infinite sequences of BB codes with possibly different check polynomials all of fixed weight, investigating whether constructing BB cover code families that saturate these upper bounds in dimension $D \leq w -2$ is possible may give further insight into how precisely the non-locality of the checks impacts the code parameters. 

Our techniques can be readily extended to the case where the underlying polynomial ring is the group algebra of a finite Abelian group, as these are all direct products of cyclic groups. Extensions of our results to tensor product code constructions beyond two products are also possible, such as the recent trivariate tricycle codes \cite{Jacob2025trivariate}. We leave for future research the question of whether these extensions yield code families with interesting properties, such as possessing non-Clifford transversal gates, and whose parameters can be bounded through similar arguments. 

An extension of our methods to group-algebra codes based on non-Abelian groups was developed in a recent work by Aydin et. al. using a code construction method based on group actions on cosets \cite{aydin2026breakingbicycleframecosetbased}. They recover the covering code relationship presented here as a special case, though they were unable to prove the same code parameter inequalities and the results on lifting logical operators and automorphisms that we derive. We expect a similar analysis is possible for these non-Abelian cover codes. 

Lastly, there are numerous results in the decoding literature that solve decoding problems for certain codes by mapping them to simpler codes where the decoding problem is well-understood, such as for the BB code by decomposing them into copies of the toric code \cite{tan2026generalizedmatchingdecoders2d,Sahay2026matchingdecoderbivariatebicycle}. Since we have lifting and projection maps between the Tanner graphs of larger and smaller BB codes, one interesting application would be to use them to map the decoding problem for a larger BB code to a decoding problem for a smaller BB code. 

\section{Acknowledgements}

We would like to thank Virgile Guemard, Christopher Pattison, Dylan Airey, Eric Sabo, and Surya Raghavendran for many helpful discussions and for their comments on an early draft of the paper. 
D.E.B. and B.C.B.S were supported by the Engineering and Physical Sciences Research Council [grant number EP/Y004620/1]. A.R. is supported by the UKRI through the Future Leaders Fellowship Theory to Enable Practical Quantum Advantage (MR/Y015843/1). For the purposes of Open Access, the authors have applied a CC BY public copyright licence
to any Author Accepted Manuscript version arising from this submission.

\printbibliography

\end{document}